\PassOptionsToPackage{table}{xcolor}

\documentclass[11pt]{article}
\usepackage{jheppub}

\usepackage{physics}
\usepackage{mathtools}
\usepackage{amsthm}
\usepackage{dsfont}
\usepackage{todonotes}
\usepackage{dsfont}
\usepackage{faktor}
\usepackage{mathrsfs}

\usepackage{caption}
\usepackage{subcaption}

\usepackage{tikz}
\usetikzlibrary{shapes,arrows,cd,chains,decorations.markings,decorations.pathmorphing,calc,positioning,patterns}

\tikzset{
->-/.style args={#1rotate#2}{decoration={markings, mark=at position #1 with {\arrow[scale=1.5,rotate = #2 ]{stealth}}}, postaction={decorate}}
}
\tikzset{
-r-/.style args={#1rotate#2}{decoration={markings, mark=at position #1 with {\arrow[scale=1,rotate = #2 ]{>}}}, postaction={decorate}}
}

\DeclareMathOperator{\Aut}{Aut}
\DeclareMathOperator{\Hom}{Hom}
\DeclareMathOperator{\bbZ}{\mathbb{Z}}
\DeclareMathOperator{\bbR}{\mathbb{R}}
\DeclareMathOperator{\bbC}{\mathbb{C}}
\DeclareMathOperator{\Arf}{\textrm{Arf}}
\DeclareMathOperator{\TY}{\mathrm{TY}}

\DeclareMathOperator{\Vc}{\mathrm{Vec}}
\DeclareMathOperator{\Rep}{\mathrm{Rep}}
\DeclareMathOperator{\Irr}{\mathrm{Irr}}

\DeclareMathOperator{\calC}{\mathcal{C}}
\DeclareMathOperator{\calD}{\mathcal{D}}
\DeclareMathOperator{\calA}{\mathcal{A}}
\DeclareMathOperator{\calF}{\mathcal{F}}
\DeclareMathOperator{\calB}{\mathcal{B}}
\DeclareMathOperator{\calZ}{\mathcal{Z}}
\DeclareMathOperator{\calH}{\mathcal{H}}
\DeclareMathOperator{\calM}{\mathcal{M}}
\DeclareMathOperator{\calW}{\mathcal{W}}

\DeclareMathOperator{\id}{\mathds{1}}
\DeclareMathOperator{\rk}{\mathrm{rk}}

\newcommand{\Mod}[1]{ \, \mathrm{mod} \, #1}

\newtheorem{theorem}{Theorem}
\newtheorem{proposition}{Proposition}

\usepackage{xcolor}
\usepackage{booktabs, array, makecell, multirow}

\title{Duality Defects in \texorpdfstring{$E_8$}{E8}}
\abstract{We classify all non-invertible Kramers-Wannier duality defects in the $E_8$ lattice Vertex Operator Algebra (i.e. the chiral $(E_8)_1$ WZW model) coming from $\mathbb{Z}_m$ symmetries. We illustrate how these defects are systematically obtainable as $\mathbb{Z}_2$ twists of invariant sub-VOAs, compute defect partition functions for small $m$, and verify our results against other techniques. Throughout, we focus on taking a physical perspective and highlight the important moving pieces involved in the calculations. Kac's theorem for finite automorphisms of Lie algebras and contemporary results on holomorphic VOAs play a role. We also provide a perspective from the point of view of (2+1)d Topological Field Theory and provide a rigorous proof that all corresponding Tambara-Yamagami actions on holomorphic VOAs can be obtained in this manner. We include a list of directions for future studies.}

\author[1]{Ivan M. Burbano,}
\author[2,3,4]{Justin Kulp,\note[4]{Corresponding author.}}
\author[5]{Jonas Neuser}

\affiliation[1]{Department of Physics, University of California, 366 Physics North, Berkeley, 94720-7300, CA, USA}
\affiliation[2]{Perimeter Institute for Theoretical Physics, Waterloo, ON N2L 2Y5, Canada}
\affiliation[3]{Department of Physics \& Astronomy, University of Waterloo, Waterloo, ON N2L 3G1, Canada}
\affiliation[5]{Institute for Quantum Gravity, University of Erlangen-Nürnberg, Staudtstraße 7 / B2, 91058 Erlangen, Germany}

\emailAdd{ivan\_burbano@berkeley.edu}
\emailAdd{jgjkulp@gmail.com}
\emailAdd{jonas.neuser@fau.de}

\begin{document}

\maketitle

\section{Introduction}
A topological defect line in a 2d QFT corresponds to some oriented (possibly charged) 1d inhomogeneity in the system, which is invariant under continuous deformations, so long as the deformation does not move the inhomogeneity past any other operators. 
Historically, such topological defect lines have been of interest in CFT for their connection to BCFT, twisted boundary conditions, and orbifolds \cite{Cardy:1986gw, Oshikawa:1996, Oshikawa:1996_2, Petkova:2000, Fuchs:2002cm, Fuchs:2003id, Fuchs:2004dz, Fuchs:2004xi, Fjelstad:2005ua}. However, more modern applications have exploded in recent years, including: notions of generalized orbifolds and duality \cite{FR_HLICH_2010, theoMoonshine, Beigi_2011}; connections to anyon condensation \cite{Bais:2008ni, Kong:2013aya}; better understanding of defects in statistical physics systems \cite{Aasen:2016dop,aasen2020topological}; constraints for RG flows \cite{Chang_2019, thorngren2019fusion, gaiotto2020orbifold, KikuchiRCFT}; constraints for 2d modular bootstrap \cite{Lin:2019kpn, lin2021mathbbzn, Pal:2020wwd, collier2021bootstrapping}; and even statements about confinement in (1+1)d QCD \cite{komargodski2020symmetries, Nguyen:2021naa, Delmastro:2021otj} and quantum gravity \cite{rudelius2020topological}. Very broadly, these ideas fit into a series of research programs studying applications of symmetries in quantum field theory and their generalizations to ``higher-form symmetries'' (see e.g. \cite{Gaiotto_2015}) and ``non-invertible symmetries'' (see e.g. \cite{Bhardwaj_2018} in 2d).

Our goal in this paper is to understand the construction of $\bbZ_m$ duality defects in the chiral $(E_8)_1$ WZW model, i.e. lines which separate the theory from its $\bbZ_m$ orbifold, and their associated twisted partition functions. Equivalently, we will focus on $\bbZ_m$ Tambara-Yamagami category actions on the $E_8$ lattice Vertex Operator Algebra (VOA). Let us briefly motivate some of these choices of specialization. 

First, duality defects are an interesting object of study in the realm of non-invertible symmetries. They are non-invertible topological defect lines which separate a theory from a gauged theory, and help to explain phenomena such as Kramers-Wannier duality in statistical physics \cite{KramWan, Frohlich_2004, Frohlich:2006, Frohlich:2009gb} (see \cite{Koide:2021zxj, Choi:2021kmx, Kaidi:2021xfk} for related work in 4d). They also provide ``spin-selection rules'' (see e.g. \cite{Chang_2019}) and other constraints on CFTs, which play a role in e.g. the 2d modular bootstrap. There is also a sense (which we will review) in which duality defects are the simplest defects after symmetry defects, this makes them a good tool for understanding the physics of non-invertible defect lines.

To motivate the choice of chiral $(E_8)_1$, we compare to two other systems with duality defects: the Ising CFT and the Monster CFT. The Ising CFT has a $\bbZ_2$ duality defect, and all of the properties of this defect can be understood from the fact that it has a (2+1)d interpretation as one of the three simple anyons in the Ising TFT. By comparison, the TFT/MTC associated to chiral $(E_8)_1$ is trivial (it has no irreducible modules besides itself), so we cannot use this approach to study its duality defects. Moving on, the Monster CFT has $\bbZ_p$ duality defects (where $p$ is a prime dividing the order of the Monster group) and in \cite{MonsterCFT} the $\bbZ_2$ defects of the Monster CFT were found by fermionization (see Section \ref{sec:fermionization}). Generalizing such an analysis to a ``parafermionization'' could be possible \cite{Yao:2020dqx, thorngren2021fusion}, but our analysis evades the technical difficulties involved with this process. 

All in all, our strategy should generalize to theories with duality defects who do not lift to lines in the representation category of the theory, and which are more formidable than minimal models, such as the Monster and other Sporadic group CFTs. Additionally, we conjecture that it will work for higher $n$-ality defects.

The chiral $(E_8)_1$ theory is also interesting in its own right. For example, it is the unique holomorphic VOA with central charge $8$, and hence describes the boundary edge modes of Kitaev's (2+1)d $E_8$ phase \cite{kitaevTalk} (see \cite{Kaidi:2021gbs} for a recent discussion). Commutant pairs inside the chiral $(E_8)_1$ theory have also appeared in general CFT contexts and the study of MLDEs and the conformal bootstrap \cite{Mukhi:2020sxt, Hegde:2021sdm, lin2021mathbbzn, collier2021bootstrapping}.

Note: throughout the paper we will always write the total number of spacetime dimensions.

\subsection{Outline of the Paper}
In this paper we want to construct $\bbZ_m$ duality defects of the chiral $E_8$ theory, but a side-mission will be to carve a pathway through the literature in a ground-up way for physicists, focusing on examples and highlighting various moving pieces to enable explicit calculations. We also hope this will be of benefit to non-physicists to understand the physical point of view. We do not review any formal mathematical definitions (e.g. of fusion categories or vertex algebras) as such reviews have been done better in other places, although we do provide references to those materials throughout. When we do introduce some mathematical formalism, we try to emphasize the physical context so that the mathematical definitions seem ``inevitable.''

In Section \ref{sec:topologicalDefects} we provide background for defects in 2d CFT. We recap how topological defect lines appear in CFTs in Section \ref{sec:DDCFT} and how they relate to gauging symmetries (or orbifolds) in Section \ref{sec:orbifoldingAbSym}; we also set notation for the rest of the paper. We briefly recount the story of the critical Ising model for concreteness in Section \ref{sec:IsingCFT}. We conclude in Section \ref{sec:FusionCats} by providing the data needed to define the fusion categories $\Vc_G$ and $\TY(A,\chi,\tau)$ and provide a physical story for those data.

In Section \ref{sec:DDinE8} we build all the tools necessary to compute duality defects in the chiral $(E_8)_1$ theory. Since we prefer to think of this theory from the lattice VOA viewpoint, we describe the construction of lattice VOAs in Section \ref{sec:LVOACon}. Understanding symmetries of VOAs with all the technical details is important for our construction, so in Section \ref{sec:Automorphisms} we recount the classic story about lifts from the underlying lattice, Kac's theorem for finite order automorphisms of Lie algebras, and a (computer-implementable) way to compute the relevant twisted characters. In Section \ref{sec:orbifoldsOfVE8} we discuss orbifolds from the lattice VOA perspective to make both calculations and the connections to physics concrete. In Section \ref{sec:DefectedPFs} we explain how to compute duality defected partition functions in the chiral $E_8$ theory, and in Section \ref{sec:Z2DDs} we implement our proposal for $\bbZ_2$ symmetries and compare to the result obtained by fermionization.

In Section \ref{sec:higherDDs} we show the duality defected partition functions for $\bbZ_3$, $\bbZ_4$, and $\bbZ_5$. The purpose of this section is not just the enumeration of tables of defects, but to show-off potential complications that may arise in computing defects. For instance, we see the relevance of order-doubling in Section \ref{sec:Z3Defects}, and what happens when the invariant sub-VOA does not come from a root lattice in Section \ref{sec:Z4Defects} and Section \ref{sec:Z5Defects}. In Section \ref{sec:Z4Defects} we also describe a computer algorithm for systematically computing ($q$-expansions of) defects.

In Section \ref{sec:3dTFT} we rephrase our discussion of duality defects from the point of view of (2+1)d TFTs and gapped boundary conditions. In Section \ref{sec:topBds} we provide a brief review of the relationship between (2+1)d TFTs, Modular Tensor Categories (MTCs), and their gapped boundary conditions, then discuss how orbifolds of 2d theories (and duality defects) can be understood from this picture in Section \ref{sec:DDs}. We conclude with a description of how this applies to holomorphic VOAs (like the $E_8$ lattice VOA) in Section \ref{sec:holoVOA}.

We end the paper with a list of open problems in Section \ref{sec:OpenProblems}. Namely those where explicit computations could give examples of the underlying physics.

In Appendix \ref{sec:PottsDefect} we give an alternative way to compute one of our $\bbZ_3$ duality defects from the Potts CFT, framing it from a (2+1)d point of view. In Appendix \ref{sec:Rearrangements} we make some technical comments on canonicality and symmetric non-degenerate bicharacters associated with the duality defects.

\section{Topological Defects in 2d}\label{sec:topologicalDefects}
The first example of a topological defect line in 2d occurs in the study of $0$-form symmetries.\footnote{Since we will not discuss higher-form symmetries at all (until Section \ref{sec:3dTFT}), from here out we will drop the word $0$-form.} Given a theory with a global $G$ symmetry, each element $g\in G$ corresponds to a topological line defect $X_g$ in the 2d theory: if $X_g$ sweeps past an operator insertion, it acts by $g$ in the appropriate way (see e.g. \cite{Gaiotto_2015, yujiCERN}). In the case that $G$ is a continuous group, these topological lines are simply the (exponentials of the) Noether charges obtained by integrating the conserved currents associated to $G$ along the line. Their topological nature can be understood as a consequence of the Ward identities for such currents.

Continuous or discrete, the topological defect lines associated with symmetries compose in a natural way: if two symmetry defects labelled by $g,h\in G$ are brought sufficiently close together, in parallel, and with the same orientation, they fuse as $X_{g} \otimes X_{h} = X_{gh}$. As a result, we find that topological defect lines associated with symmetries are \textit{invertible}, with $X_{g^{-1}}$ the left and right inverse of the line $X_g$.

However, not all topological defect lines in a theory are invertible. Such non-invertible defect lines in QFTs provide a rich source of information about a theory, because they are a manifestation of some non-trivial topological data including dualities, anomalies, and other constraints on RG flows. 

The classical example of a non-invertible topological defect line is the ``Kramers-Wannier'' duality defect in the Ising model \cite{Frohlich_2004,Frohlich:2006,FR_HLICH_2010}, which relates correlators in the Ising model to disorder correlators. We briefly review this in Section \ref{sec:IsingCFT} (see also \cite{Thorngren:2018bhj,Ji:2019ugf,MonsterCFT} for recent discussions of the Kramers-Wannier duality defect). 

Another source of topological defect lines are the ``Verlinde lines'' of Rational Conformal Field Theories (RCFTs), which are not mutually exclusive from the two previous examples. Verlinde lines are a natural consequence of the bulk-boundary relationship between a 2d RCFT and its associated (2+1)d TFT: they can be understood as 2d topological defects coming from anyons of the associated bulk (2+1)d TFT, brought to the boundary where the 2d RCFT data lives \cite{Verlinde:1988sn, Petkova:2000}.

In this paper, we will focus on topological defect lines which are captured by the mathematical data of a fusion category. Fusion categories describe theories of topological defect lines which are particularly discrete in that they are defined to only have a finite number of simple (irreducible) lines. In addition to being mathematically tractable, fusion categories are also interesting because they are ``rigid'' and cannot be ``continuously deformed'' (under RG flows, say), which is mathematically captured by the concept of Ocneanu rigidity \cite{ENO}. This means that they put constraints on RG flows very closely related to those implied by 't Hooft anomalies \cite{tHooft:1979rat}, see Section 7 of \cite{Chang_2019} for a modern discussion.

By now, fusion categories have appeared in a variety of ways in the physics literature, broadly in the intersection of fields that study topological aspects of quantum field theories. This includes more traditional ``high-energy'' contexts (e.g. \cite{Chang_2019, komargodski2020symmetries, Bhardwaj_2018, thorngren2019fusion, gaiotto2020orbifold}) and especially in the study of condensed matter systems and quantum computation (e.g. \cite{Kong:2013aya, Lan:2013wia, Kong:2014qka, Ji:2019eqo, Ji:2019jhk}).

Rather than re-define fusion categories, we recommend the enthusiastic physics reader to check out \cite{Bhardwaj_2018} for the definitions of a fusion category (called a ``symmetry category'' there) and to Section 3 of \cite{Chang_2019}. For the reader in a hurry, we recommend Section 5 and Appendix A of \cite{komargodski2020symmetries}.



\subsection{CFT and Defect Partition Functions}\label{sec:DDCFT}
Since a defect is an inhomogeneity in our system, part of defining a line defect in 2d involves specifying boundary conditions for fields on either side of the defect line. We can do this in a 2d CFT by describing how it acts on states of the Hilbert space.

In a 2d CFT on the plane, any defect line $X$ defines an operator $\hat{X}$ on the Hilbert space $\mathcal{H}$ of states on a circle, which acts as follows: place a state $\phi \in \mathcal{H}$ at the origin and the defect on the unit circle. Then the state that this setup provides is by definition $\hat{X}\phi \in \mathcal{H}$ \cite{FR_HLICH_2010}. Said more cylindrically, prepare the state $\phi$ in the infinite past and have the defect line waiting ``half way'' up the cylinder, then the out state is by definition $\hat{X}\phi$.

From this, we obtain a very algebraic definition of a topological defect line as one that commutes with all the modes of the stress tensor
\begin{equation}
    [\hat{L}_n,\hat{X}]=0=[\hat{\bar{L}}_n,\hat{X}]\,.
\end{equation}

The topological defect line $X$ does not need to lie on a closed path. However, if we do have a defect line along an open oriented path, then we will also have to specify how the line operator ends. The Hilbert space $\mathcal{H}_X$ of operators upon which $X$ can end is the Hilbert space of the theory on a circle, except with a single future oriented $X$ defect piercing the circle.

To obtain the CFT partition functions, we simply periodicitize time in the cylinder setups. The partition function with the $X$ line inserted along a spacelike slice (a twist in Euclidean time) is given by
\begin{equation}
    Z^X(\tau,\bar{\tau}) := \Tr_{\mathcal{H}} \hat{X} q^{L_0-\frac{c}{24}} \bar{q}^{\bar{L}_0-\frac{c}{24}}\,,
\end{equation}
where $q := e^{2\pi i \tau}$. Likewise, the partition function with the $X$ line inserted along a timelike slice (a twist in space) is given by
\begin{equation}
    Z_X(\tau,\bar{\tau}) := \Tr_{\mathcal{H}_X} q^{L_0-\frac{c}{24}} \bar{q}^{\bar{L}_0-\frac{c}{24}}\,.
\end{equation}
We illustrate these setups in Figure \ref{fig:PartitionFunctions}. Note that the two are related by the modular $S$-transformation $\tau \mapsto -\frac{1}{\tau}$ (and similarly for antiholomorphic), because the $S$-transformation swaps the two torus cycles.

\begin{figure}
	\begin{minipage}{.45\textwidth}
	\centering
	\begin{tikzpicture}[baseline={(current bounding box.center)}]
    \draw (0,0) ellipse (1.25 and 0.5);
    \draw (-1.25,0) -- (-1.25,-3.5);
    \draw (-1.25,-3.5) arc (180:360:1.25 and 0.5);


    \draw[dashed, very thick, red] (-1.25,-1.75) arc (180:360:1.25 and -0.5);
    \draw[very thick, red, ->- = 0.45 rotate 180] (-1.25,-1.75) arc (180:360:1.25 and 0.5);
    \draw[left] (-1.25, -1.75) node{$X$};
    
    \draw [dashed] (-1.25,-3.5) arc (180:360:1.25 and -0.5);
    \draw (1.25,-3.5) -- (1.25,0);  
    \fill [gray,opacity=0.1] (-1.25,0) -- (-1.25,-3.5) arc (180:360:1.25 and 0.5) -- (1.25,0) arc (0:180:1.25 and -0.5);
    \end{tikzpicture}
	\end{minipage}
	\begin{minipage}{.45\textwidth}
	\centering
	\begin{tikzpicture}[baseline={(current bounding box.center)}]
    \draw (0,0) ellipse (1.25 and 0.5);
    \draw (-1.25,0) -- (-1.25,-3.5);
    \draw (-1.25,-3.5) arc (180:360:1.25 and 0.5);
    \draw [dashed] (-1.25,-3.5) arc (180:360:1.25 and -0.5);
    \draw (1.25,-3.5) -- (1.25,0);  
    \fill [gray,opacity=0.1] (-1.25,0) -- (-1.25,-3.5) arc (180:360:1.25 and 0.5) -- (1.25,0) arc (0:180:1.25 and -0.5);
    
    \draw[very thick, red, ->- = 0.64 rotate 0] (0,-3.5-0.5) -- (0,0-0.5);
    \draw[left] (-1.25+1.20, -1.75) node{$X$};
    
    \end{tikzpicture}
    \end{minipage}
	\caption{On the left, an $X$ defect line (red) is inserted along a spatial slice, this gives a ``twist'' by $X$ in the time direction. On the right, the defect line is inserted along the time direction. Periodicitizing this setup on the right, we get the trace over the $X$-twisted Hilbert space $\mathcal{H}_X$.}
	\label{fig:PartitionFunctions}
\end{figure}
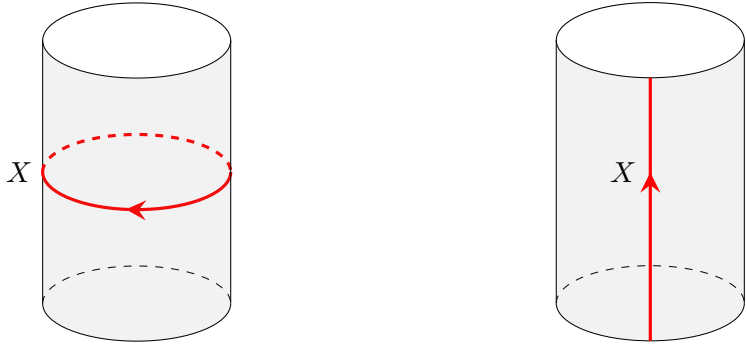

As an almost trivial example, if $X_{-}$ is a topological defect line generating a $\bbZ_2$ symmetry, then, basically by definition, if $\phi_i$ is charged under the $\bbZ_2$ symmetry we have
\begin{equation}
    \hat{X}_{-} \ket{\phi_i} = - \ket{\phi_i}\,.\label{eq:chargedAction}
\end{equation}
When tracing over the states in the Hilbert space to obtain $Z^{X_-}(\tau,\bar{\tau})$, for example, we will need to consider the signs from the charged states.

When studying a RCFT, another natural collection of topological defects comes from Verlinde line operators \cite{Verlinde:1988sn, Petkova:2000, Chang_2019}. Verlinde lines are in one-to-one correspondence with primaries, and for a diagonal RCFT, a Verlinde line $X_i$ commutes with the whole chiral algebra and acts on a primary by
\begin{equation}\label{eq:VerlindeAction}
    \hat{X}_i \ket{\phi_j} = \frac{S_{ij}}{S_{0j}}\ket{\phi_j}\,,
\end{equation}
where $S_{ij}$ is the modular $S$-matrix of the CFT. It is a straightforward application of Verlinde's formula to show that the Verlinde lines then fuse according to the fusion rules of the associated primaries, i.e.
\begin{equation}
    X_i \otimes X_j = \bigoplus_{k} N_{ij}^k X_k\,.
\end{equation}

Given any topological defect line $X$ encircling a local operator $\phi$ in the plane, we can take an equivalent view where the line $X$ squashes down leaving us with only the local operator $X\phi$ in the plane. If, instead of a general $\phi$, we had chosen the identity operator, it's clear we are just computing the expectation value of a loop of $X$ in the plane, aka the quantum dimension of $X$. Throughout, we will denote this as $\dim X$ rather than $\expval{X}$.

More generally, when a defect line $X$ sweeps past a local operator $\phi$, the result is not necessarily a local operator, but can be written as a local operator plus a defect operator $\phi^\prime$ that lies on the end of a topological tail adjoined to $X$, as depicted in Figure \ref{fig:defectLines}.

\begin{figure}
\begin{minipage}{.45\textwidth}
\centering
\begin{equation*}
	\begin{tikzpicture} [scale = 1, baseline = -2.5]
        \draw[thick, ->- = .52 rotate -10] (0,0) arc (0:360:0.8);
        \draw[fill] (-0.8,0) circle [radius = .05];
        \node[anchor = south] at (-0.8,0) {$\phi$};
        \node[anchor = east] at (-0.8*2,0) {$X$};
	\end{tikzpicture} \,=
    \begin{tikzpicture} [scale = 1, baseline = -2.5]
        \draw[fill] (0,0) circle [radius = .05];
        \node[anchor = south] at (0,0) {$X\phi$};
	\end{tikzpicture}
\end{equation*}
\end{minipage}
\begin{minipage}{.45\textwidth}
\centering
\begin{equation*}
\begin{tikzpicture}[scale = 1, baseline = 0]
        \coordinate (o) at (0,0);
        \coordinate (d) at (0.5,-1);
        \coordinate (m) at (0.5,0);
        \coordinate (u) at (0.5,1);
        \draw[fill] (o) circle [radius = .05];
        \node[anchor = south] at (o) {$\phi$};
        \draw[thick, ->- = .55 rotate 0] (d) -- (u);
        \node[anchor = west] at (m) {$X$};
    \end{tikzpicture} = \frac{1}{\dim X} \left(
    \begin{tikzpicture}[scale = 1, baseline = 0]
        \coordinate (o) at (0.7,0);
        \coordinate (d) at (0,-1);
        \coordinate (m) at (0,0);
        \coordinate (u) at (0,1);
        \draw[fill] (o) circle [radius = .05];
        \node[anchor = south] at (o) {${X\phi}$};
        \draw[thick, ->- = .55 rotate 0] (d) -- (u);
        \node[anchor = east] at (m) {$X$};
    \end{tikzpicture} \right) \,\,+\,
    \begin{tikzpicture}[scale = 1, baseline = 0]
        \coordinate (o) at (1,0);
        \coordinate (d) at (0,-1);
        \coordinate (m) at (0,0);
        \coordinate (u) at (0,1);
        \draw[fill] (o) circle [radius = .05];
        \node[anchor = south] at (o) {$\phi^\prime$};
        \draw[thick, ->- = .55 rotate 0] (d) -- (m);
        \draw[thick, ->- = .55 rotate 0] (m) -- (u);
        \draw[->- = .5 rotate 0,dashed] (m) -- (o);
        \node[anchor = east] at ($(m)!.5!(d)$) {$X$};
        \node[anchor = east] at ($(m)!.5!(u)$) {$X$};
    \end{tikzpicture}
\end{equation*}
\end{minipage}
	\caption{On the left, a topological defect line $X$ encircles a $\phi$ insertion in the plane. Since the line is topological, this is equivalent instead to the local operator $X\phi$. On the right, a topological defect line $X$ sweeps past the $\phi$ insertion leaving behind a local operator and a defect operator $\phi^\prime$. The defect operator is connected to $X$ by a topological tail (dotted).}
	\label{fig:defectLines}
\end{figure}

There are some interesting and important subtleties about all of this machinery, which are described in Section 2 of \cite{Chang_2019}. We have overlooked them because they are not particularly pressing for our discussion, but we will mention two immediate ones.

The first is that the diagram on the right of Figure \ref{fig:defectLines} simply does not reduce to the diagram on the left, we are left with the naive answer plus a topological tadpole. The vanishing of this tadpole can be proven if one assumes that the topological defect lines act faithfully on the bulk local operators. This condition can be violated in the case the CFT has multiple ground states, but such CFTs are just a direct sum of theories and won't be relevant for us.\footnote{There exists some recent literature on this topic and its interplay with higher-form symmetries \cite{komargodski2020symmetries,Sharpe:2019ddn,Yu:2020twi}.}

The second subtlety is that in 2d there are local curvature counterterms which may be added to the action. Since we are only really concerned with producing torus partition functions, this will not play any role because the torus is flat. In any case, following such terms would just change overall phases.

\subsection{Orbifolding Abelian Symmetries}\label{sec:orbifoldingAbSym}
To warm-up, consider a 2d CFT $T$ on $M$ with a global finite Abelian symmetry $A$. If the symmetry is non-anomalous, then we can couple the theory to a background $A$-connection in a well-defined way. Since $A$ is finite, this connection is necessarily flat and is specified by an $\alpha \in H^1(M,A)$ which labels the holonomies around the different cycles of $M$. 

Framed in terms of topological defects, we can represent a background $A$-connection by a network of topological $A$-defect lines meeting at trivalent junctions (see \cite{Tachikawa:finite} and diagrams within for an introduction). When a local operator passes the topological defect line labelled by $g\in A$, it applies the appropriate (linear) $g$-action.

For each background connection $\alpha\in H^1(M,A)$ we obtain a ``twisted partition function'' of the theory $Z_T[\alpha]$, which is like the ``regular'' partition function with (invertible) topological symmetry defect lines inserted around non-contractible loops. For example, on the torus, we have the $\abs{A}^2$ twisted partition functions
\begin{equation}
    Z_T[g,h] := \Tr_{\mathcal{H}_g} \hat{X}_h q^{L_0-\frac{c}{24}} \bar{q}^{\bar{L}_0-\frac{c}{24}}\,.
\end{equation}
Here $g,h\in A$, $\mathcal{H}_g = \mathcal{H}_{X_g}$ is the $g$-twisted Hilbert space, and we have suppressed dependence on $\tau$ and $\bar{\tau}$.

The \textit{orbifold} theory $[T/A]$ is the theory obtained by gauging the discrete $A$ symmetry of $T$. As is familiar from the path integral, gauging a symmetry is making the background connection dynamical, so the orbifold is obtained by summing over the twisted partition functions. Moreover, when we orbifold there is an emergent dual symmetry, given by the Pontryagin dual $\hat{A}$ of $A$, which comes from the action of the Wilson lines for the $A$ gauge fields \cite{vafa:quantumSymmetry}. This new dual symmetry can be coupled to a background flat connection $\beta \in H^1(M,\hat{A})$, and we obtain a formula relating the twisted partition functions in the different theories
\begin{equation}
    Z_{[T/A]}[\beta] = \frac{1}{\sqrt{\abs{H^1(M,A)}}} \sum_{\alpha \in H^1(M,A)} e^{i(\beta,\alpha)} Z_T[\alpha]\,.\label{eq:twistedPartFunc}
\end{equation}
The exponential denotes the intersection pairing between $H^1(M,A)$ and $H^1(M,\hat{A})$. Non-Abelian symmetries can be treated with only a bit more work.\footnote{Even more generally, we can gauge by ``symmetric Frobenius algebras'' in a theory with fusion category symmetry by inserting the corresponding triangulating trivalent mesh of topological defect lines as described in \cite{FR_HLICH_2010, Carqueville_2016, Bhardwaj_2018}. We discuss fusion categories more in Section \ref{sec:FusionCats}.}

Equation \eqref{eq:twistedPartFunc} is not necessarily the most general possible way to orbifold by a symmetry. In particular, one can include extra $U(1)$ weights $\epsilon(\alpha)$ in front of the twisted partition functions, and still obtain a gauge invariant partition function. In the CFT language, such a choice is known as a choice of ``discrete torsion'' \cite{vafa:Torsion1, vafa:Torsion2, sharpeTorsion}. Such a choice of $U(1)$ phase is classified by a group cohomology class $\nu_2\in H^2(A,U(1))$, so that most generally we have
\begin{equation}
    Z_{[T/_{\nu_2}A]}[\beta] = \frac{1}{\sqrt{\abs{H^1(M,A)}}} \sum_{\alpha \in H^1(M,A)} e^{i(\beta,\alpha)} \epsilon_{\nu_2}(\alpha) Z_T[\alpha]\,.
\end{equation}
$\epsilon_{\nu_2}$ is obtained from $\nu_2$ by evaluating it on a triangulation of $M$, e.g. on the torus with fluxes in the cycles given by $(\alpha_1,\alpha_2)$, so $\epsilon_{\nu_2}(\alpha_1,\alpha_2) = \nu_2(\alpha_1,\alpha_2)/\nu_2(\alpha_2,\alpha_1)$.

The two (commensurate) modern viewpoints on this discrete torsion phase are as follows: Anomalies of 2d theories manifest as phase ambiguities in coupling $T$ to a background connection, which in turn are encoded in a cohomology class $\mu_3 \in H^3(A, U(1))$. When a theory is non-anomalous, we don't just need to know that the cohomology class is trivial, but must pick a trivialization $\nu_2 \in H^2(A, U(1))$ that consistently resolves all phase ambiguities. Alternatively, we may view the phase as the partition function for a 2d SPT phase $\epsilon_{\nu_2}(\alpha)\sim e^{i S_{\nu_2}[\alpha]}$. Then orbifolding with discrete torsion is equivalent to taking a non-anomalous theory $T$, with phase ambiguities already resolved, and then stacking with such a 2d SPT phase before orbifolding \cite{Chen:2011pg, wenAnomalies} (see also \cite{Kapustin:2014tfa}).

\subsubsection{Example: \texorpdfstring{$\bbZ_2$}{Z2} Orbifold}
To get a better appreciation of what happens at the level of the Hilbert spaces, and see where the dual symmetry comes from, consider a theory $T$ with a non-anomalous (bosonic) $\bbZ_2$ symmetry. The untwisted orbifold partition function on the torus is
\begin{align}
    Z_{[T/\bbZ_2]}[0,0] 
        &= \frac{1}{2}\sum_{g,h \in \bbZ_2} Z_T[g,h]\\
        &= \sum_{g\in\bbZ_2} \Tr_{\mathcal{H}_g} \left[ \left(\frac{\hat{X}_{+}+\hat{X}_{-}}{2}\right) q^{L_0-\frac{c}{24}} \bar{q}^{\bar{L}_0-\frac{c}{24}}\right]\,.
\end{align}
The term in round brackets just projects us onto the states which are invariant under the $\bbZ_2$ symmetry when we trace over $\mathcal{H}_g$. This shows that the operators that contribute to the untwisted orbifold partition function are those in both the untwisted and twisted Hilbert spaces, $\mathcal{H}_{+}$ and $\mathcal{H}_{-}$ respectively, but only those which are invariant under the $\bbZ_2$ action. 

Said differently, we can decompose these Hilbert spaces into sums of even and odd subsectors
\begin{align}
    \mathcal{H}_+ &= \mathcal{H}_+^{\text{Even}} \oplus \mathcal{H}_+^{\text{Odd}}\,, \label{eq:HUntwisted}\\
    \mathcal{H}_- &= \mathcal{H}_-^{\text{Even}} \oplus \mathcal{H}_-^{\text{Odd}}\,. \label{eq:HTwisted}
\end{align}
Then the untwisted orbifold Hilbert space contains just even (i.e. gauge-invariant) operators, while the twisted orbifold Hilbert space contains just odd operators
\begin{align}
    \mathcal{H}_{\textrm{Orbi.}}^{\textrm{Even}} 
        &= \mathcal{H}_+^{\text{Even}} \oplus \mathcal{H}_-^{\text{Even}}\,,\label{eq:HOrbiUntwisted}\\
    \mathcal{H}_{\textrm{Orbi.}}^{\textrm{Odd}} 
        &= \mathcal{H}_+^{\text{Odd}} \oplus \mathcal{H}_-^{\text{Odd}}\,\label{eq:HOrbiTwisted}.
\end{align}
Hence we see the emergent $\hat{\bbZ}_2$ symmetry appear in the definition of this new twisted orbifold Hilbert space. For $A=\bbZ_m$ orbifolds we have to replace ``Even'' and ``Odd'' by the appropriate $\bbZ_m$ projectors, which are labelled by representations of $A$, i.e. elements of $\hat{A}$.

\subsection{Example: The Ising CFT}
\label{sec:IsingCFT}
Consider the (full) $c=\frac{1}{2}$ Ising CFT, which describes the Ising model at criticality. The theory has 3 Virasoro primaries: the identity operator $\mathds{1}$, the energy operator $\epsilon$, and the spin operator $\sigma$,
with conformal weights $(h,\bar{h})=(0,0),(\frac12,\frac12),(\frac{1}{16},\frac{1}{16})$ respectively \cite{BPZ}. The fusion rules are
\begin{equation}
        [\sigma][\sigma] = [\mathds{1}] \oplus [\epsilon]\,,\quad [\sigma][\epsilon] = [\sigma]\,,\quad [\epsilon][\epsilon] = [\mathds{1}]\,. \label{eq:isingFusion}
\end{equation}

The Ising CFT also has exactly 3 simple topological defect lines, these are the Verlinde lines naturally associated with the 3 primaries: $X_{\mathds{1}}$, $X_{\epsilon}$, and $X_{\sigma}$. They fuse according to the fusion rules of the primaries in Equation \eqref{eq:isingFusion} above, and so it is clear that the $\sigma$ line is not-invertible. The action of $X_{\mathds{1}}$ on primaries is straightforward. $X_\epsilon$ is the $\bbZ_2$ symmetry line, so acts only non-trivially on the $\bbZ_2$ charged operator(s) ($\sigma$, $\psi$, $\bar\psi$, etc.) as in Equation \eqref{eq:chargedAction}. The action of $X_\sigma$ on the Ising CFT primaries are depicted in Figure \ref{fig:IsingRules}.

\begin{figure}
\begin{minipage}{.45\textwidth}
\centering
\begin{equation*}
	\begin{tikzpicture}[scale = 1.2, baseline = 0]
        \coordinate (o) at (0,0);
        \coordinate (d) at (0.5,-1);
        \coordinate (m) at (0.5,0);
        \coordinate (u) at (0.5,1);
        \draw[fill] (o) circle [radius = .05];
        \node[anchor = south] at (o) {$\sigma$};
        \draw[thick, ->- = .55 rotate 0] (d) -- (u);
        \node[anchor = west] at (m) {$X_\sigma$};
    \end{tikzpicture} \,=\, 
    \begin{tikzpicture}[scale = 1.2, baseline = 0]
        \coordinate (o) at (1,0);
        \coordinate (d) at (0,-1);
        \coordinate (m) at (0,0);
        \coordinate (u) at (0,1);
        \draw[fill] (o) circle [radius = .05];
        \node[anchor = south] at (o) {$\mu$};
        \draw[thick, ->- = .55 rotate 0] (d) -- (m);
        \draw[thick, ->- = .55 rotate 0] (m) -- (u);
        \draw[->- = .5 rotate 0,dashed] (m) -- (o);
        \node[anchor = north] at ($(m)!.5!(o)$) {$X_\epsilon$};
        \node[anchor = east] at ($(m)!.5!(d)$) {$X_\sigma$};
        \node[anchor = east] at ($(m)!.5!(u)$) {$X_\sigma$};
    \end{tikzpicture} 
\end{equation*}
\end{minipage}
\begin{minipage}{.45\textwidth}
\centering
\begin{equation*}
	\begin{tikzpicture}[scale = 1.2, baseline = 0]
        \coordinate (o) at (0,0);
        \coordinate (d) at (0.5,-1);
        \coordinate (m) at (0.5,0);
        \coordinate (u) at (0.5,1);
        \draw[fill] (o) circle [radius = .05];
        \node[anchor = south] at (o) {$\epsilon$};
        \draw[thick, ->- = .55 rotate 0] (d) -- (u);
        \node[anchor = west] at (m) {$X_\sigma$};
    \end{tikzpicture} \,=\, 
    \begin{tikzpicture}[scale = 1.2, baseline = 0]
        \coordinate (o) at (1,0);
        \coordinate (d) at (0,-1);
        \coordinate (m) at (0,0);
        \coordinate (u) at (0,1);
        \draw[fill] (o) circle [radius = .05];
        \node[anchor = south] at (o) {$-\epsilon$};
        \draw[thick, ->- = .55 rotate 0] (d) -- (u);
        \node[anchor = east] at (m) {$X_\sigma$};
    \end{tikzpicture} 
\end{equation*}
\end{minipage}
	\caption{Left, the $X_\sigma$ defect sweeps past a $\sigma$ primary, turning it into the disorder operator $\mu$, plus a topological $\bbZ_2$ tail. Such a $\mu$ is traditionally called a ``twist-field.'' Right, the $X_\sigma$ defect sweeps past the relevant operator $\epsilon$ and turns it into $-\epsilon$.}
	\label{fig:IsingRules}
\end{figure}
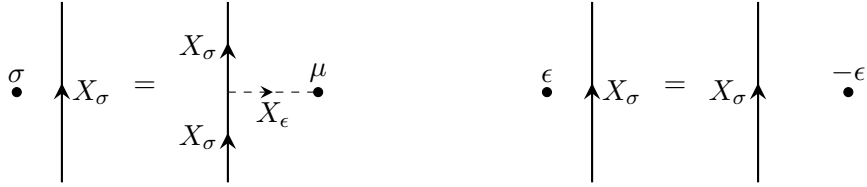

Famously, in the Ising model (not necessarily at criticality), the correlator of spins at some inverse temperature $\beta$ is equal to the correlator of disorder/twist operators at the Kramers-Wannier dual inverse temperature $\tilde\beta$, i.e. $\langle \sigma(z_1)\cdots\sigma(z_n) \rangle_\beta = \langle \mu(z_1)\cdots\mu(z_n) \rangle_{\tilde\beta}$. This exhibits some  duality between weak and strong couplings and can be used to understand the critical point of the Ising model \cite{KramWan, Frohlich_2004} (see also \cite{Aasen:2016dop, aasen2020topological}).

However, the disorder fields $\mu$ are not local operators in the theory alone, but have to be accompanied by a topological ``line of frustration,'' which has to end similarly on another $\bbZ_2$ ``twist-field.'' In the CFT context above, this line of frustration is the $\bbZ_2$ symmetry line $X_{\epsilon}$, and the $\bbZ_2$ twisted Hilbert space is $\mathcal{H}_{\epsilon} = \mathrm{span}\{\mu, \psi,\bar{\psi}\}$.

A statistical interpretation of the $X_{\sigma}$ line is not immediately clear, but can be realized by first studying how it acts on the primaries in the CFT picture. In particular, when we sweep the $X_{\sigma}$ line past a $\sigma$ insertion, we are left with a $\mu$ insertion and a topological $X_\epsilon$ tail. But in the orbifold CFT $[T/\bbZ_2]$, $\mu$ is a local operator (the $\bbZ_2$ line becomes invisible when we gauge), and $\sigma$ must be accompanied by a topological tail.

Moreover, when $X_\sigma$ sweeps past an $\epsilon$ insertion, it leaves behind a $-\epsilon$ insertion. The energy operator $\epsilon$ is a relevant operator in the Ising CFT and deformations by it flow the Ising CFT to the high-temperature (unbroken) or low-temperature (broken) phases. The action of $X_\sigma$ reflects the fact that the Ising model and its orbifold see this deformation oppositely. This also explains why Kramers-Wannier duality holds away from criticality.

For these reasons, $X_\sigma$ is the archetypal \textit{duality defect}. It extends the algebra of topological defects in the Ising model from a collection of $\bbZ_2$ group-like defects to a Tambara-Yamagami category (see Section \ref{sec:FusionCats}).

It is important to stress two basic facts about the duality defect line. The first is that it is \textit{not} the non-simple topological defect line $X_{\textrm{Proj.}}:=X_{\mathds{1}}\oplus X_\epsilon$. Inserting a triangulating mesh of $X_{\textrm{Proj.}}$ into the original theory $T$ just produces $[T/\bbZ_2]$. Indeed, $X_{\sigma}$ separates the theory $T$ from $[T/\bbZ_2]$. The second basic fact is that it is not the orbifold duality defect line $X_{\textrm{Proj.}}$... but it does square to it! This is important because when the $X_\sigma$ line collides with itself (e.g. when deformed around a cycle of the torus) it does produce a $X_{\textrm{Proj.}}$ defect.

It is also instructive to compare the twisted and ``defected'' partition functions for this theory. The twisted partition functions can be written in terms of the usual Virasoro characters as
\begin{align}
    Z_T[0,0] 
        &= \lvert\chi_0\rvert^2+\lvert\chi_\frac{1}{2}\rvert^2+\lvert\chi_\frac{1}{16}\rvert^2\,,\\
    Z_T[0,1] 
        &= \lvert\chi_0\rvert^2+\lvert\chi_\frac{1}{2}\rvert^2-\lvert\chi_\frac{1}{16}\rvert^2\,,\\
    Z_T[1,0] 
        &= \chi_0\bar{\chi}_{\frac{1}{2}}+\chi_{\frac{1}{2}}\bar{\chi}_0+\lvert\chi_\frac{1}{16}\rvert^2\,,\\
    Z_T[1,1] 
        &= -\chi_0\bar{\chi}_{\frac{1}{2}}-\chi_{\frac{1}{2}}\bar{\chi}_0+\lvert\chi_\frac{1}{16}\rvert^2\,.
\end{align}
While the defected partition function, with an $X_{\sigma}$ inserted around the spatial $S^1$, is easily computable with the help of Equation \eqref{eq:VerlindeAction} for Verlinde lines on primaries
\begin{equation}
    Z_T[0,X_\sigma] = \sqrt{2}\lvert\chi_0\rvert^2 - \sqrt{2}\lvert\chi_{\frac{1}{2}}\rvert^2\,.
\end{equation}

Such a structure for the defect partition function is typical, and can be understood from the second important fact about $X_\sigma$. Since $X_\sigma \otimes X_\sigma = X_{\mathds{1}} \oplus X_\epsilon = X_{\mathrm{Proj.}}$ we can compute (without identifying $X_\sigma$ with the Verlinde line) that $\hat{X}_{\mathrm{Proj}.}\ket{\phi} = 2 \ket{\phi}$ if $\phi$ corresponds to a $\bbZ_2$ uncharged operator, and is $0$ otherwise. This means that
\begin{equation}
    \hat{X}_\sigma \ket{\phi} = \pm \sqrt{2} \ket\phi
\end{equation}
if $\phi$ is uncharged under the $\bbZ_2$, and is $0$ otherwise. Hence the defected partition function is just a sum over all of the uncharged primaries, and all we have to do is resolve a sign.

\subsection{Fusion Categories for Physics}
\label{sec:FusionCats}
Here we will recap the pertinent properties of the 2 most important fusion categories for our purposes: the fusion category of $G$ graded vector spaces $\Vc_G$ and Tambara-Yamagami categories $\TY(A,\chi,\tau)$.

\subsubsection{Fusion Category: Graded Vector Spaces \texorpdfstring{$\Vc_G$}{VecG}}
Given a finite symmetry group $G$, the topological defects associated to the symmetry group form a unitary fusion category in a natural way. Simple objects of the category are oriented lines labelled by group elements $g,h\in G$. If two such lines are brought sufficiently close together, in parallel, with the same orientation, then they fuse according to the group multiplication law
\begin{equation}
X_g \otimes X_h = X_{gh}\,\qquad \textrm{for $g,h\in G$.}
\end{equation}
The category is unitary because the orientation reversal of a line is (by definition) the inverse $X_g^* = X_g^{-1}=X_{g^{-1}}$.\footnote{Since we are technically dealing with objects in a category, by $=$ we mean that there exists a natural isomorphism between objects in the category.}

Given three simple lines labelled by $g,h,k\in G$, the data of an associator isomorphism $\alpha$ is needed to relate the object $(X_g\otimes X_h) \otimes X_k$ to $X_g\otimes (X_h \otimes X_k)$. In practice, this means that there can be an overall phase ambiguity between the background connection depicted with $X_g$ and $X_h$ fusing first or $X_h$ and $X_k$ fusing first as in Figure \ref{fig:associator}.

\begin{figure}
\centering
\begin{equation*}
	\begin{tikzpicture}[thick, scale = 1, baseline = 0]
    
    \draw[->- = .5 rotate 0] (-1.4,-2) -- (-0.67,-0.68);
    \draw[->- = .5 rotate 0] (-0.67,-0.68) -- (0.08,0.65);
    \draw[->- = .5 rotate 0] (0.08,0.65) -- (0.08,0.65+1.52);
    \draw[->- = .5 rotate 0] (0.1,-2) -- (-0.65,-2/3);
    \draw[->- = .5 rotate 0] (1.6,-2) -- (0.08,0.65);
    
    
    \draw[below] (-1.4,-2) node {$X_g$};
    \draw[below] ( 0.1,-2) node {$X_h$};
    \draw[below] ( 1.6,-2) node {$X_k$};
    
    \draw[above] (0.08,0.65+1.52) node {$(X_g\otimes X_h)\otimes X_k$};    
    
    \end{tikzpicture} \,= \,\alpha(g,h,k)\,
    \begin{tikzpicture}[thick, scale = 1, baseline = 0]
    
    \draw[->- = .5 rotate 0] (1.4,-2) -- (0.67,-0.68);
    \draw[->- = .5 rotate 0] (0.67,-0.68) -- (-0.08,0.65);
    \draw[->- = .5 rotate 0] (-0.08,0.65) -- (-0.08,0.65+1.52);
    \draw[->- = .5 rotate 0] (-0.1,-2) -- (0.65,-2/3);
    \draw[->- = .5 rotate 0] (-1.6,-2) -- (-0.08,0.65);
    
    \draw[below] (1.4,-2) node {$X_k$};
    \draw[below] (-0.1,-2) node {$X_h$};
    \draw[below] (-1.6,-2) node {$X_g$};
    
    \draw[above] (-0.08,0.65+1.52) node {$X_g\otimes (X_h\otimes X_k)$};  
    
    
    \end{tikzpicture}
\end{equation*}
	\caption{Mathematically, the associator $\alpha$ is an element of $\Hom((X_g\otimes X_h) \otimes X_k, X_g\otimes (X_h \otimes X_k))$ that defines in what way the two tensor products are equal. Physically, given a theory with non-trivializable anomaly $\alpha \in H^3(G,U(1))$, two networks of symmetry defect lines can be different up to a $U(1)$ phase.}
	\label{fig:associator}
\end{figure}
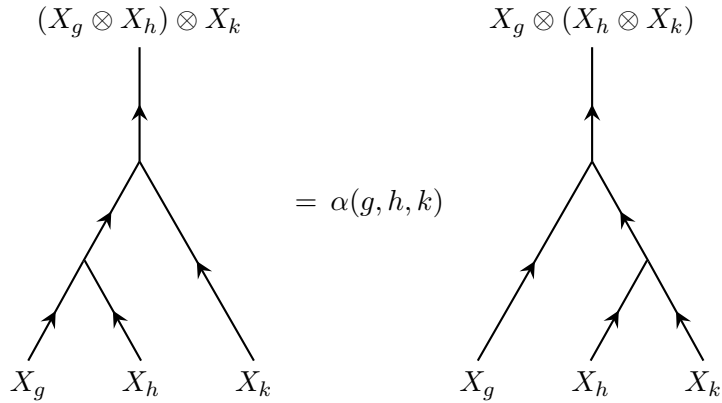

Such an $\alpha$ is required to satisfy the familiar ``pentagon identity,'' which asserts that whatever $\alpha$ is, it must agree when we consider the multiple ways to fuse four $G$ defects as in Figure \ref{fig:pentagon}.\footnote{Readers may be familiar with the data of an associator presented in the form of an $F$-symbol or $6j$-symbol.} In this case, such $\alpha$ are classified by an element of the group cohomology $H^3(G, U(1))$ (see e.g. \cite{Faddeev:1985iz, DW90}). Such phase ambiguities in coupling a theory to a background $G$ connection, which cannot be removed by adding local counterterms, are the 't Hooft anomalies from the physics literature.

\begin{figure}
\begin{center}
    \begin{tikzcd}[thick]
    & (X_g\otimes X_h) \otimes (X_k \otimes X_{\ell}) \arrow[rd, "\alpha "] & \\
    ((X_g\otimes X_h) \otimes X_k) \otimes X_{\ell} \arrow[ru, "\alpha"] \arrow[d, "\alpha \otimes 1"'] & & X_g\otimes (X_h \otimes (X_k \otimes X_{\ell})) \\
    (X_g\otimes (X_h \otimes X_k)) \otimes X_{\ell} \arrow[rr, "\alpha"'] & & X_g\otimes ((X_h \otimes X_k) \otimes X_{\ell}) \arrow[u, "1 \otimes \alpha"']
    \end{tikzcd}
\end{center}
\caption{Given 4 defect lines $X_g$, $X_h$, $X_k$, and $X_{\ell}$, we only demand that the pentagon diagram commutes. This is familiar to those in RCFT as the pentagon identity for F-symbols.}
	\label{fig:pentagon}
\end{figure}

All considered, such a collection of defects labelled by elements of $G$ with a choice of group cohomology class $\alpha \in H^3(G,U(1))$ form the fusion category $\Vc_G^\alpha$ of $G$-graded vector spaces. In practice, we are going to be interested in the case where the symmetry in question does not have an anomaly, so that we can orbifold by that symmetry, the fusion category $\Vc_G$.

It may not be a surprising fact, when paralleled with our discussion of coupling a CFT in the regular group symmetry framework, that the group of ``autoequivalences'' of the fusion category $\Vc_G$ (basically, automorphisms of $\Vc_G$ as a fusion category) is $\Aut(G) \ltimes H^2(G,U(1))$. The $\Aut(G)$ roughly corresponding to renaming of $G$ lines, and $H^2(G,U(1))$ corresponding to stacking with 2d SPT phases \cite{Bhardwaj_2018}.

\subsubsection{Fusion Category: Tambara-Yamagami  \texorpdfstring{$\TY(A,\chi,\tau)$}{TY(A,X,t)}}
In the previous subsubsection we started with a set of desired fusion rules (and unitary structure) and then were left with classifying the associators compatible with those fusion rules/unitary structure and pentagon identity. This is the same thing that Tambara and Yamagami do in a ground-up approach in their seminal paper \cite{tambara1998tensor} (see also \cite{tambara2000representations}), except their fusion rules are minimally enriched by a duality defect.

Start with a finite Abelian group $A$, the simple objects of $\TY(A,\chi,\tau)$ are again associated with elements $a\in A$, but the fusion algebra is now extended by an additional object $m$ of quantum dimension $\sqrt{\abs{A}}$.\footnote{In fact, such a fusion ring makes sense if the group is not Abelian, but it is only categorifiable if $A$ is Abelian (see \cite{tensorCatBook} Example 4.10.5).} In addition to the familiar $\Vc_A$ fusion rules, we also demand
\begin{align}
    X_a \otimes X_m &= X_m\,,\\
    X_m \otimes X_m &= \bigoplus_{a\in A} X_a\,.
\end{align}

What Tambara and Yamagami show, by direct computation, is that the data of potential associators is classified by a symmetric non-degenerate bicharacter $\chi:A\times A \to U(1)$, and a sign or ``Frobenius-Schur indicator'' $\tau = \pm 1/\sqrt{\abs{A}}$ \cite{tambara1998tensor,tambara2000representations} (see also \cite{homotopyFusion, Bhardwaj_2018, thorngren2019fusion, aasen2020topological}). In particular, the non-trivial associators are
\begin{align}
    \alpha_{a,m,b} &= \chi(a,b)\id_m\,,\\
    \alpha_{m,a,m} &= \bigoplus_b \chi(a,b) \id_b\,,\\
    \alpha_{m,m,m} &= (\tau\chi(a,b)^{-1}\id_m)_{a,b}\,.
\end{align}

A first point of note, which will be relevant later, is that TY categories are a $\bbZ_2$-extension of $\Vc_A$. In general, a $G$-extension of a fusion category $\mathcal{D}$ is a $G$-graded fusion category
\begin{equation}
    \mathcal{F} = \bigoplus_{g\in G} \mathcal{F}_g\,,
\end{equation}
satisfying $\mathcal{F}_e \cong \mathcal{D}$.

Another thing to note is that $\chi:A\times A \to U(1)$, as opposed to the more familiar pairing of $A$ with $\hat{A}$. This means $\chi$ is an assignment of (magnetic) $A$ charges to $A$ defect lines, in contradistinction to the natural Fourier-like pairing in orbifolds $\rho:A\times \hat{A} \to U(1)$. Pushing further, a bicharacter $\chi$ defines a homomorphism $A\to \hat{A}$, and non-degeneracy implies that this is an isomorphism. Since $T$ and $[T/A]$ have $A$ and $\hat{A}$ symmetry respectively, this $\chi$ reflects the freedom of choice in picking an isomorphism from $A$ to $\hat{A}$ (or rather reflects the non-canonicality of $A\cong \hat{A}$). By the same logic, $\chi$ determines an isomorphism from $\hat{A}\to A$, and while $A$ and its dual $\hat{A}$ are not canonically isomorphic, $A$ \textit{is} canonically isomorphic to its double dual $\hat{\hat{A}}$. The symmetric property ensures that whatever isomorphism $\chi$ determines from $A$ to $\hat{A}$, and $\hat{A}$ to $A$, that it respects the canonicality $A=\hat{\hat{A}}$ i.e. it is honestly the identity on $A$ and not some random automorphism of $A$.

The role of $\tau$ is less physically obvious, and shows up in the crossing relations for $m$ \cite{thorngren2019fusion}. It originates categorically as an associativity-like constraint similar in spirit to the associator in $\Vc_A^\alpha$ \cite{homotopyFusion}, but could be better understood physically.

As a closing remark, the Ising Fusion category that we are familiar with forms $\TY(\bbZ_2,1,+1/\sqrt{2})$ (see Section 4.1 of \cite{thorngren2019fusion} for elaboration).

\section{Duality Defects in \texorpdfstring{$(E_8)_1$}{(E8)1}}\label{sec:DDinE8}
The main goal of this paper is to classify and construct $\bbZ_m$ duality defects in the chiral WZW-model $(E_8)_1$. Said differently, we want to understand the actions of $\bbZ_m$ Tambara-Yamagami categories on the holomorphic lattice VOA $V_{E_8}$.

Vertex Operator Algebras are widespread in physics as the axiomatization of ``chiral algebras,'' which glue together to form a ``full CFT'' in two dimensions. We refer readers to the standard mathematical texts on VOAs \cite{frenkel1989vertex, dong1993generalized, kac1998vertex, Frenkel:2004jn, lepowsky2004introduction} or any standard textbook on 2d CFT or string theory for an introduction (a nice review for both communities appeared in \cite{Bae:2020pvv}). Lattice VOAs are also prominent in the physics literature, but understanding subtleties of their construction, automorphisms, and characters is important for our purposes, so we review them briefly in the following sections. Some particularly helpful reviews on lattice VOAs besides those already listed include \cite{Dolan:1989vr, dolan1996, van_Ekeren_2020, moller2017cyclic}.

\subsection{Lattice Vertex Operator Algebras}\label{sec:LVOACon}
Given a positive definite even lattice $L$ one can construct a lattice VOA $V_L$. Mathematically, $V_L = M_{\hat{\mathfrak{h}}}(1) \otimes \bbC_{\epsilon}[L]$ where the first factor is the ``Heisenberg VOA'' describing oscillator modes of chiral bosons (also called $\mathrm{Bos}(\mathfrak{h})$), and the second factor is the ``twisted group algebra'' describing ``quantized momentum.''\footnote{See e.g. Section 5 of \cite{moller2017cyclic} and references listed within for the formal mathematical construction of these objects and proof $V_L$ is rational, $C_2$-cofinite, self-contregredient, and of CFT type.} 

We start the construction with the quantized momentum modes: for every vector $\alpha\in L$ we have a state $\ket{\alpha}$ which is created by the vertex operator $\Gamma_{\alpha}$. By definition, the states are orthonormal $\braket{\alpha}{\beta}=\delta_{\alpha\beta}$. There are also chiral bosonic oscillators $a_n^i$, $n\in\bbZ$ and $1 \leq i \leq \rk(L)$, satisfying the usual Heisenberg commutation relations
\begin{equation}
    [a_m^i, a_n^j] = m\delta^{ij}\delta_{m+n,0}\,,
\end{equation}
with $a_m^{i\dag} = a_{-m}^i$.

Combining the two, the positive oscillator modes should annihilate the $\ket{\alpha}$
\begin{equation}
    a_n^i \ket{\alpha} = 0\,,\qquad \textrm{if $n > 0$,}
\end{equation}
and $a_0^i$ should actually behave as the momentum operator, i.e.
\begin{equation}
    a_0^i \ket{\alpha} = \alpha^i \ket{\alpha}.
\end{equation}
The physical Hilbert space is thusly generated by the $a_{-n}^i$, $n > 0$, acting on the $\ket{\alpha}$, so that a general basis state is of the form
\begin{equation}\label{eq:statesLVOA}
    a^{i_N}_{-n_N}\cdots a^{i_2}_{-n_2} a^{i_1}_{-n_1}\ket{\alpha}\,.
\end{equation}
Such a state has integral conformal weight
\begin{equation}
    h = \frac{1}{2} ( \alpha, \alpha ) + \sum_{i=1}^N n_i\,,
\end{equation}
where $(\,.\,,\,.\,)$ denotes the inner product on the lattice $L$.

As mentioned, the vertex operator $\Gamma_\alpha$ creates the state $\ket{\alpha}$, and one expects that $\Gamma_\alpha \Gamma_\beta \propto \Gamma_{\alpha+\beta}$. However, asking for $\Gamma$ to be a representation of the lattice is too strong; locality requires us to consider projective representations, i.e.
\begin{equation}
    \Gamma_\alpha \Gamma_\beta = \epsilon(\alpha,\beta)\Gamma_{\alpha+\beta}\,,
\end{equation}
where $\epsilon: L\times L \to \{\pm1\}$ is a (normalized) 2-cocycle satisfying $\epsilon(\alpha,\beta) = (-1)^{(\alpha,\beta)}\epsilon(\beta,\alpha)$. This skew condition on the 2-cocycle determines $\epsilon$ up to a coboundary, and so $\epsilon \in H^2(L,\{\pm1\})$, which will play a role for us in understanding automorphisms.\footnote{In other words, the $\Gamma_\alpha$ furnish, at minimum, a representation of a double cover $\hat{L}$ of the root lattice $L$. The freedom of choice in $\epsilon$ (due to the central extension defining the double cover not splitting) reflects the fact that there is no functor from root lattices to vertex algebras, i.e. a choice of $\epsilon$ is required \cite{theobook}. Double covers of $L$ coming from differing choices of $\epsilon$ are isomorphic, which is why one can non-canonically write $V_L = M_{\hat{\mathfrak{h}}}(1)\otimes\bbC[L]$ if they settle some signs once and for all.} 

When $L$ is the root lattice of a simply laced Lie algebra $\mathfrak{g}$, this construction gives a vertex-representation of the $\hat{\mathfrak{g}}_1$ WZW model (see e.g. Section 15.6.3 of \cite{francesco2012conformal}). In this case, the sign ambiguities which occur in trying to build $V_L$ from the root lattice $L$ are precisely the same as those in trying to build the Lie algebra $\mathfrak{g}$ from $L$. For non-simply laced algebras (or $k > 1$), one must add a free fermion (or parafermion) respectively.

Up to isomorphism, the \textit{untwisted} irreducible modules of $V_L$ are labelled by elements of the discriminant group $D_L := L^*/L$, here $L^*$ is the dual lattice to $L$ \cite{dong1993vertex}. The MTC of representations $\Rep(V_L)$ has $D_L$ group-like fusion in the obvious way. For some $\lambda \in L^*$, the character of the irreducible $V_L$-module $V_{\lambda+L}$
\begin{equation}
    \chi_{V_{\lambda+L}}(\tau) = \tr_{V_{\lambda+L}}q^{L_0-\frac{c}{24}} = \frac{\theta_{\lambda+L}(\tau)}{\eta(\tau)^{\rk(L)}}\,,
\end{equation}
where $\theta_{\lambda+L}(\tau)$ is the theta function of the shifted lattice $\lambda+ L$, i.e.
\begin{equation}
    \theta_{\lambda+L}(\tau) = \sum_{\xi \in \lambda + L} q^{\frac{1}{2}( \xi , \xi )}\,.
\end{equation}

\subsubsection{Example: \texorpdfstring{$V_{E_8}$}{VE8} as a \texorpdfstring{$V_{E_8}$}{VE8} Module} \label{ex:VE8}
The $E_8$ root lattice is positive-definite, even, and self-dual, therefore $V_{E_8}$ itself is the only irreducible $V_{E_8}$ module (aka the ``adjoint module''). VOAs whose only irreducible module are the adjoint module are called \textit{holomorphic}, i.e. $\Rep(V_{E_8}) = \Vc$.

Using the description of the $E_8$ lattice as points in $\bbZ^8\cup(\bbZ+\frac{1}{2})^8$, whose coordinates sum to an even integer, we obtain the $E_8$ partition function
\begin{align}
    \tr_{V_{E_8}} q^{L_0 - \frac{c}{24}} 
        &= \frac{1}{\eta(\tau)^{8}}\sum_{\xi \in E_8} q^{\frac{1}{2}( \xi,\xi)}\\
        &= \frac{1}{2\eta(\tau)^{8}}\sum_{x \in \bbZ^8} \left(1+(-1)^{\sum{x_i}}\right)\left(q^{\frac{1}{2}\sum{x_i^2}}+q^{\frac{1}{2}\sum{(x_i+\frac{1}{2})^2}}\right)\\
    Z_{E_8}(\tau)    
        &= \frac{1}{2\eta(\tau)^{8}}\left(\theta_1(\tau)^8+\theta_2(\tau)^8+\theta_3(\tau)^8+\theta_4(\tau)^8\right) \label{eq:E8partition}
\end{align}
Note: even though this expression is not modular invariant, since $c_L-c_R \in 8\bbZ$ its modular non-invariance can be cured by a (2+1)d gravitational Chern-Simons term \cite{Alvarez-Gaume:1983ihn, Witten:1988hf} (see \cite{Kaidi:2021gbs} for a recent discussion).

\subsection{Automorphisms and Twisted Partition Functions}\label{sec:Automorphisms}
\subsubsection*{Describing Automorphisms}
Recall that a VOA $V$ has an automorphism/symmetry given by an (invertible) linear map $\hat{g}: V \to V$, if the action of $\hat{g}$ preserves the vacuum $\hat{g}\ket{0} = \ket{0}$, preserves the (chiral) stress-tensor $\hat{g}T(z)= T(z)$, and commutes with the state operator map
\begin{equation}
    (\hat{g}\phi)(z) = \hat{g}\phi(z)\hat{g}^{-1}\,.
\end{equation}

Suppose $V_L$ is a lattice VOA constructed from the even lattice $L$. The subgroup of $\Aut(V_L)$ which maps the subalgebra of bosonic oscillators to itself is given by the (usually) non-split extension $T.O(L)$. Here $O(L)$ are the symmetries of the lattice itself and $T := \bbR^{\rk(L)}/L^*$ is the torus dual to $L$. $T$ carries the natural action of $O(L)$ viewed as a subgroup of $O(\rk(L),\bbR)$ \cite{dong1999automorphism}.

The subgroup $T < T.O(L)$ acts trivially on the bosonic oscillators, and simply modifies the group algebra of quantized momentum modes by phases. That is, for any $x\in T$ we have the induced map
\begin{align}
    a^i_n &\mapsto a^i_n \,, \\
    \Gamma_\alpha &\mapsto e^{2\pi i (x,\alpha)} \Gamma_\alpha \,.\label{eq:innerAutomorphismT}
\end{align}
Note again that this map is actually trivial if $x\in L^*$, and so we are only interested in $x \Mod L^*$. Conversely, the $O(L)$ factor acts on the bosons and also permutes the $\Gamma_\alpha$ \cite{lepowsky1985calculus, dong1999automorphism, moller2017cyclic, van_Ekeren_2020, Hohn:2020xfe}  (see also Section 4.4 of \cite{evans2020tambarayamagami}).

In this vein, suppose $g \in O(L)$ is any automorphism of the lattice $L$, then $g$ preserves the skew form $(-1)^{(\,\cdot\,,\,\cdot\,)}$. Since the skew condition determines the 2-cocycle $\epsilon$ up to a 2-coboundary, we see that $\epsilon(\alpha,\beta)$ and $\epsilon(g\alpha,g\beta)$ are in the same cohomology class; and so
\begin{equation}
    \frac{\epsilon(\alpha,\beta)}{\epsilon(g\alpha,g\beta)} = \frac{u(\alpha) u(\beta)}{u(\alpha + \beta)} \label{eq:liftExists}\,,
\end{equation}
for all $\alpha,\beta \in L$ and some 1-cocycle $u:L\to \{\pm1\}$.

It is a classic result that such a $g\in O(L)$ and choice of 1-cocycle $u$ define a lift to an automorphism $\hat{g}_u\in T.O(L) < \Aut(V_L)$ \cite{lepowsky1985calculus,dong1999automorphism}. This automorphism also acts in a nice way
\begin{align}
    \hat{g}_u(a_n)
        &=(g a)_n \,, \label{eq:liftCartan}\\
    \hat{g}_u\Gamma_{\alpha} 
        &= u(\alpha) \Gamma_{g\alpha} \label{eq:liftRoot}\,.
\end{align}
Furthermore, one can always choose ``a standard lift'' such that $u(\alpha)=1$ on $L^g$ \cite{lepowsky1985calculus}. We will not work with standard lifts right away, but with foresight we quote the fact that: any two standard lifts are conjugate in $\Aut(V_L)$ \cite{van_Ekeren_2020}.\footnote{The term standard lift is a bit of a misnomer and does not have the properties one might expect: for example, there can be more than one ``standard'' lift; moreover, since the extension $T.O(L)$ is \textit{not} split (i.e. it is not a semi-direct product), the product of standard lifts is not necessarily standard. All of this richness can be traced back to the fact that there is no canonical way to construct a Lie algebra and lattice VOA from the root lattice $L$. For further discussion see \cite{van_Ekeren_2020, moller2017cyclic, theobook} and references within.}

Given a positive definite even lattice $L$ and a lifted automorphism $\hat{g}_u \in \Aut(V_L)$, the irreducible \textit{$\hat{g}_u$-twisted} modules can be constructed and classified similar to the untwisted modules \cite{dong1996algebraic, bakalov2004twisted, moller2017cyclic}. In short, they are labelled by elements of the discriminant $L^*/L$ fixed under the $g$ action. i.e. they are of the form $V_{\lambda + L}$, where $\lambda + L \in (L^*/L)^g = \{\lambda\in L^*/L : (\id - g)\lambda\in L \}$. 

The twisted characters are
\begin{equation}\label{eq:twistedPartition}
    \tr_{V_L} \hat{g}_u q^{L_0 - \frac{c}{24}} =\frac{ \theta_{L^g, u}(\tau)}{\eta_g(\tau)}\,,
\end{equation}
where we've defined the generalized theta function
\begin{equation}
    \theta_{L^g, u}(\tau):=\sum_{\alpha \in L^g} u(\alpha) q^{( \alpha, \alpha )/2} \,.
\end{equation}
The eta product $\eta_g(\tau)$ is given by
\begin{equation}\label{eq:eta_product}
    \eta_g(\tau) := \prod_{t | m}\eta(t \tau)^{b_t},
\end{equation}
when $g\in O(L)$ has cycle shape $\prod_{t|m} t^{b_t}$. The cycle shape can be determined by calculating the characteristic polynomial of the matrix representing the automorphism, which is of the form $\prod_{t|m} (\lambda^{t} - 1)^{b_t}$. 

More physically, factors of $\eta(\tau)$ arise in characters/partition functions from the contributions of each independent tower of (chiral) bosonic modes. If $g$ is, for example, a $\bbZ_2$ symmetry swapping the $a^1_m$ and $a^2_n$ towers, then one expects to get factors of $\eta(2\tau)$ in the vacuum character of the invariant sub-VOA due to the ``extra length'' of the combined $\bbZ_2$-invariant tower.

Finally, given an order $m$ automorphism $g\in O(L)$, a standard lift $\hat{g}_u\in\Aut(V_L)$ has: order $m$ if $m$ is odd; order $m$ if $m$ is even and $(\alpha,g^{m/2}\alpha) \in 2\bbZ$ for all $\alpha$; and order $2m$ otherwise (see e.g. Corollary 5.3.6 of \cite{moller2017cyclic}).

All of these technicalities about canonicality, lifts, and order doubling may seem overkill, but actually plays an important role in constructing as many duality defects as possible.

\subsubsection*{Finding Automorphisms}
To understand how to find automorphisms of lattice VOAs, it is helpful to expand on the structure of the automorphism group of a VOA and its relationship to Lie algebras and physics.

We use the following fact for orientation: For a VOA $V$ of ``CFT-type'' (like lattice VOAs or the Monster), the weight one subspace $V_1$ is a (possibly trivial) complex Lie algebra with bracket $[u,v]=u_0 v$. Moreover, all the weight spaces of $V$-modules are Lie algebra modules for $V_1$ (including the weight spaces of the adjoint module, i.e. the $V_n$), so that we can see the VOA and its modules as strata of $V_1$ Lie algebra modules \cite{frenkel1993axiomatic}. From a physics perspective, we are not surprised that the currents form a Lie algebra under which things transform, this is just Noether's theorem and Wigner's theorem.

There are two extreme examples for the Lie algebra $V_1$ in relation to $V$: the symmetry currents $V_1$ generate $V$ and $\Aut(V_1) \cong \Aut(V)$, as in affine Kac-Moody VOAs or their semi-simple quotients; or alternatively, there are no currents and $V_1 = 0$, as with the Virasoro algebra $\Aut(\mathrm{Vir})= \{\mathds{1}\}$, and the Monster CFT $\Aut(V^\natural) = \mathbb{M}$ \cite{dong1999automorphism}.

The key move is to note that any automorphism/symmetry $\hat{g}\in \Aut(V)$ restricts to a Lie algebra automorphism of $V_1$.\footnote{A \textit{super}symmetry would relate operators of $V_1$ to operators with different spins.} Hence we ask: \textit{to what extent can we understand $\Aut(V)$ by lifting automorphisms of the Lie algebra $\Aut(V_1)$?}

The answer to this question is known for lattice VOAs. For lattice VOAs, the weight one subspace is
\begin{equation}\label{eq:weightOneSubspace}
    (V_L)_1 = \mathcal{H} \oplus \bigoplus_{\alpha^2=2} \bbC\{\Gamma_\alpha\}\,,
\end{equation}
where $\mathcal{H}$ is a Cartan subalgebra spanned by the weight-1 chiral bosons and the $\bbC\{\Gamma_\alpha\}$ play the role of the root spaces. The exact Lie algebra structure constants depend on the choice of $\epsilon$, but all choices are isomorphic to a Lie algebra $\mathfrak{g}$ with root lattice $L$.

Using this, in \cite{dong1999automorphism} the authors prove that $\Aut(V_L)$ is a non-split product
\begin{equation}
    \Aut(V_L) = K.O(\hat{L})\,,
\end{equation}
where $K = \langle \{e^{u_0}: u\in (V_L)_1\} \rangle$ is the inner automorphism group generated by (exponentials of) ``currents.'' In other words, we have that $\mathrm{der}(V_L) := \mathrm{Lie}(\Aut(V_L)) = \mathfrak{g}$, and so we will be interested in finite order automorphisms of $\mathfrak{g}$.\footnote{If we were studying the derivations of a non-rational VOA, whose representations need not form an MTC and thus not be so discrete, its possible that $\mathrm{der}(V)$ isn't entirely formed by currents. In this case we could ask about infinitesimal outer automorphisms $\mathrm{out}(V) := \mathrm{der}(V)/\{\mathrm{currents}\} = \mathrm{der}(\Rep(V))$, which correspond to smooth deformations of the representation category. This is not exotic: consider $n$ free chiral bosons, $V=\mathrm{Bos}(n)$, in this case $\Aut(\Rep(V)) \cong O(n)$.} Note that what we called $T$ before is the maximal toral subgroup, i.e. $T = \langle \{e^{u_0}: u\in \mathcal{H} \} \rangle$.

Beautifully, conjugacy classes of finite order automorphisms of (semi-)simple Lie algebras were classified by Kac \cite{kats1969automorphisms, kac_1990} (also see Section 8.3.3 of \cite{theobook} and \cite{moller2017cyclic} for examples). This includes both inner and outer automorphisms! Even better, the classification gives an easy way to read off the fixed point Lie subalgebra under that symmetry! The theorem is as follows: 

\begin{theorem}[Theorem 8.6, \cite{kac_1990}]\label{th:KacClass}
Let $\mathfrak{g}$ a finite dimensional simple Lie algebra with Dynkin diagram $X_n$, choose $k=1,2,3$ and write $X^{(k)}_n$ for the corresponding affine Dynkin diagram (or ``twisted Dynkin diagram'' if $k>1$) with $\ell+1$ nodes. Then choose non-negative relatively prime integers $s=(s_0,\cdots,s_\ell)$ and set 
\begin{equation}
    m = k \sum_{i=0}^\ell a_i s_i\,,
\end{equation}
where $a_i$ are the marks of $X_n^{(k)}$. Then the following statements are true:
\begin{enumerate}
    \item The choices $(k,s)$ define an order $m$ automorphism $g_{k,s}\in\Aut(\mathfrak{g})$. For $k=1$, if one chooses a set of simple roots and works in the usual Chevalley basis, the action is given by
        \begin{equation}
            g_{1,s}(E^{j}) = e^{2 \pi i s_j/m} E^{j}\,.
        \end{equation}
    For $k=2,3$ the definition of $E^{j}$ is slightly more complicated (see \cite{kac_1990}).
    \item Up to conjugation, all order $m$ automorphisms of $\mathfrak{g}$ are obtained this way.
    \item Two automorphisms $g_{k, s}$ and $g_{k', s'}$ obtained in this way are conjugate by an automorphism of $\mathfrak{g}$ if and only if $k = k'$ and the sequence $s$ can be transformed into $s'$ by an automorphism of the diagram $X^{(k)}_n$.
    \item The number $k$ is the least positive integer for which $g^k_{k,s}$ is an inner automorphism. i.e. $k=1$ automorphisms are inner, and $k=2,3$ automorphisms are outer.
    \item Let $i_1, \cdots , i_p$ be all the indices such that $s_{i_1} = \cdots = s_{i_p} = 0$. Then the fixed-point Lie subalgebra $\mathfrak{g}^{g_{k,s}}$ is isomorphic to a direct sum of: the $(\ell-p)$-dimensional center $\cong \bbC^{\ell-p}$, and a semisimple Lie algebra whose Dynkin diagram is the subdiagram of the affine diagram $X^{(k)}_n$ consisting of the nodes $i_1, \cdots, i_p$.
\end{enumerate}
\end{theorem}

Kac's theorem is great for 3 purposes: enumerating finite order automorphisms of $\mathfrak{g}$, finding the fixed Lie subalgebra $\mathfrak{g}^{g_{k,s}}$, and understanding the systematic construction of $x\in T < \Aut(V_L)$ describing inner automorphisms of lattice VOAs like Equation \eqref{eq:innerAutomorphismT}. In particular, for an order $m$ inner automorphism, $x$ is constructed from the sequence $s$ using the basis of fundamental weights $\omega_1,\dots, \omega_l$, dual to the basis of simple coroots:
\begin{equation}
    x = \frac{1}{m} \sum_{j=1}^l s_j \, \omega_j\,.
\end{equation}



On the other hand, it is systematically easier for calculations of twisted characters of outer automorphisms to use a slightly different procedure from \cite{Hohn:2020xfe}, which we describe in the remainder of this section.

Let $V_L$ be a lattice VOA and fix a choice of Cartan subalgebra $\calH$ of $(V_L)_1$ as in Equation \eqref{eq:weightOneSubspace}, further fix a choice of simple roots $\Delta$. Then it can be shown that: up to conjugacy, every finite-order $\hat{\sigma} \in \Aut(V_L)$ is of the form
\begin{equation}
    \hat{\sigma} = e^{2\pi i (x,\,\cdot\,)}.\hat{g} \label{eq:svensAutomorphisms}
\end{equation}
where $\hat{g}$ is a (fixed) choice of \textit{standard lift} of $g\in H_{\Delta}$ and $x\in L^g \otimes \mathbb{Q}$ is some phase understood modulo $(L^g)^*$. By $H_{\Delta} := O(L)_{\{\Delta\}}$ we mean the setwise stabilizer of $\Delta$ in $O(L)$, note that $H_{\Delta} \cong O(L)/W$ which is in turn isomorphic to the automorphism group of the Dynkin diagram associated to $L$, its just key that we've fixed a base $\Delta$. See \cite{dong1999automorphism} and Theorem 2.13 of \cite{Hohn:2020xfe} for details. 

This theorem allows for a simple computation of the $\hat{\sigma}$-twisted traces 
\begin{equation}
    \tr_{V_L} \hat{\sigma}\, q^{L_0 - \frac{c}{24}} =\frac{1}{\eta_g(\tau)}\sum_{\alpha\in L^g}  e^{2\pi i (x,\alpha)} q^{(\alpha,\alpha)/2}\,. \label{eq:SvensTwisted}
\end{equation}

Importantly, if $\hat{\sigma}$ is built from the standard lift $\hat{g}$ of $g\in O(L)$ as above, then $\hat{\sigma}$ has order $\abs{g}$ if and only if
\begin{equation}
    x\in
    \begin{cases}
        0+(1/m)(L^*)^g & \text{and $\abs{\hat{g}} = m$}\\
        \zeta+(1/m)(L^*)^g & \text{and $\abs{\hat{g}} = 2m$}\\
    \end{cases}
\end{equation}
where $\zeta \in (1/2m) (L^*)^g$.

The short of all this is that: rather than studying finite order automorphisms of $(V_L)_1$ and lifting to automorphisms of the VOA, we can take a more ground-up approach. For example, suppose we want to find all $\hat{\sigma}$-twisted characters of a VOA $V_L$ up to conjugacy, and we want $\hat{\sigma}$ to have order $2$. Then we fix a set of simple roots for $L$ and enumerate all $g\in H_{\Delta}$ of order $1$ or $2$. For each $g$, we make some choice of standard lift $\hat{g} \in \Aut(V_L)$, so $\hat{g}$ has order $1$, $2$, or $4$. Then we enumerate all inner automorphisms of order up to $4$ (by making choices of vectors $x$ as described previously). Finally, we combine all of our choices as in \eqref{eq:svensAutomorphisms}, and throw away whichever do not have order $2$. The theorem guarantees that we have produced all $\hat{\sigma}$ of order $2$ up to conjugacy.

\subsection{Orbifolds of \texorpdfstring{$V_{E_8}$}{VE8}}\label{sec:orbifoldsOfVE8}
Now that we understand automorphisms of lattice VOAs, we may study their associated orbifolds. Throughout, let $V := V_{E_8}$ be the holomorphic VOA\footnote{Also simple, rational, $C_2$-cofinite, self-contregredient, CFT-type. We will ignore such technical details from here out in this section.} constructed from the $E_8$ root lattice $L$ as in Example \ref{ex:VE8}. $V$ is the unique (up to isomorphism) holomorphic CFT with chiral central charge $c_L = 8$, and gauging by any non-anomalous $\bbZ_m$ symmetry will simply produce another holomorphic CFT with the same central charge, so $[V/\bbZ_m] \cong V$.

For this reason, computing orbifold partition functions isn't particularly illuminating, they are all just as in Equation \eqref{eq:E8partition}. However, while the orbifold theory $[V/\bbZ_m]$ is isomorphic to $V$, they are not identical (just as in the Ising model). To appreciate this, we should look at what it means to orbifold $V$ at the level of lattices.

For this, we will follow the conventions of \cite{conway2013sphere}. In particular, $L$ will be spanned by simple roots
\begin{equation}
    \begin{aligned}
    \alpha_1&=(-1,1,0^6),\\
    \alpha_2&=(0,-1,1,0^5),\\
    &\vdots\\
    \alpha_7&=(0^6,-1,1),\\
    \alpha_8&=\qty(\frac{1}{2}^5,-\frac{1}{2}^3),
    \end{aligned}
\end{equation}
which fit into the Dynkin diagram
\begin{equation}
\begin{tikzpicture}[baseline={(current bounding box.center)}, scale = 1]
        \draw (0,0) node {$\alpha_1$};
        \draw (0.3,0) -- (0.8,0) node [anchor = west] {$\alpha_2$};
        \draw (1.5,0) -- (2,0) node [anchor = west] {$\alpha_3$};
        \draw (2.7,0) -- (3.2,0) node [anchor = west] {$\alpha_4$};
        \draw (3.9,0) -- (4.4,0) node [anchor = west] {$\alpha_5$};
        \draw (5.1,0) -- (5.6,0) node [anchor = west] {$\alpha_6$};
        \draw (6.3,0) -- (6.8,0) node [anchor = west] {$\alpha_7$};
        \draw (4.75,0.3) -- (4.75,0.8) node [anchor = south] {$\alpha_8$}; 
        \end{tikzpicture}\,.
\end{equation}
This diagram can be extended to its affine version with the help of the highest root
\begin{equation}\label{eq:highest_e8}
    \alpha_{\text{high}}=2\alpha_1+3\alpha_2+4\alpha_3+5\alpha_4+6\alpha_5+4\alpha_6+2\alpha_7+3\alpha_8\,,
\end{equation}
which yields
\begin{equation}\label{eq:affine_e8}
\begin{tikzpicture}[baseline={(current bounding box.center)}, scale = 1]
        \draw (0.3,0) node [anchor=east] {$(-\alpha_{\text{high}})$};
        \draw (0.3,0) -- (0.8,0) node [anchor = west] {$\alpha_1$};
        \draw (1.5,0) -- (2,0) node [anchor = west] {$\alpha_2$};
        \draw (2.7,0) -- (3.2,0) node [anchor = west] {$\alpha_3$};
        \draw (3.9,0) -- (4.4,0) node [anchor = west] {$\alpha_4$};
        \draw (5.1,0) -- (5.6,0) node [anchor = west] {$\alpha_5$};
        \draw (6.3,0) -- (6.8,0) node [anchor = west] {$\alpha_6$};
        \draw (7.5,0) -- (8,0) node [anchor = west] {$\alpha_7$};
        \draw (5.95,0.3) -- (5.95,0.8) node [anchor = south] {$\alpha_8$}; 
        \end{tikzpicture}\,.
\end{equation}

Now, let $G = \langle g_{s} \rangle \cong \bbZ_m$ be a subgroup of $T < \Aut(V)$ generated by the sequence $s$ as in the previous section (and with $k=1$) i.e. $G$ acts by a phase on the group algebra. Explicitly, write $x := \sum x_j \omega_j$ with $x_j := s_j/m$, then
\begin{equation}
\begin{aligned}
    a^i_n &\mapsto a^i_n \,, \\
    \Gamma_\alpha &\mapsto e^{2\pi i (x,\alpha)} \Gamma_\alpha \,, \label{eq:innerAutAction}
\end{aligned}
\end{equation}
with $(x,\alpha_i) \geq 0$, $(x, \alpha_{\mathrm{high}}) \leq 1$, and $x_i \in \frac{1}{m}\bbZ$. 

Such a symmetry is non-anomalous and can be orbifolded if we can extend the invariant sub-VOA $V^G$ by twisted sectors to a distinct VOA.\footnote{Note: If $V$ is a lattice VOA, then under an inner automorphism the fixed sub-VOA $V^G$ is always another lattice VOA.} In physics terms, the orbifold Hilbert space is obtained as an extension of the subspace of unscreened/invariant operators, as in Equation \eqref{eq:HOrbiUntwisted} for a $\bbZ_2$ symmetry. The anomaly classifies the obstruction to this extension procedure. 

The group action above leaves the chiral Cartan bosons unscreened as well as the vertex operators with $(x,\alpha) \in \bbZ$. This last condition distinguishes a sublattice $L_0 \subset L$ of index $m$, so that the fixed sub-VOA $V^G$ is the lattice VOA constructed from $L_0$. Said in reverse, $\rk(L_0)=\rk(L)$ both define the same Heisenberg module of chiral bosons, but in $V^G$ only the momentum states corresponding to $L_0$ are admissible.\footnote{Note: in the mathematics literature, $V^G$ is often called the ``orbifold VOA,'' but we will always mean the full extension.}

Our question then becomes: can we extend $L_0$ to an even self-dual lattice $L^\prime$, distinct from $L$, inside $L_0^*$?\footnote{Dropping the even requirement is how one produces fermionic orbifolds.} This can be depicted
\begin{equation}
    \begin{tikzpicture}[scale = 1, baseline = -2.5]
        \node at (-1,0) {$L_0$};
        \node at (-0.5,0.5) {\rotatebox{45}{$\subseteq$}};
        \node at (-0.5,-0.5) {\rotatebox{-45}{$\subseteq$}};
        \node at (0.5,0.8) {$L=L^*$};
        \node at (0.5,-0.8) {$L^\prime = L^{\prime*}$};
        \node at (1.5,-0.5) {\rotatebox{45}{$\subseteq$}};
        \node at (1.5,0.5) {\rotatebox{-45}{$\subseteq$}};
        \node at (2,0) {$L_0^*$};
    \end{tikzpicture}\,\,.
\end{equation}

Such an $L^\prime$ exists (and is unique) if and only if
\begin{equation}
    q(x) := \frac{(x,x)}{2} \in \frac{1}{m}\bbZ\,.\label{eq:disc}
\end{equation}
Generally, $q(x) \in \frac{1}{m^2}\bbZ$ because $m x \in L$, so the anomaly class is the class
\begin{equation}
    [q(x)] \in \left(\frac{1}{m}\bbZ\right) / \left(\frac{1}{m^2}\bbZ\right) \cong \bbZ_m\,.
\end{equation}
Note: $H^3(\bbZ_m, U(1)) = \bbZ_m$, so the matchup with anomalies is evident.

\begin{table}
\centering
\begin{tabular}{c|*4c}
    \toprule
            & $[V/\bbZ_m]$ & $[V/\bbZ_m]^1$ & $\cdots$ & $[V/\bbZ_m]^{m-1}$  \\\midrule
    $V$     & $\calH^0_0$ & $\calH^1_0$ & $\cdots$ & $\calH^{m-1}_0$  \\
    $V_1$ & $\calH^0_1$ & $\calH^1_1$ & $\cdots$ & $\calH^{m-1}_1$  \\
    $\vdots$ & $\vdots$ & $\vdots$ & $\ddots$ & $\vdots$    \\
    $V_{m-1}$ & $\calH^0_{m-1}$ & $\calH^1_{m-1}$ & $\cdots$ & $\calH^{m-1}_{m-1}$  \\
    \bottomrule
    \hline
\end{tabular}
\caption{The holomorphic VOA $V$ decomposes as a direct sum of the untwisted Hilbert spaces, where the upper grading denotes the ``electric'' $\bbZ_m$ charge. The orbifold VOA decomposes as a direct sum of the gauge invariant Hilbert spaces, where the lower index grading denotes the dual ``magnetic'' $\hat{\bbZ}_m$ charge.}\label{tab:RepVG}
\end{table}

In the language of VOAs we are just saying that the irreducible representations $\Irr(V^G)$ have group like fusion, with a fusion group given by a central extension of $\bbZ_m$ by $\bbZ_m$, reflecting the structure of $L_0^*/L_0$ when we try to factor through $L$. The structure of the fusion group is controlled by a class living in $H^2(\bbZ_m,\hat{\bbZ}_m) = \bbZ_m$, but should not be confused with the group of anomalies. The anomaly does influence this group, but the map between anomalies and fusion rules is not 1-to-1 (see \cite{van_Ekeren_2020}), as can be seen in following examples with $\bbZ_2$. However, in the case that $G$ acts non-anomalously, then $\Irr(V^G) = \bbZ_m \times \bbZ_m$ for sure.

If we suggestively label the elements of $\Irr(V^G)$ as $\mathcal{H}_j^i$, mimicking the decomposition following Equation \eqref{eq:HUntwisted}, then $\Irr(V^G)$ can be organized like Table \ref{tab:RepVG}. The original VOA decomposes into the ``electrically charged'' sectors
\begin{equation}
    V = \bigoplus_{i=1}^m \mathcal{H}_0^i\,,
\end{equation}
and the orbifold VOA, if it exists, decomposes into the ``magnetically charged'' sectors as
\begin{equation}
    [V/\bbZ_m] = \bigoplus_{j=1}^m \mathcal{H}_j^0\,.
\end{equation}
Each of the $\calH^i_j \in \Irr(V^G)$ is associated with a coset of $L_0^*/L_0$, and the discriminant $q(x)$ is simply telling us the conformal weight in that sector. In the non-anomalous case, the conformal weight of a vector in $\calH_j^i$ is $\frac{ij}{m} + \bbZ$ which gives $\Irr(V^G)$ the structure of a ``metric Abelian group'' (see Section 4 of \cite{moller2017cyclic} for more details).

We record the automorphisms of $V$, anomaly, and corresponding Lie algebra for small $m$ in Table \ref{tab:allsym}. In Table \ref{tab:nonanomsym} we record just the non-anomalous symmetries with the marked Dynkin diagram.

\begin{table}[t]
\centering
\begin{tabular}{>{\centering\arraybackslash} m{2cm} >{\centering\arraybackslash} m{5cm} >{\centering\arraybackslash} m{3cm} >{\centering\arraybackslash} m{3cm}}    \toprule
    \emph{Order} & \emph{Automorphism} & $\bbZ_m$ \emph{Anomaly} & \emph{Fixed Subalgebra} \\
    \midrule
    \multirow{2}{*}{2} 
        & (0,0,0,0,0,0,0,1,0) & 0 & $D_8$ \\
        & (0,1,0,0,0,0,0,0,0) & 1 & $A_1 \times E_7$ \\
    \midrule
    \multirow{4}{*}{3}    
        & (0,0,0,0,0,0,0,0,1) & 1 & $A_8$ \\
        & (1,0,0,0,0,0,0,1,0) & 2 & $D_7 \times \mathfrak{u}(1)$\\
        & (0,0,1,0,0,0,0,0,0) & 0 & $A_2 \times E_6$ \\
        & (1,1,0,0,0,0,0,0,0) & 1 & $E_7 \times \mathfrak{u}(1)$ \\
    \midrule
    \multirow{7}{*}{4}    
        & (1,0,0,0,0,0,0,0,1) & 0 & $A_7 \times \mathfrak{u}(1)$ \\
        & (2,0,0,0,0,0,0,1,0) & 2 & $D_7 \times \mathfrak{u}(1)$ \\
        & (0,0,0,0,0,0,1,0,0) & 3 & $A_1 \times A_7$ \\
        & (0,0,0,1,0,0,0,0,0) & 2 & $A_3 \times D_5$ \\
        & (1,0,1,0,0,0,0,0,0) & 3 & $A_1 \times E_6 \times \mathfrak{u}(1)$ \\
        & (2,1,0,0,0,0,0,0,0) & 1 & $E_7 \times \mathfrak{u}(1)$ \\
        & (0,1,0,0,0,0,0,1,0) & 1 & $A_1 \times D_6 \times \mathfrak{u}(1)$ \\
    \bottomrule
    \multirow{14}{*}{5}    
        & (2,0,0,0,0,0,0,0,1) & 4 & $A_7 \times \mathfrak{u}(1)$ \\
        & (3,0,0,0,0,0,0,1,0) & 2 & $D_7 \times \mathfrak{u}(1)$ \\
        & (0,0,0,0,0,0,0,1,1) & 1 & $A_7 \times \mathfrak{u}(1)$ \\
        & (1,0,0,0,0,0,0,2,0) & 3 & $D_7 \times \mathfrak{u}(1)$ \\
        & (1,0,0,0,0,0,1,0,0) & 2 & $A_1 \times A_6 \times \mathfrak{u}(1)$ \\
        & (0,0,0,0,1,0,0,0,0) & 0 & $A_4 \times A_4$ \\
        & (1,0,0,1,0,0,0,0,0) & 1 & $A_2 \times D_5 \times \mathfrak{u}(1)$ \\
        & (2,0,1,0,0,0,0,0,0) & 3 & $A_1 \times E_6 \times \mathfrak{u}(1)$ \\
        & (0,0,1,0,0,0,0,1,0) & 4 & $A_2 \times D_5 \times \mathfrak{u}(1)$ \\
        & (3,1,0,0,0,0,0,0,0) & 1 & $E_7 \times \mathfrak{u}(1)$ \\
        & (0,1,0,0,0,0,0,0,1) & 3 & $A_1 \times A_6 \times \mathfrak{u}(1)$ \\
        & (1,1,0,0,0,0,0,1,0) & 0 & $D_6 \times \mathfrak{u}(1)^{2} $ \\
        & (0,1,1,0,0,0,0,0,0) & 2 & $A_1 \times E_6 \times \mathfrak{u}(1)$ \\
        & (1,2,0,0,0,0,0,0,0) & 4 & $E_7 \times \mathfrak{u}(1)$ \\
    \bottomrule
    \hline
\end{tabular}
\caption{The finite order automorphisms of the $E_8$ Lie algebra can be enumerated, up to conjugacy, by Theorem \ref{th:KacClass}. The automorphism is given by a sequence $s$ which describes the automorphism. In physics terms, $s$ describes how the operators of $V_{E_8}$ are screened and the anomaly of that symmetry.}\label{tab:allsym}
\end{table}

\begin{table}[t]
\centering
\begin{tabular}{>{\centering\arraybackslash} m{2cm} >{\centering\arraybackslash} m{8cm} >{\centering\arraybackslash} m{3cm}}    \toprule
    \emph{Order} & \emph{Automorphism} & \emph{Fixed Subalgebra} \\\midrule
        \multirow{1}{*}{2}    
        & \begin{tikzpicture}
        \draw (0,0) node {0};
        \draw (0.2,0) -- (0.7,0) node [anchor = west] {0};
        \draw (1.2,0) -- (1.7,0) node [anchor = west] {0};
        \draw (2.2,0) -- (2.7,0) node [anchor = west] {0};
        \draw (3.2,0) -- (3.7,0) node [anchor = west] {0};
        \draw (4.2,0) -- (4.7,0) node [anchor = west] {0};
        \draw (5.2,0) -- (5.7,0) node [anchor = west] {0};
        \draw (6.2,0) -- (6.7,0) node [anchor = west] {1};
        \draw (4.95,0.3) -- (4.95,0.8) node [anchor = south] {0}; 
        \end{tikzpicture}  
        & $D_8$      \\
        \midrule
        \multirow{1}{*}{3}    
        & \begin{tikzpicture}
        \draw (0,0) node {0};
        \draw (0.2,0) -- (0.7,0) node [anchor = west] {0};
        \draw (1.2,0) -- (1.7,0) node [anchor = west] {1};
        \draw (2.2,0) -- (2.7,0) node [anchor = west] {0};
        \draw (3.2,0) -- (3.7,0) node [anchor = west] {0};
        \draw (4.2,0) -- (4.7,0) node [anchor = west] {0};
        \draw (5.2,0) -- (5.7,0) node [anchor = west] {0};
        \draw (6.2,0) -- (6.7,0) node [anchor = west] {0};
        \draw (4.95,0.3) -- (4.95,0.8) node [anchor = south] {0}; 
        \end{tikzpicture}  
        & $A_2 \times E_6$      \\
        \midrule
        \multirow{1}{*}{4}    
        & \begin{tikzpicture}
        \draw (0,0) node {1};
        \draw (0.2,0) -- (0.7,0) node [anchor = west] {0};
        \draw (1.2,0) -- (1.7,0) node [anchor = west] {0};
        \draw (2.2,0) -- (2.7,0) node [anchor = west] {0};
        \draw (3.2,0) -- (3.7,0) node [anchor = west] {0};
        \draw (4.2,0) -- (4.7,0) node [anchor = west] {0};
        \draw (5.2,0) -- (5.7,0) node [anchor = west] {0};
        \draw (6.2,0) -- (6.7,0) node [anchor = west] {0};
        \draw (4.95,0.3) -- (4.95,0.8) node [anchor = south] {1}; 
        \end{tikzpicture}  
        & $A_7 \times \mathfrak{u}(1)$      \\
        \midrule
        \multirow{2}[12]{*}{5}    
        & \begin{tikzpicture}
        \draw (0,0) node {0};
        \draw (0.2,0) -- (0.7,0) node [anchor = west] {0};
        \draw (1.2,0) -- (1.7,0) node [anchor = west] {0};
        \draw (2.2,0) -- (2.7,0) node [anchor = west] {0};
        \draw (3.2,0) -- (3.7,0) node [anchor = west] {1};
        \draw (4.2,0) -- (4.7,0) node [anchor = west] {0};
        \draw (5.2,0) -- (5.7,0) node [anchor = west] {0};
        \draw (6.2,0) -- (6.7,0) node [anchor = west] {0};
        \draw (4.95,0.3) -- (4.95,0.8) node [anchor = south] {0}; 
        \end{tikzpicture}  
        & $A_4 \times A_4$      \\
        & \begin{tikzpicture}
        \draw (0,0) node {1};
        \draw (0.2,0) -- (0.7,0) node [anchor = west] {1};
        \draw (1.2,0) -- (1.7,0) node [anchor = west] {0};
        \draw (2.2,0) -- (2.7,0) node [anchor = west] {0};
        \draw (3.2,0) -- (3.7,0) node [anchor = west] {0};
        \draw (4.2,0) -- (4.7,0) node [anchor = west] {0};
        \draw (5.2,0) -- (5.7,0) node [anchor = west] {0};
        \draw (6.2,0) -- (6.7,0) node [anchor = west] {1};
        \draw (4.95,0.3) -- (4.95,0.8) node [anchor = south] {0}; 
        \end{tikzpicture}  
        & $D_6 \times \mathfrak{u}(1)^2$      \\
    \bottomrule
    \hline
\end{tabular}
\caption{List of non-anomalous symmetries of the $E_8$ Lie algebra and fixed point Lie subalgebras up to order $m=5$. Instead of listing the corresponding sequence $s$ for the automorphism, we mark $s$ onto the (affine) $E_8$ Dynkin diagram, illustrating the origin of the fixed Lie subalgebras.}\label{tab:nonanomsym}
\end{table}

\subsubsection{Example: \texorpdfstring{$\bbZ_2$}{Z2} Orbifold of \texorpdfstring{$V_{E_8}$}{VE8}}\label{sec:Z2Orbifold}
Let $V:=V_{E_8}$ throughout. By Kac's theorem and the discussion of Section \ref{sec:Automorphisms}, $V$ has two automorphisms of order $2$ up to conjugacy (see Table \ref{tab:allsym}). Call the non-anomalous one 2A and the anomalous one 2B, they are generated by
\begin{align}
    x^A &:= \frac{1}{2}\omega_7\\
    x^B &:= \frac{1}{2}\omega_1\,.
\end{align}

Let us start with the 2A symmetry. In this case, our fixed sub-VOA is $V^{\bbZ_2^A} = V_{D_8}$. To see this at the level of lattices, first note $(x^A,\alpha_i)=0$ for all $i$ except $i=7$, so the only nontrivial action is $\Gamma_{\alpha_7}\mapsto -\Gamma_{\alpha_7}$. Vertex operators like $\Gamma_{2\alpha_7}$ do remain invariant though, so the invariant VOA is the lattice VOA constructed from the span of $\{\alpha_1,\dots,\alpha_6,2\alpha_7,\alpha_8\}$. The reason why Kac's theorem works is because of the form of the highest root $\alpha_{\mathrm{high}}$ in Equation \eqref{eq:highest_e8}, which guarantees that the fixed sublattice is also spanned by $\{\alpha_1,\dots,\alpha_6,-\alpha_{\mathrm{high}},\alpha_8\}$. It is then clear from the way these vectors fit inside of the affine Dynkin diagram in Equation \eqref{eq:affine_e8} that they span a $D_8$ lattice.

Since 2A is non-anomalous, the preceding discussion tells us that the fusion group is a metric Abelian group $\Irr(V^{\bbZ_2^A}) = \bbZ_2 \times \bbZ_2$ with two isotropic subgroups. To see this at the level of lattices, we note that the $E_8$ lattice can be presented as the span of the rows of
\begin{equation}
    \begin{pmatrix}
        2 & 0 & 0 & 0 & 0 & 0 & 0 & 0\\
        -1 & 1 & 0 & 0 & 0 & 0 & 0 & 0\\
        0 & -1 & 1 & 0 & 0 & 0 & 0 & 0\\
        0 & 0 & -1 & 1 & 0 & 0 & 0 & 0\\
        0 & 0 & 0 & -1 & 1 & 0 & 0 & 0\\
        0 & 0 & 0 & 0 & -1 & 1 & 0 & 0\\
        0 & 0 & 0 & 0 & 0 & -1 & \,\,1\,\, & 0\\
        \tfrac12 & \tfrac12 & \tfrac12 & \tfrac12 & \tfrac12 & \tfrac12 & \tfrac12 & \,\,\tfrac12\,\,
    \end{pmatrix}\,,
\end{equation} 
as in Chapter 4 of \cite{conway2013sphere} (these are not roots). Similarly, the $D_8$ lattice and $D_8^*$ lattice can be written as the span of the rows of
\begin{equation}
    \begin{pmatrix}
        -1 & -1 & 0 & 0 & 0 & 0 & 0 & 0\\
        1 & -1 & 0 & 0 & 0 & 0 & 0 & 0\\
        0 & 1 & -1 & 0 & 0 & 0 & 0 & 0\\
        0 & 0 & 1 & -1 & 0 & 0 & 0 & 0\\
        0 & 0 & 0 & 1 & -1 & 0 & 0 & 0\\
        0 & 0 & 0 & 0 & 1 & -1 & 0 & 0\\
        0 & 0 & 0 & 0 & 0 & 1 & -1 & 0\\
        0 & 0 & 0 & 0 & 0 & 0 & 1 & -1
    \end{pmatrix}\quad\textrm{and}\quad
    \begin{pmatrix}
       \,\,1\,\, & 0 & 0 & 0 & 0 & 0 & 0 & 0\\
        0 & \,\,1\,\, & 0 & 0 & 0 & 0 & 0 & 0\\
        0 & 0 & \,\,1\,\, & 0 & 0 & 0 & 0 & 0\\
        0 & 0 & 0 & \,\,1\,\, & 0 & 0 & 0 & 0\\
        0 & 0 & 0 & 0 & \,\,1\,\, & 0 & 0 & 0\\
        0 & 0 & 0 & 0 & 0 & \,\,1\,\, & 0 & 0\\
        0 & 0 & 0 & 0 & 0 & 0 & \,\,1\,\, & 0\\
        \tfrac12 & \tfrac12 & \tfrac12 & \tfrac12 & \tfrac12 & \,\,\tfrac12\,\, & \tfrac12 & \,\,\tfrac12\,\,
    \end{pmatrix}\
\end{equation}
respectively. Conveniently for us, \cite{conway2013sphere} also records the following elements of $D_8^*$ which are representatives of separate classes in $D_8^*/D_8$:
\begin{align}
    [0] &:= (0,0,\dots,0)\,,\\
    [1] &:= (\tfrac12,\tfrac12,\dots,\tfrac12)\,,\\
    [2] &:= (0,0,\dots,1)\,,\\
    [3] &:= (\tfrac12,\tfrac12,\dots,-\tfrac12)\,.
\end{align}
From this we see that $D_8^*/D_8 \cong \bbZ_2\times\bbZ_2$ with decomposition and discriminant
\begin{equation}
D_8^* =
    \begin{tabular}{|c|c|}
        \hline
        $D_8 + [0]$ & $D_8 + [1]$\\
        \hline
        $D_8 + [3]$ & $D_8 + [2]$\\
        \hline
    \end{tabular}
    \quad\stackrel{\mathrm{Weight}}{\rightsquigarrow}\quad
    q^A(x)=\begin{tabular}{|c|c|}
        \hline
        $0$ & $0$\\
        \hline
        $0$ & $\tfrac{1}{2}$\\
        \hline
    \end{tabular}\,.
\end{equation}
Note that $(D_8 + [0]) \cup (D_8 + [1])$ and $(D_8 + [0]) \cup (D_8 + [3])$ both give $E_8$ lattices in $D_8^*$, and that they're not the same (their intersection is only a $D_8$ lattice).

The characters can be obtained as in Section \ref{ex:VE8} and are
\begin{align}
    \tr_{V_{D_8 + [0]}} q^{L_0-\frac{c}{24}}
        &=\frac{1}{2\eta(\tau)^8}\qty(\theta_3^8(q) + \theta_4^8(q))\label{eq:ch_D8}\,,\\
    \tr_{V_{D_8 + [1]}} q^{L_0-\frac{c}{24}}
        &= \frac{1}{2\eta(\tau)^8}\qty(\theta_2^8(q) + \theta_1^8(q))\,,\\
    \tr_{V_{D_8 + [2]}} q^{L_0-\frac{c}{24}}
        &= \frac{1}{2\eta(\tau)^8}\qty(\theta_3^8(q) - \theta_4^8(q))\,,\\
    \tr_{V_{D_8 + [3]}} q^{L_0-\frac{c}{24}}
        &= \frac{1}{2\eta(\tau)^8}\qty(\theta_2^8(q) - \theta_1^8(q))\,.
\end{align}
Since $\theta_1(\tau) = 0$, we can clearly see how to decompose $V$ into $V_{D_8}$ modules in two distinct ways just by comparing twisted partition functions. 

Looking at the 2B symmetry, $V^{\bbZ_2^B} = V_{A_1\times E_7}$ and has fusion group $\bbZ_2 \times \bbZ_2$ as well. We may write the discriminant in this case and it's
\begin{equation}
q^B(x)=\begin{tabular}{|c|c|}
        \hline
        $0$ & $0$\\
        \hline
        $\tfrac{3}{4}$ & $\tfrac{1}{4}$\\
        \hline
    \end{tabular}\,.
\end{equation}
As expected, the only isotropic subgroup corresponds to the decomposition we started with initially.

It would be remiss not to point out that $q^A(x)$ and $q^B(x)$ are the spins of simple anyons in the $\bbZ_2$ gauge theory and twisted $\bbZ_2$ gauge theory (also called semion-antisemion) respectively. 

Generally, since $V$ is well-behaved, $\Rep(V^{\bbZ_m})$ is a modular tensor category (this follows rigorously from technical results of \cite{miyamoto2004uniform, miyamoto2015c_2, carnahan2016regularity}) whose simple objects, the irreducible modules $\calH_{n_m}^{n_e}$, correspond to simple anyons in a (2+1)d TFT which is interpretable as a (possibly twisted) $\bbZ_m$ gauge theory. If we couple $V^{\bbZ_m}$ to the bulk TFT as a 2d boundary theory, then the anyon associated to $\calH_{n_m}^{n_e}$ may terminate on operators from the corresponding module. 
We will elaborate on these (2+1)d points further in Section \ref{sec:3dTFT} (see also \cite{gaiotto2020orbifold}).

\subsection{Computing Defected Partition Functions in \texorpdfstring{$(E_8)_1$}{(E8)1}}\label{sec:DefectedPFs}
In this section we use our general results about automorphisms of lattice VOAs to explain how to compute $\bbZ_m$ ``duality defected'' partition functions in the chiral $(E_8)_1$ WZW model, i.e. we compute twisted characters for $\TY(\bbZ_m)$ actions on the lattice VOA $V := V_{E_8}$.

Suppose we want to find partition functions with a duality defect line $X$ twisting the Euclidean time direction, $Z_V[0,X]$, separating $V$ from some $\bbZ_m$ orbifold $[V/\bbZ_m]$. In practice, we proceed as follows for order $m$ symmetries:
\begin{enumerate}
    \item Following Kac's theorem, we find all inner automorphisms of the $E_8$ Lie algebra of order $m$. Using Equation \eqref{eq:disc}, we compute which symmetries of $V$ are non-anomalous and thus can actually be gauged. Note: there may be multiple conjugacy classes of non-anomalous symmetries e.g. for $m=5$ there are two (see Table \ref{tab:nonanomsym}).
    \item Taking one of the non-anomalous symmetry groups $G \cong \bbZ_m$, the fixed-point sub-VOA $V^G$ is a lattice VOA and is obtained as described in Section \ref{sec:orbifoldsOfVE8}. The irreps $A := \Irr(V^G)$ form a metric Abelian group $(A,h)$, with $A \cong \bbZ_m \times \hat{\bbZ}_m$ and metric function $h(i,j) = ij/m$. $V$ and $[V/G]$ decompose as direct sums over the isotropic subgroups with irrep labels $(i,0)$ and $(0,j)$ respectively, i.e.
        \begin{align}
            V &= \bigoplus_{i} \calH_0^i\,,\\
            [V/G] &= \bigoplus_{j} \calH_j^0\,.
        \end{align}
    \item Finding automorphisms of $V^G$ which ``swap the axes'' corresponding to $V$ and $[V/G]$ in $A$, we obtain the defected partition functions as the \textit{twisted} characters for this second automorphism.
\end{enumerate}
We've already seen the first two points in detail, so we elaborate on the final point.


Write $O(A,h)$ for the group of automorphisms of $A$ which preserve $h$. For simplicity, consider the case that $m$ is prime, then we have that
\begin{equation}
    O(A,h) = \bbZ_m^\times \rtimes \bbZ_2\,.
\end{equation}
The $\bbZ_m^\times$ factor corresponds to the freedom in redefining the action of $G$ on $V$ and is the normal subgroup $SO(A,h)$ which preserves the axes of $A$. This leaves $\bbZ_2 = O(A,h)/SO(A,h)$ to act by swapping the axes of $A$. The $\bbZ_2$ factor acts by electric-magnetic duality in the (2+1)d TFT. We can work this group out for more general $m$, but generally $O(A,h)$ is just the group of symmetries of $\bbZ_m$ Dijkgraaf-Witten theory \cite{DW90} aka the group of braided auto-equivalences of $\calZ(\Vc_G)$ (see \cite{gaiotto2020orbifold, jaumeDiego} for physics discussions, and \cite{homotopyFusion, FPSV:brauerGroupAbelDW, nikshychRiepel:catLagGrass} for mathematics discussions).

Now, any automorphism of $V^G$ induces an automorphism of $A$ which preserves the conformal weight $h$. Note: many elements of $\Aut(V^G)$ induce the same automorphism in $O(A,h)$.

Our principal claim for this paper is that: \textit{$(V^{G})^{\bbZ_2} = V^{\TY(G)}$ if and only if the $\bbZ_2$ action on $V^{G}$ switches the axes of $A$}. We prove this in Theorem \ref{prop:proofTY}. At the level of partition functions, this means if $\hat{\sigma}\in \Aut(V^G)$ induces a $\bbZ_2$ action on $A$ which ``swaps the axes'' of $\Irr(V^G)$, then there is a $\bbZ_m$ Tambara-Yamagami line $X_{\hat{\sigma}}$ in $V$, whose partition function is the $\hat{\sigma}$-twisted character of $V^G$
\begin{equation}\label{eq:DualityDefectPF}
    Z_V[0,X_{\hat{\sigma}}] = \tr_{V^G} \hat{\sigma} q^{L_0-\frac{c}{24}} \,.
\end{equation}

In practice, this is where the technical details about lifts becomes relevant. We will identify $\bbZ_2$ symmetries of the underlying lattices associated to $V^G$ which should switch the axes, and then consider their lifts to the VOA. This means even after we have picked a non-anomalous $\bbZ_m$ symmetry and a $\bbZ_2$ which switches the axes, there is \textit{still} degeneracy in the partition function depending on how the $\bbZ_2$ lifts to a VOA automorphism, leading to different TY-lines.

Comparing to the end of Section \ref{sec:IsingCFT}, we see a 2d shadow of this equation. We already saw that the duality defected partition function involves only uncharged operators, and that we just had to resolve a sign in the partition function which tracks ``if $V$ and $[V/G]$ see the primary $\phi$ differently.'' This role of the $\hat{\sigma}$-twisted character is to track this sign.

\subsection{\texorpdfstring{$\bbZ_2$}{Z2} Duality Defects}\label{sec:Z2DDs}
Given $V := V_{E_8}$, there is one non-anomalous $\bbZ_2$ symmetry up to conjugation, and there are 4 duality defected partition functions, i.e. $\bbZ_2$ Tambara-Yamagami lines or duality defects. In this section we will compute them in detail using our VOA-theoretic techniques and then compare to the result using fermionization as in \cite{Thorngren:2018bhj, Ji:2019ugf, MonsterCFT}.

\subsubsection{\texorpdfstring{$\bbZ_2$}{Z2} Duality Defects from Lie Theory}\label{sec:Z2DDfromLieTheory}
We start by following our general procedure outlined in the previous section. There is one non-anomalous $\bbZ_2$ action on $V$, generated by $x^A = \frac{1}{2}\omega^7$ in the sense of Equation \eqref{eq:innerAutAction}, fixing the subalgebra $V^{\bbZ_2^A} \cong V_{D_8}$. We found this by taking the affine $E_8$ Dynkin diagram and ``chopping'' the 7th root to produce the $D_8$ Dynkin diagram.

We know from our discriminant calculations in Section \ref{sec:Z2Orbifold} that there are two ways to extend the $D_8$ root lattice into the $E_8$ lattice. In fact, we can see these two extensions directly from the Dynkin diagram
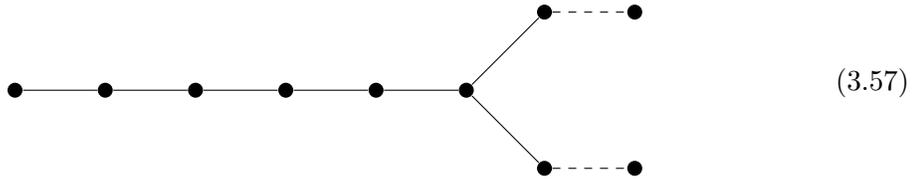
\begin{equation}
\begin{tikzpicture}[baseline={(current bounding box.center)}, scale = 1]
	\tikzstyle{vertex}=[circle, fill=black, minimum size=2pt,inner sep=2pt];
	\def\r{1.2};
	\node[vertex] (T1) at (\r*1,\r*0) {};
	\node[vertex] (T2) at (\r*2,\r*0) {};
	\node[vertex] (T3) at (\r*3,\r*0) {};
	\node[vertex] (T4) at (\r*4,\r*0) {};
	\node[vertex] (T5) at (\r*5,\r*0) {};
	\node[vertex] (T6) at (\r*6,\r*0) {};
	\node[vertex] (T7) at (\r*6.866,\r*0.866) {};
	\node[vertex] (T8) at (\r*6.866,-1*\r*0.866) {};
	\node[vertex] (T72) at (\r*7.866,\r*0.866) {};
	\node[vertex] (T82) at (\r*7.866,-1*\r*0.866) {};
	
	\draw[-] (T1) -- (T2);
	\draw[-] (T2) -- (T3);
	\draw[-] (T3) -- (T4);
	\draw[-] (T4) -- (T5);
	\draw[-] (T5) -- (T6);
	\draw[-] (T6) -- (T7);
	\draw[-] (T6) -- (T8);
	\draw[-, dashed] (T7) -- (T72);
	\draw[-, dashed] (T8) -- (T82);
\end{tikzpicture}
\end{equation}

To obtain a TY-category, we want a $\bbZ_2$ action on $V^{\bbZ_2}$ which swaps the axes of $\mathrm{Irr}(V^{\bbZ_2})$. At the level of root lattices, this means swapping the two $E_8$'s in $D_8^*$, and the $\bbZ_2$ which does this is obvious: the $D_8$ Dynkin diagram automorphism.

We start by looking at the conjugacy classes of $\bbZ_2$ outer automorphisms of $\mathfrak{so}(16)$, which all come from the $D_8$ Dynkin diagram automorphism. Using Kac's theorem with $m=2$ and $k=2$, we need to study the Dynkin diagram $D_8^{(2)}$. This is given (with its marks) by

\begin{equation}
\begin{tikzpicture}[baseline={(current bounding box.center)}, scale = 1]
	\tikzstyle{vertex}=[circle, fill=black, minimum size=2pt,inner sep=2pt];
	\def\r{1.2};
	\node[vertex] (T0) at (\r*0,\r*0) {};
	\node[vertex] (T1) at (\r*1,\r*0) {};
	\node[vertex] (T2) at (\r*2,\r*0) {};
	\node[vertex] (T3) at (\r*3,\r*0) {};
	\node[vertex] (T4) at (\r*4,\r*0) {};
	\node[vertex] (T5) at (\r*5,\r*0) {};
	\node[vertex] (T6) at (\r*6,\r*0) {};
	\node[vertex] (T7) at (\r*7,\r*0) {};
	
	\draw[above] (T0) node {$1$};
	\draw[above] (T1) node {$1$};
	\draw[above] (T2) node {$1$};
	\draw[above] (T3) node {$1$};
	\draw[above] (T4) node {$1$};
	\draw[above] (T5) node {$1$};
	\draw[above] (T6) node {$1$};
	\draw[above] (T7) node {$1$};
	
	\draw[below] (T0) node {$\alpha_0$};
	\draw[below] (T1) node {$\alpha_1$};
	\draw[below] (T2) node {$\alpha_2$};
	\draw[below] (T3) node {$\alpha_3$};
	\draw[below] (T4) node {$\alpha_4$};
	\draw[below] (T5) node {$\alpha_5$};
	\draw[below] (T6) node {$\alpha_6$};
	\draw[below] (T7) node {$\alpha_7$};
	
	\draw[-r- = 0.40 rotate 180, double,double distance = 0.05cm] (T0) -- (T1);
	\draw[-] (T1) -- (T2);
	\draw[-] (T2) -- (T3);
	\draw[-] (T3) -- (T4);
	\draw[-] (T4) -- (T5);
	\draw[-] (T5) -- (T6);
	\draw[-r- = 0.60 rotate 0, double,double distance = 0.05cm,double,double distance = 0.05cm] (T6) -- (T7);
\end{tikzpicture}
\end{equation}

Kac's theorem tells us that there are four different outer $\bbZ_2$ automorphisms, which are given by sequences $(s_0, \dots, s_7)$, where all components vanish except $s_i = 1$ with $i\in\{0,1,2,3\}$.\footnote{To be very explicit, if we think of $\mathfrak{so}(n)$ as being $n\times n$ antisymmetric matrices, then these four classes of outer automorphisms can be thought of as coming from the adjoint actions of the four matrices $\mathrm{diag}(+1^{2i+1},-1^{15-2i}) \in O(16)$.} The theorem also produces the fixed Lie subalgebras
\begin{equation}
    \mathfrak{so}(1)\oplus\mathfrak{so}(15)\,,\quad
    \mathfrak{so}(3)\oplus\mathfrak{so}(13)\,,\quad
    \mathfrak{so}(5)\oplus\mathfrak{so}(11)\,,\quad
    \mathfrak{so}(7)\oplus\mathfrak{so}(9)\,.
\end{equation}
The fixed Lie subalgebras tell us about the weight-one subspace of the four different $(V_{D_8})^{\bbZ_2}$s. Of course, $\mathfrak{so}(1)$ doesn't really exist, but the demand for its appearance in the pattern becomes clear if we recall that $\mathfrak{so}(n)_1$ is the theory of $n$ free fermions. The patterns above reflect that the four different $(V_{D_8})^{\bbZ_2}$s are just breaking into theories of $(1,15)$, $(3,13)$, $(5,11)$, and $(7,9)$ left-moving free fermions. We address this further in Section \ref{sec:Z2DDFermionization}.

Moving on, let's compute the appropriate $\bbZ_2$-twisted characters of $V_{D_8}$. By the result of \cite{Hohn:2020xfe} they all take the form $\hat{\sigma}_i = e^{2\pi i (x^i,\,\cdot\,)}.\hat{g}$ and have twisted character formula given by Equation \eqref{eq:SvensTwisted}. The first piece for the computation is the twisted eta product $\eta_g(\tau)$. Using \eqref{eq:eta_product}, the $D_8$ Dynkin diagram automorphism gives
\begin{equation}
    \eta_g(\tau) = \eta(\tau)^6 \eta(2 \tau) = \frac{\eta(\tau)^8}{\sqrt{\theta_3(\tau) \theta_4(\tau)}}.
\end{equation}
Moreover, the invariant \textit{root lattice} $L^g$ can be read off from the $D_8$ Dynkin diagram
\begin{equation}
\begin{tikzpicture}[baseline={(current bounding box.center)}, scale = 1]
	\tikzstyle{vertex}=[circle, fill=black, minimum size=2pt,inner sep=2pt];
	\def\r{1.2};
	\node[vertex] (T3) at (\r*3,\r*0) {};
	\node[vertex] (T4) at (\r*4,\r*0) {};
	\node[vertex] (T5) at (\r*5,\r*0) {};
	\node[vertex] (T6) at (\r*6,\r*0) {};
	\node[vertex] (T7) at (\r*6.866,\r*0.866) {};
	\node[vertex] (T8) at (\r*6.866,-1*\r*0.866) {};
	
	\draw[-] (T3) -- (T4);
	\node[] (dots) at (\r*4.5,\r*0) {$\dots$};
	\draw[-] (T5) -- (T6);
	\draw[-] (T6) -- (T7);
	\draw[-] (T6) -- (T8);
	
	\draw [->- = 0.9 rotate 0, ->- = 0.1 rotate 180, black, shorten >=5pt, shorten <=5pt] (T7) to [out=-60,in=60] (T8);
\end{tikzpicture}\quad
\begin{tikzpicture}[scale = 1]
	\tikzstyle{vertex}=[circle, fill=black, minimum size=2pt,inner sep=2pt];
	\def\r{1.2};
	\draw[->,decorate,decoration={snake,amplitude=.4mm,segment length=2mm,post length=1mm}] (\r*0,-0.1) -- (\r*1.5,-0.1) node[midway, above] {Invariant};
\end{tikzpicture}\quad
\begin{tikzpicture}[baseline={(current bounding box.center)}, scale = 1]
	\tikzstyle{vertex}=[circle, fill=black, minimum size=2pt,inner sep=2pt];
	\def\r{1.2};
	\node[vertex] (T3) at (\r*3,\r*0) {};
	\node[vertex] (T4) at (\r*4,\r*0) {};
	\node[vertex] (T5) at (\r*5,\r*0) {};
	\node[vertex] (T6) at (\r*6,\r*0) {};
	\node[vertex] (T7) at (\r*7,\r*0) {};
	
	\draw[-] (T3) -- (T4);
	\node[] (dots) at (\r*4.5,\r*0) {$\dots$};
	\draw[-] (T5) -- (T6);
	\draw[-r- = 0.40 rotate 180, double,double distance = 0.05cm,double,double distance = 0.05cm] (T6) -- (T7);
\end{tikzpicture}
\end{equation}
Note that the fixed root lattice is actually a $C_7$ lattice, not the $B_7$ lattice that we might naively expect from looking at the fixed point Lie subalgebras.

The Dynkin diagram automorphism does not experience order doubling, so we can obtain the weights $e^{2\pi i(x^i,\alpha)}$ by finding vectors in $(1/2)({D_8}^*)^g$ as described in Theorem 2.13 of \cite{Hohn:2020xfe}. We find four choices: 
\begin{alignat}{3}
    &\mathfrak{so}(1)\oplus\mathfrak{so}(15):\quad
        &&x^0 = 0\,,\\
    &\mathfrak{so}(3)\oplus\mathfrak{so}(13):\quad
        &&x^1 =\frac{1}{2} \omega_1\,,\\
    &\mathfrak{so}(5)\oplus\mathfrak{so}(11):\quad
        &&x^2 = \frac{1}{2}\omega_2\,,\\
    &\mathfrak{so}(7)\oplus\mathfrak{so}(9):\quad
        &&x^3 = \frac{1}{2} \omega_3\,.
\end{alignat}

We can then finally sum over the $C_7$ root lattice, with these weights inserted, and combine with the twisted eta functions to obtain the following four duality defects:
\begin{align}
    \tr_{V_{D_8}} \hat{\sigma_i}\, q^{L_0 - \frac{c}{24}} = \frac{\sqrt{\theta_3(\tau) \theta_4(\tau)}}{2 \eta(\tau)^8}\qty(\theta_3^i(\tau) \theta_4^{7-i}(\tau) + \theta_3^{7-i}(\tau) \theta_4^{i}(\tau))\,.\label{eq:DD2FromLA}
\end{align}
In terms of $\mathfrak{so}(2r+1)_1$ characters
\begin{align}\label{eq:BnCharacters}
    \chi^{(r)}_{\hat{\omega}_0}(\tau)
        &=\frac{1}{2}\left(\frac{\theta_3(\tau)^{r+\frac{1}{2}}+\theta_4(\tau)^{r+\frac{1}{2}}}{\eta(\tau)^{r+\frac{1}{2}}}\right)\,,\\
    \chi^{(r)}_{\hat{\omega}_1}(\tau)
        &=\frac{1}{2}\left(\frac{\theta_3(\tau)^{r+\frac{1}{2}}-\theta_4(\tau)^{r+\frac{1}{2}}}{\eta(\tau)^{r+\frac{1}{2}}}\right)\,,\\
    \chi^{(r)}_{\hat{\omega}_r}(\tau)
        &=\frac{1}{\sqrt{2}}\frac{\theta_2(\tau)^{r+\frac{1}{2}}}{\eta(\tau)^{r+\frac{1}{2}}}\,,
\end{align}
we have (up to normalization)
\begin{equation}
    Z_V[0,X_i] 
    = \chi^{(i)}_{\hat{\omega}_0}(\tau)\chi^{(7-i)}_{\hat{\omega}_0}(\tau)
    - \chi^{(i)}_{\hat{\omega}_1}(\tau)\chi^{(7-i)}_{\hat{\omega}_1}(\tau)\,.
\end{equation}

\subsubsection{Fermionization}\label{sec:fermionization}
Given a bosonic theory $T_b$ with a non-anomalous global $\bbZ_2$ symmetry, we may produce a fermionic theory $T_f$ by turning the $\bbZ_2$ symmetry into a $\bbZ_2^f$ ``Grassmann parity'' i.e. $(-1)^F$ symmetry, whose background connection is the spin structure on the manifold (an affine $\bbZ_2$ connection).\footnote{We assume unitarity so that spin and statistics are related.} This is a generalization of the classic ``Jordan-Wigner transformation.'' In practice, on a spacetime $M$, this Jordan-Wigner transformation is given by
\begin{equation}
    Z_{T_f}[\rho] = \frac{1}{\sqrt{\abs{H^1(M,\bbZ_2)}}}\sum_{\alpha\in H^1(M,\bbZ_2)} (-1)^{\Arf[\alpha+\rho]+\Arf[\rho]} Z_{T_b}[\alpha]\,,
\end{equation}
where $\Arf[\rho] = 0$ for even spin structures and $\Arf[\rho]=1$ for odd spin structures. In other words, it is the number of zero modes of the Dirac operator mod 2 \cite{johnson:quadraticSpin, atiyah:spinStruct,seibergWitten:spinStructInString, karchTongTurner:webOf2d}.

In modern language, the partition function $(-1)^{\Arf[\rho]}$ is the generator of the group $\Hom(\Omega_2^{\textrm{Spin}}(pt),U(1)) = \bbZ_2$ of invertible topological phases that can be stacked with a theory with $(-1)^F$ symmetry \cite{fermionicSPTCobordism}. It can also be thought of as the effective action for the continuum version of the Majorana-Kitaev chain. Practically, the effect of stacking with the $\Arf$ phase is to change the relative sign of the even and odd partition functions. 

Conversely, given any fermionic theory $T_f$ (with $c_L-c_R \in 8\bbZ$), it is always possible to gauge the $(-1)^F$ symmetry, i.e. sum over spin-structures, and produce a bosonic theory $T_b$ with a non-anomalous global $\bbZ_2$ symmetry by the inverse of the Jordan-Wigner transformation, or ``GSO projection'' (see \cite{seibergWitten:spinStructInString, worldsheetGSO}). This can be done in two ways, by
\begin{equation}
    Z_{T_b^A}[\alpha] = \frac{1}{\sqrt{\abs{H^1(M,\bbZ_2)}}}\sum_{\rho} (-1)^{\Arf[\alpha+\rho]+\Arf[\rho]} Z_{T_f}[\rho]\,,
\end{equation}
or by first stacking with $\Arf$
\begin{equation}
    Z_{T_b^B}[\alpha] = \frac{1}{\sqrt{\abs{H^1(M,\bbZ_2)}}}\sum_{\rho} (-1)^{\Arf[\alpha+\rho]} Z_{T_f}[\rho]\,.
\end{equation}
The two bosonizations are related by gauging the emergent $\bbZ_2$ symmetry \cite{KramWan, karchTongTurner:webOf2d, Thorngren:2018bhj}, i.e.
\begin{equation}
    [\mathrm{GSO}[T_f]/\bbZ_2] = \mathrm{GSO}[T_f \times \Arf]\,.
\end{equation}
For more examples of this procedure in action, see \cite{MonsterCFT, Ji:2019ugf, fermionicMinimalModels, moreMinimalModels, lessonsFromRamond, okuda2020u1, Smith:2021luc}.

When we gauge $(-1)^F$ we can ask what happens to the chiral fermion parities $(-1)^{F_L}$ and $(-1)^{F_R}$. As explained in \cite{Thorngren:2018bhj, MonsterCFT, Ji:2019ugf} they become $\TY(\bbZ_2)$ duality defect lines separating $T_b^A$ from $T_b^B$.

To make all this concrete, recall the partition function(s) for a chiral (left-moving say) Majorana-Weyl fermion in it's different sectors
\begin{align}
    \tr_{\textrm{NS}} q^{L_0 - \frac{c}{24}} 
        &= \sqrt{\frac{\theta_3(\tau)}{\eta(\tau)}}\,, \\
    \tr_{\textrm{R}} q^{L_0 - \frac{c}{24}}
        &= \sqrt{\frac{\theta_2(\tau)}{\eta(\tau)}}\,,\\
    \tr_{\textrm{NS}} (-1)^F q^{L_0 - \frac{c}{24}}
        &= \sqrt{\frac{\theta_4(\tau)}{\eta(\tau)}}\,,\\
    \tr_{\textrm{R}} (-1)^F q^{L_0 - \frac{c}{24}} 
        &= \sqrt{\frac{\theta_1(\tau)}{\eta(\tau)}}\,.
\end{align}
A ``full'' Majorana fermion (with both left and right-moving contribution) has the spin-structure dependent partition function
\begin{equation}
    Z_{\textrm{Full}}[\rho] = Z_{\textrm{Maj.}}[\rho]\bar{Z}_{\textrm{Maj.}}[\rho]\,.
\end{equation}
For example, we can GSO project the Majorana-Weyl fermion using the formulas above, and in either case we get the Ising partition function
\begin{equation}
    Z_{T_b^{A,B}}[0,0] =  \frac{1}{2}\left(\abs{\frac{\theta_3}{\eta}}+\abs{\frac{\theta_2}{\eta}}+\abs{\frac{\theta_4}{\eta}}\pm\abs{\frac{\theta_1}{\eta}}\right)\,.
\end{equation}

To recover the Ising duality defected partition function, we need to gauge the diagonal spin-structure but first insert a line twisting by the chiral fermion parity i.e. we have
\begin{equation}
    \sum_{\rho} Z_{\textrm{Maj.}}[\rho_1,\rho_2+1]\bar{Z}_{\textrm{Maj.}}[\rho_1,\rho_2] \propto \sqrt{2}\lvert\chi_0\rvert^2 - \sqrt{2}\lvert\chi_{\frac{1}{2}}\rvert^2\,.
\end{equation}

More generally, all of this comes from the fact that $(-1)^{F_L}$ has one unit of the ``mod 8 anomaly'' coming from $\Hom(\Omega_3^{\textrm{Spin}}(B\bbZ_2),U(1)) = \bbZ_8$. For $n$ Majorana fermions the bosonization is the $\mathrm{Spin}(n)_1$ WZW model, and the $(-1)^{F_L}$ symmetry line bosonizes to $\TY(\bbZ_2,1,+1/\sqrt{2})$ if $n=1,7\,\Mod 8$ and $\TY(\bbZ_2,1,-1/\sqrt{2})$ if $n=3,5\,\Mod 8$. See \cite{Thorngren:2018bhj, Ji:2019ugf, Qi_2013,Ryu_2012} and references within for more technical details.

\subsubsection{Duality Defects from Fermionization}\label{sec:Z2DDFermionization}
From Section \ref{ex:VE8} we see that chiral theory $V := V_{E_8}$ is just 16 chiral Majorana-Weyl fermions. From the discussion in the previous section, we see that we have $4$ choices of duality defect based on if we're going to view the $E_8$ theory as coming from the GSO projection of
\begin{equation}
    Z_{\textrm{Maj.}}^1[\rho] Z_{\textrm{Maj.}}^{15}[\rho]\,,\quad
    Z_{\textrm{Maj.}}^3[\rho] Z_{\textrm{Maj.}}^{13}[\rho]\,,\quad
    Z_{\textrm{Maj.}}^5[\rho] Z_{\textrm{Maj.}}^{11}[\rho]\,,\quad
    Z_{\textrm{Maj.}}^7[\rho] Z_{\textrm{Maj.}}^{9}[\rho]\,.
\end{equation}

Let $X_p$ ($p=1,3,5,7$) be the duality defect obtained from bosonizing $(-1)^{F_L}$ in the $Z_{\textrm{Maj.}}^p [\rho] Z_{\textrm{Maj.}}^{16-p}[\rho]$ setup. Then the partition function with $X_p$ inserted along the spatial $S^1$ is
\begin{align}
    Z_V[0,X_p] 
        &\propto \frac{1}{2}\sum_{\rho} Z_{\textrm{Maj.}}^p[\rho_1,\rho_2+1] Z_{\textrm{Maj.}}^{16-p}[\rho_1,\rho_2]\,,\\
        &=\frac{(\theta_3(\tau)\theta_4(\tau))^{\frac{p}{2}}}{2\eta(\tau)^8}\left(\theta_3(\tau)^{8-p}+\theta_4(\tau)^{8-p}\right)\,.
\end{align}
This precisely matches the results of Equation \eqref{eq:DD2FromLA}, where no intimate knowledge of superfusion categories, bosonization, or fermionic anomalies were required.

\section{Higher order duality defects}\label{sec:higherDDs}
In this section, we investigate duality defects for higher order cyclic symmetries. Along the way we highlight a number of potential phenomena that one might encounter in enacting our defect hunting procedure described in Section \ref{sec:DefectedPFs}, such as order doubling and $\mathfrak{u}(1)$ factors. We compute the $q$-expansions of the defect partition functions using Magma \cite{Magma}.

\subsection{\texorpdfstring{$\bbZ_3$}{Z3} Duality Defects}\label{sec:Z3Defects}
As explained in Section \ref{sec:Automorphisms}, there is only one non-anomalous $\bbZ_3$ symmetry up to conjugation. We will see that this symmetry gives rise to seven $\bbZ_3$ Tambara-Yamagami lines, which we will investigate using both numerical and analytical methods.

In Section \ref{sec:orbifoldsOfVE8} we found that the non-anomalous $\bbZ_3$ symmetry is given by $x = \frac{1}{3} \omega_2$. Following our procedure, we obtain the fixed sub-VOA by ``chopping'' the second root of the affine $E_8$ Dynkin diagram
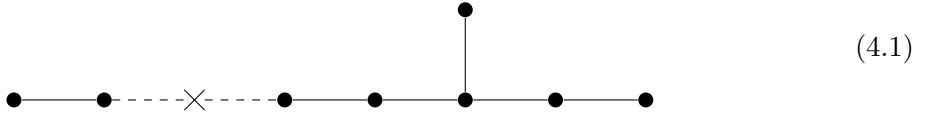
\begin{equation}
\begin{tikzpicture}[baseline={(current bounding box.center)}, scale = 1]
	\tikzstyle{vertex}=[circle, fill=black, minimum size=2pt,inner sep=2pt];
	\def\r{1.2};
	\node[vertex] (T1) at (\r*1,\r*0) {};
	\node[vertex] (T2) at (\r*2,\r*0) {};
	\node[cross out,draw] (T3) at (\r*3,\r*0) {};
	\node[vertex] (T4) at (\r*4,\r*0) {};
	\node[vertex] (T5) at (\r*5,\r*0) {};
	\node[vertex] (T6) at (\r*6,\r*0) {};
	\node[vertex] (T7) at (\r*7,\r*0) {};
	\node[vertex] (T8) at (\r*8,\r*0) {};
	\node[vertex] (T9) at (\r*6,\r*1) {};

	\draw[-] (T1) -- (T2);
	\draw[-, dashed] (T2) -- (T3);
	\draw[-, dashed] (T3) -- (T4);
	\draw[-] (T4) -- (T5);
	\draw[-] (T5) -- (T6);
	\draw[-] (T6) -- (T7);
	\draw[-] (T7) -- (T8);
	\draw[-] (T9) -- (T6);
	\label{eq:A2E6Chopping}
\end{tikzpicture}
\end{equation}
Therefore, we find that $(V_{E_8})^{\bbZ_3} \cong V_{A_2 \times E_6}$.

Before starting the analysis of the Tambara-Yamagami lines, let's study $V_{A_2\times E_6}$ more carefully and verify that we can extend it into two different versions of $V_{E_8}$. For this, recall the following theta functions for $A_2$ and $E_6$ (for details see \cite{conway2013sphere}):
\begin{align}
    \theta_{A_2}(\tau) &= \theta_2(2 \tau) \theta_2(6 \tau) + \theta_3(2\tau) \theta_3(6\tau)\,,\\
    \theta_{E_6}(\tau) &= \frac{1}{4}f(\tau)^3 + \theta_{A_2}(\tau)^3\,,
\end{align}
where 
\begin{equation}
    f(\tau) = \theta_{A_2}\left(\frac{\tau}{3}\right) - \theta_{A_2}(\tau)\,.
\end{equation}

The $A_2$ and $E_6$ lattices have three cosets inside their respective duals. These cosets are parametrized by shift vectors which we will call $[1]$ and $[2]$ in both cases. This gives two shifted theta functions for each lattice given by
\begin{align}
    \theta_{A_2+[1]}(\tau) = \theta_{A_2+[2]}(\tau) &= \frac{1}{2}f(\tau),\\
    \theta_{E_6+[1]}(\tau) = \theta_{E_6+[2]}(\tau) &= \frac{3}{4}f(\tau)^2 \,\theta_{A_2}(\tau).
\end{align}

We organize the cosets of $A_2\times E_6$ inside its dual as a $3\times3$ square, where the first column and the top row give two different extensions to an $E_8$ lattice inside the dual of the $A_2\times E_6$ lattice:

\begin{equation}
(A_2\times E_6)^* =
    \begin{tabular}{|c|c|c|}\hline
        $A_{2}\oplus E_{6}$            
            & $A_{2}^{[1]}\oplus E_{6}^{[2]}$
            & $A_{2}^{[2]}\oplus E_{6}^{[1]}$ 
            \\\hline
        $A_{2}^{[1]}\oplus E_{6}^{[1]}$              
            & $A_{2}^{[2]}\oplus E_{6}$
            & $A_{2}\oplus E_{6}^{[2]}$
            \\\hline
        $A_{2}^{[2]}\oplus E_{6}^{[2]}$ 
            & $A_{2}\oplus E_{6}^{[1]}$
            & $A_{2}^{[1]}\oplus E_{6}$
            \\\hline
    \end{tabular}
    \quad\stackrel{\mathrm{Weight}}{\rightsquigarrow}\quad
    q^A(x)=
    \begin{tabular}{|c|c|c|}
        \hline
        $0$ & $0$ & $0$\\
        \hline
        $0$ & $\tfrac{1}{3}$ & $\tfrac{2}{3}$\\
        \hline
        $0$ & $\tfrac{2}{3}$ & $\tfrac{1}{3}$\\
        \hline
    \end{tabular}\,.
\end{equation}
This shows explicitly that this order three symmetry is non-anomalous. Summing the relevant theta functions we obtain
\begin{equation}
    \theta_{E_8}(\tau) = \theta_{A_2}(\tau)\left(f(\tau)^3 + \theta_{A_2}(\tau)^3\right).
\end{equation}

Now we turn towards the duality defects of this $\bbZ_3$-symmetry. For this purpose, we need to study order two symmetries of the $A_2\times E_6$ lattice and their possible lifts to the corresponding VOA. After that, we investigate which of the $\bbZ_2$ symmetries switch the two $E_8$ extensions and swap the axes of $\Irr(V^{\bbZ_3})$.

We actually start by looking at the automorphisms of the relevant Lie algebras. The $A_2$ Lie algebra has two inner automorphisms of order $ \leq 2$ (up to conjugation) and one outer automorphism of order 2. To see this, note the affine and twisted Dynkin diagrams are
\begin{equation}
A_2^{(1)}\,=
\begin{tikzpicture}[baseline={(current bounding box.center)}, scale = 1]
	\tikzstyle{vertex}=[circle, fill=black, minimum size=2pt,inner sep=2pt];
	\def\r{1.2};
	\node[vertex] (T3) at (\r*0,\r*0) {};
	\node[vertex] (T4) at (\r*1/2,\r*1.732/2) {};
	\node[vertex] (T5) at (\r*1,\r*0) {};
	
	\draw[-] (T3) -- (T4);
	\draw[-] (T3) -- (T5);
	\draw[-] (T5) -- (T4);
	
	\draw[left] (T3) node {$1$};
	\draw[right] (T4) node {$1$};
	\draw[right] (T5) node {$1$};
\end{tikzpicture}\qquad
A_2^{(2)}\,=
\begin{tikzpicture}[baseline={(current bounding box.center)}, scale = 1]
	\tikzstyle{vertex}=[circle, fill=black, minimum size=2pt,inner sep=2pt];
	\def\r{1.2};
	\node[vertex] (T6) at (\r*6,\r*0) {};
	\node[vertex] (T7) at (\r*7,\r*0) {};

    \draw[double,double distance = 0.15cm] (T6) -- (T7);
	\draw[-r- = 0.40 rotate 180, double,double distance = 0.05cm] (T6) -- (T7);
	
	\draw[below] (T6) node {$2$};
	\draw[below] (T7) node {$1$};
\end{tikzpicture}
\end{equation}
The inner automorphisms can be read off from the affine Dynkin diagram. Kac's Theorem tells us we can either chop one node with weight $s_i=2$, or we can chop two nodes with weights $s_i=s_j=1$; by the symmetry of the affine Dynkin diagram, it doesn't matter which we chop. In the first case we get the trivial automorphism, and in the second case we get an inner automorphism of order $2$ fixing an $A_1 \times \mathfrak{u}(1)$ Lie algebra. There is only one outer automorphism which comes from chopping the node marked $1$ with weight $1$, fixing an $A_1$ Lie algebra. Standard lifts of this automorphism experience order doubling (as we will explain in Section \ref{sec:OrderDoubling}).

The $E_6$ case is similar, we have
\begin{equation}
E_6^{(1)} \,=
\begin{tikzpicture}[baseline={(current bounding box.center)}, scale = 1]
	\tikzstyle{vertex}=[circle, fill=black, minimum size=2pt,inner sep=2pt];
	\def\r{1.2};
	\node[vertex] (T4) at (\r*4,\r*0) {};
	\node[vertex] (T5) at (\r*5,\r*0) {};
	\node[vertex] (T6) at (\r*6,\r*0) {};
	\node[vertex] (T7) at (\r*7,\r*0) {};
	\node[vertex] (T8) at (\r*8,\r*0) {};
	\node[vertex] (T9) at (\r*6,\r*1) {};
	\node[vertex] (T10) at (\r*6,\r*2) {};

	\draw[-] (T4) -- (T5);
	\draw[-] (T5) -- (T6);
	\draw[-] (T6) -- (T7);
	\draw[-] (T7) -- (T8);
	\draw[-] (T9) -- (T6);
	\draw[-] (T9) -- (T10);
	
	\draw[below] (T4) node {$1$};
	\draw[below] (T5) node {$2$};
	\draw[below] (T6) node {$3$};
	\draw[below] (T7) node {$2$};
	\draw[below] (T8) node {$1$};
	\draw[right] (T9) node {$2$};
	\draw[right] (T10) node {$1$};
\end{tikzpicture}\quad
E_6^{(2)}\,=
\begin{tikzpicture}[baseline={(current bounding box.center)}, scale = 1]
	\tikzstyle{vertex}=[circle, fill=black, minimum size=2pt,inner sep=2pt];
	\def\r{1.2};
	\node[vertex] (T4) at (\r*4,\r*0) {};
	\node[vertex] (T5) at (\r*5,\r*0) {};
	\node[vertex] (T6) at (\r*6,\r*0) {};
	\node[vertex] (T7) at (\r*7,\r*0) {};
	\node[vertex] (T8) at (\r*8,\r*0) {};

	\draw[-] (T4) -- (T5);
	\draw[-] (T5) -- (T6);
	\draw[-r- = 0.40 rotate 180, double,double distance = 0.05cm,double,double distance = 0.05cm] (T6) -- (T7);
	\draw[-] (T7) -- (T8);

	\draw[below] (T4) node {$1$};
	\draw[below] (T5) node {$2$};
	\draw[below] (T6) node {$3$};
	\draw[below] (T7) node {$2$};
	\draw[below] (T8) node {$1$};
\end{tikzpicture}
\end{equation}
We find 3 inner automorphisms of the $E_6$ Lie algebra up to conjugation, given by: chopping a node marked $1$ with weight $2$ (the trivial automorphism); chopping a node marked $2$ with weight $1$, fixing $A_1 \times A_5$; or chopping two nodes marked $1$ with weight $1$, fixing $D_5\times \mathfrak{u}(1)$. We find $2$ outer automorphisms of $E_6$, fixing an $F_4$ or a $C_4$ Lie algebra.

Putting everything together, we have 15 automorphisms of the $A_2\times E_6$ Lie algebra of order less than or equal to 2. These automorphisms can be organised as in Table \ref{tab:A2E6automorphsims}.
\begin{table}[t]
    \centering
    \begin{tabular}{c c c}
    \toprule
    & $E_6$ inner (3) & $E_6$ outer (2)\\
    \midrule 
    $A_2$ inner (2) & 6         & 4 \\
    $A_2$ outer (1) & 3         & 2 \\
    \bottomrule
    \end{tabular}
    \caption{There are 15 automorphisms of the $A_2 \times E_6$ Lie algebra of order less than or equal to 2. The row and column labels denote the origin of the $A_2\times E_6$ automorphisms. For example, there are $4$ $\bbZ_2$ actions on $A_2\times E_6$ which come from one of the (2) inner actions on $A_2$ times one of the (2) outer actions on $E_6$.}
    \label{tab:A2E6automorphsims}
\end{table}
Similar to the $D_8$ case, we can read off automorphisms which switch the axes of $\Irr(V^{\bbZ_3})$ by looking at the chopped $E_8$ Dynkin diagram \eqref{eq:A2E6Chopping}. We see that there are two ways to recombine the $A_2$ and $E_6$ into an (affine) $E_8$: by gluing the chopped node back in, or by acting with \textit{one} of (the lifts of) the Dynkin diagram automorphisms of $A_2$ or $E_6$, \textit{but not both}, and then gluing. From the off-diagonals of Table \ref{tab:A2E6automorphsims}, this gives a total of seven automorphisms switching the axes and hence seven Tambara-Yamagami lines. However, before computing the defect partition functions, we need to look at order-doubling in the $A_2$ lattice.

\subsubsection{Example: Order Doubling in \texorpdfstring{$A_2$}{A2}}\label{sec:OrderDoubling}
Before moving forward, let's consider a \textit{standard} lift of the one non-trivial outer automorphism of $A_2$, as it is a prototypical example of order-doubling.

So let $g$ be the automorphism of the $A_2$ lattice that comes from the Dynkin diagram automorphism; switching the two simple roots $\alpha_1$ and $\alpha_2$. As described in Section \ref{sec:Automorphisms}, $g$ can be lifted to a symmetry $\hat{g}$ of the lattice VOA in a ``standard way'': swapping the chiral bosons $a^1_n \leftrightarrow a^2_n$, and sending $\Gamma_{\alpha} \mapsto u(\alpha) \Gamma_\alpha$, with $u(\alpha)=1$ on the fixed sublattice. In this case, the fixed sublattice is a ``long $A_1$'' spanned by $\alpha_1+\alpha_2$
\begin{equation}
\begin{tikzpicture}[baseline={(current bounding box.center)}, scale = 1]
	\tikzstyle{vertex}=[circle, fill=black, minimum size=2pt,inner sep=2pt];
	\def\r{1.2};
	\node[vertex] (T3) at (\r*3,\r*0) {};
	\node[vertex] (T4) at (\r*4,\r*0) {};
	
	\draw[-] (T3) -- (T4);
	\draw [->- = 0.9 rotate 0, ->- = 0.1 rotate 180, black, shorten >=5pt, shorten <=5pt] (T3) to [out=45,in=135] (T4);
\end{tikzpicture}\quad
\begin{tikzpicture}[scale = 1]
	\tikzstyle{vertex}=[circle, fill=black, minimum size=2pt,inner sep=2pt];
	\def\r{1.2};
	\draw[->,decorate,decoration={snake,amplitude=.4mm,segment length=2mm,post length=1mm}] (\r*0,-0.1) -- (\r*1.5,-0.1) node[midway, above] {Invariant};
\end{tikzpicture}\quad
\begin{tikzpicture}[baseline={(current bounding box.center)}, scale = 1]
	\tikzstyle{vertex}=[circle, fill=black, minimum size=2pt,inner sep=2pt];
	\def\r{1.2};
	\node[vertex] (T3) at (\r*3,\r*0) {};
\end{tikzpicture}
\end{equation}

Clearly $\hat{g}^2$ is the identity on the chiral bosons, but
\begin{equation}
    \hat{g}^2 \Gamma_{\alpha} = u(\alpha) u(g\alpha) \Gamma_{\alpha}\,.
\end{equation}
Using Equation \eqref{eq:liftExists} we have
\begin{equation}
    u(\alpha) u(g\alpha) = \frac{\epsilon(\alpha,g\alpha)}{\epsilon(g\alpha,g^2\alpha)}u(\alpha+g\alpha) = (-1)^{(g\alpha,\alpha)} u(\alpha+g\alpha)\,.
\end{equation}
Thus the order of $\hat{g}$ is doubled if and only if $(g\alpha,\alpha) \notin 2\bbZ$. 

Better yet, we explicitly see the origin of order doubling here: we get order doubling because the Dynkin diagram automorphism folds together two simple roots which were not originally orthogonal. Hence, just by looking at Dynkin diagrams, we can tell which automorphisms will have order doubling.

At the end of Section \ref{sec:Automorphisms} we saw all outer order $2$ automorphisms of $V_{A_2}$ can be viewed as a product of a (fixed) standard lift of this Dynkin diagram automorphism, and a twist by an ``inner'' automorphism (see Equation \eqref{eq:svensAutomorphisms}). Now we see the importance of allowing twists by inner automorphisms of order up to $4$. If we had not, we might falsely conclude that ``there are no outer automorphisms of $V_{A_2}$ of order 2,'' which would directly contradict the results of \cite{dong1999automorphism} and the fact known to physicists that the lattice automorphism $\alpha\mapsto-\alpha$ always lifts to a ``charge conjugation automorphism'' of order $2$ in the VOA (note: $\alpha\mapsto-\alpha$ is not in the Weyl group of the $A_2$ lattice).


If instead we choose $u$ so that $u(\alpha_1 +\alpha_2) = -1$ the lift has order two, and we can compute the appropriate twisted theta function. The fixed sublattice is a long $A_1$ lattice with theta function $\theta_3(2\tau)$, but, due to the non-trivial lift, we adjust some signs and obtain $\theta_4(2\tau)$.

\subsubsection{\texorpdfstring{$\bbZ_3$}{Z3} Duality Defects from Lie Theory}
The partition functions corresponding to the seven Tambara-Yamagami lines can be computed using similar techniques to the order two case. A careful study of the lattices and partition functions gives the following list of theta functions 
\begin{align}
    \theta_{A_2,\text{inner}}(\tau) 
        &= \theta_3(2 \tau) \theta_3(6 \tau) - \theta_2(2 \tau) \theta_2(6\tau)\\
    \theta_{A_2,\text{outer}}(\tau) 
        &= \theta_4(2\tau)\\
    \theta_{E_6,\text{inner 1}}(\tau) 
        &= 2 \theta_{A_5}(\tau) \theta_3(2 \tau) - \theta_{E_6}(\tau)\\\
    \theta_{E_6,\text{inner 2}}(\tau) 
        &= 1 + 8q - 50q^2 + 80q^3 - 88q^4 + \dots\\
    \theta_{E_6,\text{outer 1}}(\tau) 
        &= \tfrac{1}{2}( \theta_3^4(\tau) + \theta_4^4(\tau))\\
    \theta_{E_6,\text{outer 2}}(\tau) 
        &= \theta_4^4(2\tau)\,.
\end{align}
The missing theta function corresponding to an inner $E_6$ symmetry is particularly messy, so we display the first few terms of its $q$-expansion. The twisted eta products are easily computed being $\eta(\tau)^6 \eta(2\tau)$ for the outer $A_2$ symmetries and $\eta(\tau)^4 \eta(2\tau)^2$ for the outer $E_6$ symmetries. In summary, we have the following seven defected partition functions:

\begin{align}
    Z_V[0,X_{A_2}^{(1)}]
        &= \frac{\theta_{A_2,\text{outer}}(\tau) \theta_{E_6}(\tau)}{\eta(\tau)^6 \eta(2\tau)}\,,\label{eq:Z3Defect1}\\
    Z_V[0,X_{A_2}^{(2)}]
        &= \frac{\theta_{A_2,\text{outer}}(\tau) \theta_{E_6,\text{inner 1}}(\tau)}{\eta(\tau)^6 \eta(2\tau)}\,,\\
   Z_V[0,X_{A_2}^{(3)}] 
        &= \frac{\theta_{A_2,\text{outer}}(\tau) \theta_{E_6,\text{inner 2}}(\tau)}{\eta(\tau)^6 \eta(2\tau)}\,,\\
    Z_V[0,X_{E_6}^{(1)}] 
        &= \frac{\theta_{A_2}(\tau) \theta_{E_6, \text{outer 1}}(\tau)}{\eta(\tau)^4 \eta(2\tau)^2}\,,\label{eq:Z3Defect4}\\
    Z_V[0,X_{E_6}^{(2)}] 
        &= \frac{\theta_{A_2,\text{inner}}(\tau) \theta_{E_6, \text{outer 1}}(\tau)}{\eta(\tau)^4 \eta(2\tau)^2}\,,\\
    Z_V[0,X_{E_6}^{(3)}] 
        &= \frac{\theta_{A_2}(\tau) \theta_{E_6, \text{outer 2}}(\tau)}{\eta(\tau)^4 \eta(2\tau)^2}\,,\\
    Z_V[0,X_{E_6}^{(4)}] 
        &= \frac{\theta_{A_2,\text{inner}}(\tau) \theta_{E_6, \text{outer 2}}(\tau)}{\eta(\tau)^4 \eta(2\tau)^2}\,.\label{eq:Z3Defect7}
\end{align}
We have recorded the $q$-expansions of these defect partition functions (in the same order) with their fixed Lie subalgebras in Table \ref{tab:class_duality_3}. We also obtain the defect in Equation \eqref{eq:Z3Defect4} from the Potts CFT in Appendix \ref{sec:PottsDefect}. 

\begin{table}[t]
\centering
\begin{tabular}{ >{\centering\arraybackslash} m{4cm} >{\raggedright\arraybackslash} m{8cm}}    \toprule
    \emph{$((V^{\bbZ_3})^{\bbZ_2})_1$} & \emph{$q$-Expansion of $\bbZ_3$ Duality Defects} \\\midrule
        $A_1 \times E_6$ 
            & $q^{-\frac{1}{3}}(1 + 76 q + 574 q^2 + 3000 q^3 + O(q^4))$\\
        $A_1 \times A_1 \times A_5$
            & $q^{-\frac{1}{3}}(1 - 4 q + 14 q^2 - 40 q^3 + O(q^4))$\\
        $A_1 \times D_5 \times \mathfrak{u}(1)$
            & $q^{-\frac{1}{3}}(1 + 12 q - 2 q^2 + 56 q^3 + O(q^4))$\\
        $A_2 \times F_4$
            & $q^{-\frac{1}{3}}(1 + 34  q + 304 q^2 + 1446 q^3 + O(q^4))$\\
        $A_1 \times F_4 \times \mathfrak{u}(1)$
            & $q^{-\frac{1}{3}}(1 + 26 q + 80 q^2 + 350 q^3 + O(q^4))$\\
        $A_2 \times C_4$
            & $q^{-\frac{1}{3}}(1 + 2 q - 16 q^2 + 38 q^3 + O(q^4))$\\
        $A_1 \times C_4 \times \mathfrak{u}(1)$
            & $q^{-\frac{1}{3}}(1 - 6 q + 16 q^2 - 34 q^3 + O(q^4))$\\
    \bottomrule
    \hline
\end{tabular}
\caption{List of the defect partition functions of order 3, with defect twisting time. On the left are the weight one subspaces of $(V^{\bbZ_3})^{\bbZ_2}$; these are Lie algebras, not invariant lattices. On the right are $q$-expansions of the defect partition functions from Equations (\ref{eq:Z3Defect1}\,--\,\ref{eq:Z3Defect7}). The first 3 defects are lifted from outer $A_2$ lattice actions, the next 4 are lifted from outer $E_6$ lattice actions.}\label{tab:class_duality_3}
\end{table}

As a final check, we've obtained the same $q$-expansions using computer algebra software Magma (in a way that we will explain in the following section) that matches with (at least) the first 11 terms of the $q$-expansions of the exact theta series above. This convinces us that our numerical methods work in the order three case, and justifies the computation of $\bbZ_4$ and $\bbZ_5$ defected partition functions using numerics only. In principle, exact formulas are obtainable in those cases following the techniques above.

We comment on data relevant to the symmetric non-degenerate bicharacter in Appendix \ref{sec:Rearrangements}.

\subsection{\texorpdfstring{$\bbZ_4$}{Z4} Duality Defects and Computer Implementation}\label{sec:Z4Defects}
Moving to the order 4 defects, our direct Lie algebra methods become less useful. As seen in Table \ref{tab:nonanomsym}, the only non-anomalous order 4 symmetry of $E_8$ has fixed Lie algebra $A_7\times\mathfrak{u}(1)$. This is not a simple Lie algebra and thus does not have an associated root system.

\subsubsection{An Algorithm for Computing Defects}
In this case, we can still construct the fixed sublattice as in the beginning of Section \ref{sec:Z2Orbifold}. The symmetry is generated by the vector 
\begin{equation}
    x=\frac{1}{4}\omega_8.
\end{equation}
Accordingly, the fixed sublattice is
\begin{equation}\label{eq:Z4fixedLattice}
    L_0=\mathbb{Z}\alpha_1\oplus\cdots\oplus\mathbb{Z}\alpha_7\oplus \mathbb{Z}4\alpha_8.
\end{equation}

We can now compute $q$-expansions of our defect partition functions with the aid of Magma \cite{Magma}.\footnote{We thank Sven M\"oller for providing us with our original Magma code and helpful Magma lessons.} The code is effectively an implementation of Theorem 2.13 of \cite{Hohn:2020xfe}. The algorithm is as follows:
\begin{enumerate}
    \itemsep0em
    \item[Step 0.] Input an even lattice $L_0$. This can be done in Magma by providing a basis for $L_0$ as above with built-in function BasisWithLattice(\,).
    \item[Step 1.] Compute the automorphism group $O(L_0)$ with built-in function AutomorphismGroup(\,) and the setwise stabilizer of a fixed set of roots $H_{\Delta}$. To obtain $H_{\Delta}$ is more involved:
    \begin{enumerate}
        \itemsep0em
        \item Enumerate a full set of simple roots (of length 2) $\Delta$ and split them into positive and negative roots. Note: this collection of simple roots may not be full rank, e.g. if there are $\mathfrak{u}(1)$'s involved and the lattice is not a root-lattice.
        \item Find the subgroup of $O(L_0)$ stabilizing some choice of simple roots by looking at group orbits with built-in function OrbitAction(\,,\,). It is possible that $O(L_0)$ won't act faithfully on the collection $\Delta$, in which case continually append $\Delta$ with vectors of higher norm until $O(L_0)$ acts faithfully, then obtain $H_{\Delta}$.
    \end{enumerate}
    In view of \cite{Hohn:2020xfe}, every finite-order automorphism of the VOA $V_L$ can be constructed by considering a fixed set of representatives $\nu$ of the conjugacy classes of $H_\Delta$. Obtained with built-in function ConjugacyClasses(\,).
    \item[Step 2.] Compute the centralizer $C_{O(L)}(\nu)$ of $\nu$ in $O(L)$ or, more precisely, a fixed set of orbit representatives of its action on $L^\nu\otimes_\mathbb{Z}\mathbb{Q}/{\pi_\nu(L')}$.
    Here $\pi_\nu$ is the projection of $L\otimes_\mathbb{Z}\mathbb{C}$ onto $L^\nu\otimes_\mathbb{Z}\mathbb{C}$, with $L^\nu$ the sublattice fixed by $\nu$. These are our ``outer $\bbZ_2$'' actions.
    
    Proposition 2.15 of \cite{Hohn:2020xfe} guarantees that in order to find automorphisms of order 2 swapping the axes, we only need to look at order 1 and 2 lattice automorphisms, $\nu$, and restrict the choice of $h$'s according to whether or not the standard lift of $\nu$ has order doubling.
    
    \item[Step 3.] Loop through all $\nu$ and compute all possible independent $h$ for each $\nu$. For each $\nu$ loop over all $h$ to create the automorphism $\hat{\nu}e^{2\pi i h(0)}$.
    
    \item[Step 4.] Implementing the formula in Equation \eqref{eq:SvensTwisted}, compute the $q$-expansions of our $\bbZ_2$-twisted characters for $V^{\bbZ_m}$ by summing over weighted lattice vectors.
\end{enumerate}

At this point, we have enumerated all order 2 automorphisms of our VOA $V_{L_0}$ and the (q-expansions of the) $\bbZ_2$ twisted characters. However, not all order 2 automorphisms swap the axes of the metric Abelian group $A \cong \bbZ_m \times \hat{\bbZ}_m$. Fortunately, this is simple to check within the code described above. To do this, note that our automorphisms extend to all of $\bbR^8$ and hence to our original $E_8$ lattice $L$. Thus, by looking at the image of the basis vectors under $\nu$, we can determine if the automorphism leaves the original $E_8$ invariant. In the case that it is not left invariant, then $L$ must be mapped to a different lattice, which must be the second $E_8$ extension $L^\prime$ by orthogonality of the automorphism. A similar procedure also tells us about the symmetric non-degenerate bicharacter, see Appendix \ref{sec:Rearrangements}.

\subsubsection{The \texorpdfstring{$\bbZ_4$}{Z4} Defects}
Let $b_i=\alpha_i$ for $i=1,\dots,7$, and let $b_8 = 4\alpha_8$. Following our algorithm above, we find three non-trivial lattice automorphisms $\nu$ of $L_0 = \bbZ\{b_i\}_{i=1}^8$. As matrices $b_i \mapsto (\nu_k)_{ij} b_j$, we have $\nu_1^{(4)}$, $\nu_2^{(4)}$ and $\nu_3^{(4)}$, which are respectively
\begin{equation}
\begin{aligned}
    \begin{pmatrix}
        0\, & 0\, & 0\, & 0\, & 0\, & 0\, & 1\, & -4\\
        0 & 0 & 0 & 0 & 0 & 1 & 0 & -8\\
        0 & 0 & 0 & 0 & 1 & 0 & 0 & -12\\
        0 & 0 & 0 & 1 & 0 & 0 & 0 & -12\\
        0 & 0 & 1 & 0 & 0 & 0 & 0 & -12\\
        0 & 1 & 0 & 0 & 0 & 0 & 0 & -8\\
        1 & 0 & 0 & 0 & 0 & 0 & 0 & -4\\
        0 & 0 & 0 & 0 & 0 & 0 & 0 & -1
    \end{pmatrix}\qc\begin{pmatrix}
       1\, & 0\, & 0\, & 0\, & 0\, & 0\, & 0\, & -3\\
        0 & 1 & 0 & 0 & 0 & 0 & 0 & -6\\
        0 & 0 & 1 & 0 & 0 & 0 & 0 & -9\\
        0 & 0 & 0 & 1 & 0 & 0 & 0 & -12\\
        0 & 0 & 0 & 0 & 1 & 0 & 0 & -15\\
        0 & 0 & 0 & 0 & 0 & 1 & 0 & -10\\
        0 & 0 & 0 & 0 & 0 & 0 & 1 & -5\\
        0 & 0 & 0 & 0 & 0 & 0 & 0 & -1
    \end{pmatrix}\qc\begin{pmatrix}
       0\, & 0\, & 0\, & 0\, & 0\, & 0\, & 1\, & -1\\
        0 & 0 & 0 & 0 & 0 & 1 & 0 & -2\\
        0 & 0 & 0 & 0 & 1 & 0 & 0 & -3\\
        0 & 0 & 0 & 1 & 0 & 0 & 0 & 0\\
        0 & 0 & 1 & 0 & 0 & 0 & 0 & 3\\
        0 & 1 & 0 & 0 & 0 & 0 & 0 & 2\\
        1 & 0 & 0 & 0 & 0 & 0 & 0 & 1\\
        0 & 0 & 0 & 0 & 0 & 0 & 0 & 1
    \end{pmatrix}.
\end{aligned}\label{eq:NonTrivialZ4s}
\end{equation}
To reiterate, these are representatives for the 3 conjugacy classes of $\bbZ_2$ automorphisms of the lattice $L_0$. We can check the actions of these matrices on the $E_8$ basis vectors spanning $L$, and find that $\nu_2^{(4)}$ and $\nu_3^{(4)}$ do not leave $L$ invariant, and hence swap the two distinct $E_8$'s in $L_0^*/L_0$. Neither $\nu_2^{(4)}$ nor $\nu_3^{(4)}$ experience order doubling.

Combining everything, we find six defected partition functions, with the weight one subspaces and $q$-expansions recorded in Table \ref{tab:class_duality_4} for the $\bbZ_4$ case.
\begin{table}[t]
\centering
\begin{tabular}{ >{\centering\arraybackslash} m{4cm} >{\raggedright\arraybackslash} m{8cm}}    \toprule
    \emph{$((V^{\bbZ_4})^{\bbZ_2})_1$} & \emph{$q$-Expansion of $\bbZ_4$ Duality Defects} \\\midrule
        $A_7$ 
        & $q^{-\frac{1}{3}}(1 + 62 q + 784 q^2 + 5088 q^3 + O(q^4))$\\
        $A_1 \times A_5 \times \mathfrak{u}(1)$
        & $q^{-\frac{1}{3}}(1 + 14 q + 16 q^2 + 96 q^3 + O(q^4))$\\
        $A_3 \times A_3 \times \mathfrak{u}(1)$
        & $q^{-\frac{1}{3}}(1 - 2 q + 16 q^2 - 32 q^3 + O(q^4))$\\
        $C_4 \times \mathfrak{u}(1)$
        & $q^{-\frac{1}{3}}(1 + 10 q + 80 q^2 + 224 q^3 + O(q^4))$\\
        $C_4 \times \mathfrak{u}(1)$
        & $q^{-\frac{1}{3}}(1 + 10 q + 16 q^2 + 96 q^3 + O(q^4))$\\
        $D_4 \times \mathfrak{u}(1)$
        & $q^{-\frac{1}{3}}(1 - 6 q + 16 q^2 - 32 q^3 + O(q^4))$\\
    \bottomrule
    \hline
\end{tabular}
\caption{List of the defect partition functions of order 4, with defect twisting time. The first three come from the lattice automorphism $\nu_2^{(4)}$ while the rest come from $\nu_3^{(4)}$.} \label{tab:class_duality_4}
\end{table}

These results are plausible from Lie theory. Ignoring the $\mathfrak{u}(1)$ factor there are several symmetries of the $A_7$ Lie algebra. First one can leave the roots of the $A_7$ Dynkin diagram unchanged. The corresponding phase factors can then be chosen to obtain the fixed Lie subalgebras $A_{6-i} \times A_i \times \mathfrak{u}(1)$. The other option is to swap the roots in the $A_7$ Dynkin diagram. In this situation the fixed subalgebra is given by $C_4$. This compares well to the structure presented in Table \ref{tab:class_duality_4}, where the first three lines correspond to the first case and the next two lines to the second case.

More generally, in the case of $\bbZ_m$ symmetries with $m$ not prime, its also important to keep in mind that we need to ``totally swap'' the axes, since there can be $\bbZ_2$ automorphisms of $\Irr(V^{\bbZ_m})$ which only swap a subgroup of $\bbZ_m$ with a subgroup of $\hat{\bbZ}_m$. 

We record results on the symmetric non-degenerate bicharacter at the end of Appendix \ref{sec:Rearrangements}.

\subsection{\texorpdfstring{$\bbZ_5$}{Z5} Duality Defects}\label{sec:Z5Defects}
There are two non-anomalous $\bbZ_5$ symmetries up to conjugacy, and so two families of duality defects of order $5$. We will arbitrarily call the two $\bbZ_5$ conjugacy classes $\bbZ_5^A$ and $\bbZ_5^B$ with
generators
\begin{align}
    x^A &:= \frac{1}{5}\omega_4\,,\\
    x^B &:= \frac{1}{5}\omega_1+\frac{1}{5}\omega_7\,.
\end{align}
The sub-VOAs $V^{\bbZ_5^A}$ and $V^{\bbZ_5^B}$ have weight one subspaces isomorphic to $A_4\times A_4$ and $D_6 \times \mathfrak{u}(1)^2$ respectively.

\subsubsection{\texorpdfstring{$\bbZ_5^A$}{Z5A} Duality Defects and \texorpdfstring{$A_4 \times A_4$}{A4 x A4}}
The invariant sublattice $L_0^A$ under the $\bbZ_5^A$ symmetry is a genuine root lattice spanned by the $A_4\times A_4$ inside $E_8$, so we can proceed in a straightforward way. Consider the chopped Dynkin diagram
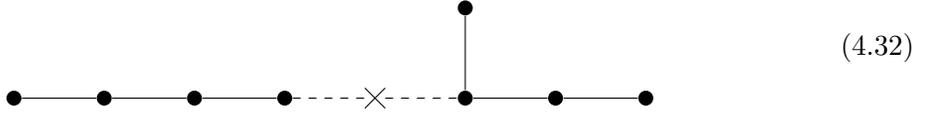
\begin{equation}
\begin{tikzpicture}[baseline={(current bounding box.center)}, scale = 1]
	\tikzstyle{vertex}=[circle, fill=black, minimum size=2pt,inner sep=2pt];
	\def\r{1.2};
	\node[vertex] (T1) at (\r*1,\r*0) {};
	\node[vertex] (T2) at (\r*2,\r*0) {};
	\node[vertex] (T3) at (\r*3,\r*0) {};
	\node[vertex] (T4) at (\r*4,\r*0) {};
	\node[cross out,draw] (T5) at (\r*5,\r*0) {};
	\node[vertex] (T6) at (\r*6,\r*0) {};
	\node[vertex] (T7) at (\r*7,\r*0) {};
	\node[vertex] (T8) at (\r*8,\r*0) {};
	\node[vertex] (T9) at (\r*6,\r*1) {};

	\draw[-] (T1) -- (T2);
	\draw[-] (T2) -- (T3);
	\draw[-] (T3) -- (T4);
	\draw[-, dashed] (T4) -- (T5);
	\draw[-, dashed] (T5) -- (T6);
	\draw[-] (T6) -- (T7);
	\draw[-] (T7) -- (T8);
	\draw[-] (T9) -- (T6);
	\label{eq:A4A4Chopping}
\end{tikzpicture}
\end{equation}
Focusing on just one $A_4$, the $A_4^{(1)}$ and $A_4^{(2)}$ Dynkin diagrams tell us that the $A_4$ Lie algebra has 3 inner automorphisms of order less than or equal to $2$, with fixed Lie subalgebras: $A_4$ (trivial); $A_3 \times \mathfrak{u}(1)$; and $A_2 \times A_1 \times \mathfrak{u}(1)$. Similarly, there is one order 2 outer automorphism with a fixed $B_2$ subalgebra.

Since we have two $A_4$ Dynkin diagrams life is slightly more complicated than the $\bbZ_3$ case. First, there are the usual commuting $\bbZ_2$ actions on each $A_4$ Dynkin diagram $\sigma_1$ and $\sigma_2$ say; but we can also exchange the diagrams themselves in different ways, this gives another $\bbZ_2$ action $\tau$. The three together generate a dihedral group of order $8$ $\langle \,\sigma_1,\sigma_2,\tau\,\rangle \cong \mathrm{Di}_8$.

However, we are not interested in all eight elements of this group, but in the conjugacy classes. There are three classes which do not switch the axes of $\Irr(V^{\bbZ_5})$ given by $\{e\}$, $\{\sigma_1\sigma_2\}$ and $\{\sigma_1 \tau, \sigma_2 \tau\}$ and two classes switching the axes given by $\{\sigma_1, \sigma_2\}$ and $\{\tau, \sigma_1\sigma_2 \tau\}$. 

As before, we can obtain explicit matrix representatives of these classes. If we represent the $A_4\times A_4$ lattice as the span of $b_i=\alpha_i$ for $i\neq 4$ and $b_4=5\alpha_4$, then representatives for the non-trivial lattice automorphisms $\nu^{(5)}_{[\sigma_1\sigma_2]}$, $\nu^{(5)}_{[\sigma_1\tau]}$, $\nu^{(5)}_{[\sigma_1]}$, and $\nu^{(5)}_{[\tau]}$ are given respectively by
\begin{equation}
    \begin{aligned}
    \begin{pmatrix}
        1 & 0 & 0 & 0 & 0 & 0 & 1 & 0\\
        2 & 0 & 1 & -5 & 0 & 0 & 0 & 0\\
        3 & 1 & 0 & -5 & 0 & 0 & 0 & 0\\
        1 & 0 & 0 & -1 & 0 & 0 & 0 & 0\\
        6 & 0 & 0 & -10 & 0 & 1 & 0 & 0\\
        4 & 0 & 0 & -10 & 1 & 0 & 0 & 0\\
        2 & 0 & 0 & -5 & 0 & 0 & 0 & 1\\
        3 & 0 & 0 & -5 & 0 & 0 & 1 & 0
    \end{pmatrix}\qc 
    &\begin{pmatrix}
        0 & 0 & 0 & -5 & 0 & 0 & 1 & 1\\
        0 & 0 & 0 & -10 & 1 & 0 & 2 & 0\\
        0 & 0 & 0 & -10 & 0 & 1 & 3 & 0\\
        0 & 0 & 0 & -2 & 0 & 0 & 1 & 0\\
        0 & 0 & 1 & -15 & 0 & 0 & 6 & 0\\
        0 & 1 & 0 & -10 & 0 & 0 & 4 & 0\\
        1 & 0 & 0 & -5 & 0 & 0 & 2 & 0\\
        0 & 0 & 0 & -5 & 0 & 0 & 3 & 0
    \end{pmatrix}\qc\\
    \begin{pmatrix}
        1 & 0 & 0 & 0 & 0 & 0 & 1 & 0\\
        2 & 0 & 1 & -5 & 0 & 0 & 0 & 0\\
        3 & 1 & 0 & -5 & 0 & 0 & 0 & 0\\
        1 & 0 & 0 & -1 & 0 & 0 & 0 & 0\\
        6 & 0 & 0 & -12 & 1 & 0 & 0 & 0\\
        4 & 0 & 0 & -8 & 0 & 1 & 0 & 0\\
        2 & 0 & 0 & -4 & 0 & 0 & 1 & 0\\
        3 & 0 & 0 & -6 & 0 & 0 & 0 & 1
    \end{pmatrix}\qc
    &\begin{pmatrix}
        0 & 0 & 0 & -5 & 0 & 0 & 1 & 1\\
        0 & 0 & 0 & -10 & 1 & 0 & 2 & 0\\
        0 & 0 & 0 & -10 & 0 & 1 & 3 & 0\\
        0 & 0 & 0 & -2 & 0 & 0 & 1 & 0\\
        0 & 1 & 0 & -14 & 0 & 0 & 6 & 0\\
        0 & 0 & 1 & -11 & 0 & 0 & 4 & 0\\
        0 & 0 & 0 & -3 & 0 & 0 & 2 & 0\\
        1 & 0 & 0 & -7 & 0 & 0 & 3 & 0
    \end{pmatrix}.
    \end{aligned}
\end{equation}

The elements of this $A_4$ ``flipping'' conjugacy class $\{\sigma_1, \sigma_2\}$ flip roots of one $A_4$ while leaving one copy of $A_4$ invariant. In this case, the fixed Lie subalgebra is one of the inners multiplied by a $B_2$ coming from the $A_4$ outer automorphism, giving 3 duality defects for this conjugacy class. Note that standard lifts from this class experience order doubling because the outer automorphism folds roots together which are not orthogonal. The ``exchange'' conjugacy class $\tau$ fixes a ``diagonal'' $A_4$ Lie subalgebra giving 1 final duality defect. It does not experience order doubling. We enumerate these 4 cases with their $q$-expansions in Table \ref{tab:class_duality_5A}.

\begin{table}[t]
\centering
\begin{tabular}{ >{\centering\arraybackslash} m{4cm} >{\raggedright\arraybackslash} m{8cm}}    \toprule
    \emph{$((V^{\bbZ_5^A})^{\bbZ_2})_1$} & \emph{$q$-Expansion of $\bbZ_5^A$ Duality Defects} \\\midrule
        $A_4 \times B_2$ 
        & $q^{-\frac{1}{3}}(1 + 20 q + 34 q^2 + 140 q^3 + O(q^4))$\\
        $A_3 \times B_2 \times \mathfrak{u}(1)$
        & $q^{-\frac{1}{3}}(1 + 4 q - 14 q^2 + 28 q^3 + O(q^4))$\\
        $A_1 \times A_2 \times B_2 \times \mathfrak{u}(1)$
        & $q^{-\frac{1}{3}}(1 - 4 q + 10 q^2 - 28 q^3 + O(q^4))$\\
        $A_4$
        & $q^{-\frac{1}{3}}(1 + 24 q + 124 q^2 + 500 q^3 + O(q^4))$\\
    \bottomrule
    \hline
\end{tabular}
\caption{List of the defect partition functions of order 5A, with defect twisting time. The first 3 defects come from the $\{\sigma_1,\sigma_2\}$ ``flip'' conjugacy class lattice action. The final defect lifts from the $\{\tau,\sigma_1\sigma_2\tau\}$ ``exchange'' conjugacy class lattice action.} \label{tab:class_duality_5A}
\end{table}

\subsubsection{\texorpdfstring{$\bbZ_5^B$}{Z5B} Duality Defects and \texorpdfstring{$D_6 \times \mathfrak{u}(1)^2$}{D6 x u(1)**2}}\label{sec:Z5B}
For the case of $D_6 \times \mathfrak{u}(1)^2$ the solution is slightly less elegant and cannot be read off just by looking at Dynkin diagrams. Proceeding as in Section \ref{sec:Z4Defects}, we can realize the fixed sublattice as the span of $b_i=\alpha_i$ for $i\neq 1$, $7$, and $b_1=5\alpha_1$ and $b_7=4\alpha_1+\alpha_7$. We use our previous methods to obtain representatives of the conjugacy classes of non-trivial automorphisms $\nu^{(5)}_1$, $\nu^{(5)}_2$, $\nu^{(5)}_3$, $\nu^{(5)}_4$, given respectively by
\begin{equation}
    \begin{aligned}
    \begin{pmatrix}
        -1 & 0 & 0 & 0 & 0 & 0 & 0 & 0\\
        -10 & 1 & 0 & 0 & 0 & 0 & -9 & 0\\
        -10 & 0 & 1 & 0 & 0 & 0 & -10 & 0\\
        -10 & 0 & 0 & 1 & 0 & 0 & -11 & 0\\
        -10 & 0 & 0 & 0 & 1 & 0 & -12 & 0\\
        -5 & 0 & 0 & 0 & 0 & 1 & -7 & 0\\
        0 & 0 & 0 & 0 & 0 & 0 & -1 & 0\\
        -5 & 0 & 0 & 0 & 0 & 0 & -6 & 1
    \end{pmatrix}\qc
    &\begin{pmatrix}
        -5 & 0 & 0 & 0 & 0 & 0 & -4 & 0\\
        -3 & 1 & 0 & 0 & 0 & 0 & -2 & 0\\
        0 & 0 & 1 & 0 & 0 & 0 & 0 & 0\\
        3 & 0 & 0 & 1 & 0 & 0 & 2 & 0\\
        6 & 0 & 0 & 0 & 1 & 0 & 4 & 0\\
        6 & 0 & 0 & 0 & 0 & 1 & 4 & 0\\
        6 & 0 & 0 & 0 & 0 & 0 & 5 & 0\\
        3 & 0 & 0 & 0 & 0 & 0 & 2 & 1
    \end{pmatrix}\qc\\
    \begin{pmatrix}
        5 & 0 & 0 & 0 & 0 & 0 & 3 & 0\\
        -16 & 1 & 0 & 0 & 0 & 0 & -12 & 0\\
        -20 & 0 & 1 & 0 & 0 & 0 & -15 & 0\\
        -24 & 0 & 0 & 1 & 0 & 0 & -18 & 0\\
        -28 & 0 & 0 & 0 & 1 & 0 & -21 & 0\\
        -18 & 0 & 0 & 0 & 0 & 0 & -13 & 1\\
        -8 & 0 & 0 & 0 & 0 & 0 & -5 & 0\\
        -14 & 0 & 0 & 0 & 0 & 1 & -11 & 0
    \end{pmatrix}\qc
    &\begin{pmatrix}
        -7 & 0 & 0 & 0 & 0 & 0 & -5 & 0\\
        5 & 1 & 0 & 0 & 0 & 0 & 2 & 0\\
        10 & 0 & 1 & 0 & 0 & 0 & 5 & 0\\
        15 & 0 & 0 & 1 & 0 & 0 & 8 & 0\\
        20 & 0 & 0 & 0 & 1 & 0 & 11 & 0\\
        15 & 0 & 0 & 0 & 0 & 0 & 9 & 1\\
        10 & 0 & 0 & 0 & 0 & 0 & 7 & 0\\
        10 & 0 & 0 & 0 & 0 & 1 & 5 & 0
    \end{pmatrix}.
    \end{aligned}\label{eq:NonTrivialZ5s}
\end{equation}

The automorphisms switching the axes of $\Irr(V^{\bbZ_5^B})$ are $\nu^{(5)}_2$ and $\nu^{(5)}_3$. The standard lifts of both of these experience order doubling. In total, there are 7 defected partitions functions, 4 which lift from $\nu_2^{(5)}$ and 3 which lift from $\nu_3^{(5)}$. They are recorded in Table \ref{tab:class_duality_5B}.

\begin{table}[t]
\centering
\begin{tabular}{ >{\centering\arraybackslash} m{4cm} >{\raggedright\arraybackslash} m{8cm}}    \toprule
    \emph{$((V^{\bbZ_5^B})^{\bbZ_2})_1$} & \emph{$q$-Expansion of $\bbZ_5^B$ Duality Defects} \\\midrule
        $D_5 \times \mathfrak{u}(1)^2$ 
        & $q^{-\frac{1}{3}}(1 + 26 q + 80q^2 + 352 q^3 + O(q^4))$\\
        $A_5 \times \mathfrak{u}(1)^2 $
        & $q^{-\frac{1}{3}}(1 + 6 q + 32 q^2 + 32 q^3 + O(q^4))$\\
        $A_5 \times \mathfrak{u}(1)^2 $
        & $q^{-\frac{1}{3}}(1 + 6 q + 32 q^2 - 44 q^3 + O(q^4))$\\
        $A_3\times A_3\times\mathfrak{u}(1)$
        & $q^{-\frac{1}{3}}(1 - 6 q + 16 q^2 - 32 q^3 + O(q^4))$\\
        $B_5\times\mathfrak{u}(1)$ & $q^{-\frac{1}{3}}(1 + 44 q + 266 q^2 + 1268 q^3 + O(q^4))$\\
        $A_1\times B_4\times\mathfrak{u}(1)$ & $q^{-\frac{1}{3}}(1 + 12 q + 10 q^2 + 84 q^3 + O(q^4))$\\
        $B_2\times B_3\times\mathfrak{u}(1)$ & $q^{-\frac{1}{3}}(1 - 4 q + 10 q^2 - 28 q^3 + O(q^4))$\\
    \bottomrule
    \hline
\end{tabular}
\caption{List of the defect partition functions of order 5B, with defect twisting time. The first four come from the automorphism $\nu^{(5)}_2$, while the rest come from $\nu^{(5)}_3$.} \label{tab:class_duality_5B}
\end{table}

\section{A (2+1)d TFT Perspective}\label{sec:3dTFT}
There is a natural (2+1)d perspective which sheds light on our previous discussions and the procedure in Section \ref{sec:DefectedPFs} and Equation \eqref{eq:DualityDefectPF}. This all follows from the theory of modular tensor categories, fusion categories, and their relationship to gapped boundary conditions for (2+1)d TFTs. 

The relevant pieces of this formalism are well-described in mathematical physics works such as \cite{Fuchs_2013, Kong:2013aya} and are well-known in the condensed-matter literature (see e.g. \cite{Barkeshli:2014cna, Cong_2017}). For this reason, we will provide only a small recap of the beautiful connection between (2+1)d TFTs, their gapped boundary conditions, and the relevant mathematics, before turning to duality defects. Helpful mathematical references include \cite{tensorCatBook2, tensorCatBook}.

\subsection{Topological Boundaries of (2+1)d TFTs}\label{sec:topBds}
Given a fusion category $\calA$ it's Drinfeld center $\calZ(\calA)$ is naturally a braided fusion category. This leads to the natural mathematical question: are all braided fusion categories the Drinfeld center of some fusion category? The answer is no, as already demonstrated by various Chern-Simons theories \cite{Freed:2020qfy}. 

The next simplest question in this line of thinking is: if one has a braided (non-degenerate\footnote{Non-degenerate means that the S-matrix is non-degenerate. By definition this is true for a MTC, so applies to e.g. reps of a rational VOA. More broadly, non-degeneracy is satisfied iff there are no anyons which are invisible to braiding with the entire rest of the category (Theorem 3.4 of \cite{DGNO}), known as ``remote detectability'' \cite{Levin_2013, kong2014braided, johnsonfreyd2020classification}. So-called ``slightly degenerate'' braided fusion categories are also of interest in the study of fermionic topological orders \cite{DNO, johnsonfreyd2020classification}.}) fusion category $\calC$, when is $\calC = \calZ(\calA)$ for some fusion category $\calA$? The answer to this question is also known: the data of a braided equivalence $\calC \cong \calZ(\calA)$ determines a Lagrangian algebra object $A \in \calC$; and inversely, each Lagrangian algebra object $A\in\calC$ determines a braided equivalence $\calC \cong \calZ(\calC_A)$ where $\calC_A$ is the category of right-$A$ modules in $\calC$. Moreover, equivalences $\calZ(\calA) \cong \calZ(\calB)$ are known to be in bijection with indecomposable $\calA$-module categories \cite{ETINGOF2011176, MUGER200381}. Altogether, one has the following result:

\begin{proposition}[\cite{DMNO}, Proposition 4.8]\label{prop:DMNOGappedBC}
    Let $\calA$ be a fusion category and $\calC = \calZ(\calA)$. There is a bijection between the sets of (isomorphism classes of) Lagrangian algebras in $\calC$ and (equivalence classes of) indecomposable $\calA$-module categories.
\end{proposition}

These questions arise naturally in physics as well, as the data of a (2+1)d TFT are encoded as a Modular Tensor Category of topological line defects or anyons \cite{freedQGfromPI, tensorCatBook2, bartlett2015modular}. A question of interest when studying a (2+1)d TFT is: does the (2+1)d theory admit any topological boundary conditions?

In the physical case, the MTC of anyons $\calC$ is the braided fusion category of interest, and it can be shown that if $\calC \cong \calZ(\calA)$ for some $\calA$, then the bulk admits a topological boundary condition $B$ with boundary line defects modelled by $\calA$. So we see that the natural question for physicists is essentially the same as those studying braided fusion categories \cite{Fuchs_2013, Kong:2013aya}.

More generally, the theory of ``anyon condensation'' allows one to describe topological interfaces between two TFTs, described by MTCs $\calC$ and $\calD$, in the case that the anyons of $\calC$ ``condense'' to a new vacuum for $\calD$ at the separating interface \cite{Bais:2008ni, Kong:2013aya}. The formalism allows one to easily compute condensation interfaces for $\calC$, the new phase $\calD$, and the excitations on the interface (see Theorem 4.7 of \cite{Kong:2013aya} and \cite{mattAnyon} for computational references). In particular, the search for topological boundary conditions is a condensation to vacuum, controlled by Lagrangian algebra objects $A\in \calC$.\footnote{We only care about bosonic anyon condensations and bosonic topological boundary conditions. For results on fermionic anyon condensation and fermionic topological orders see \cite{Lou:2020gfq} and references within.}

It is common in the physics literature to talk about a particular realization of a (2+1)d TFT modeled by the MTC $\calC$, rather than just discuss abstract MTC data. For example, in Turaev-Viro/Levin-Wen models \cite{Turaev:1992hq, Levin:2004mi} one chooses a particular ``input fusion category'' $\calA$ and constructs the (2+1)d TFT whose anyons are modeled by the MTC $\calC = \calZ(\calA)$. From this frame, the boundary condition with excitations modeled by $\calA$ is particularly distinguished as a ``reference'' or ``Dirichlet'' boundary condition, and one identifies gapped boundaries with indecomposable $\calA$-module categories \cite{Cong_2017}.\footnote{We will use the term Dirichlet boundary condition to mean the canonical boundary condition for the Turaev-Viro theory based on $\calA$.} By Proposition \ref{prop:DMNOGappedBC} these are in bijection with the Lagrangian algebras.

Similarly, if the theory admits a gauge theory description, making such distinctions are equivalent to saying ``who the Wilson line is'' in the TFT. The choice of a particular boundary condition $B$ defines the Wilson line by declaring that $B$ is the Dirichlet boundary condition for the bulk gauge fields in that frame. Mathematically, the Wilson lines are the kernel of our particular bulk-boundary map $\calZ(\calA) \to \calA$, i.e. their image is a sum of copies of the vacuum operator on the boundary.

The simplest and most well known example of this relationship in physics is in the toric code \cite{Kitaev_2003}. Algebraically, the simple anyons of the (2+1)d TFT are described by the irreducible representations of the untwisted quantum double $D(\bbZ_2)$, or more geometrically as (2+1)d Dijkgraaf-Witten theory with gauge group $\bbZ_2$ and trivial topological action \cite{DW90}. In short, the simple anyons in this model are $\{1,e,m,f\}$ with spins $\{0,0,0,1/2\}$ respectively and non-trivial fusion relations
\begin{equation}
    e\otimes m = f\,,\quad e\otimes f = m\,,\quad m\otimes f = e\,.
\end{equation}

Famously, such a model admits two \textit{bosonic} indecomposable topological boundary conditions as described in \cite{bravyi1998quantum}, where they manifest in the microscopic description as ``smooth'' and ``rough'' boundaries of the square lattice. Such a description makes it pictorially clear how certain anyons of the toric code ``condense'' as they reach the boundary. For example, the $e$ anyon living along a string of vertices is condensed at a rough boundary, but the $m$ anyon gets stuck and becomes a boundary defect line (see also \cite{Beigi_2011} and \cite{Kitaev_2012}). In other words, the ``rough'' topological boundary condition $B_e$ condenses the $e$ anyon to the vacuum on the topological boundary, and corresponds to the Lagrangian algebra object $1\oplus e$. The topological boundary condition is populated with boundary excitations $\{1\oplus e,m\oplus f\}$, which behave like $\mathcal{A}_e = \Vc_{\bbZ_2}$.

\subsection{Duality Defects}\label{sec:DDs}

Gauging 2d theories has a simple (2+1)d interpretation in terms of boundary conditions and interfaces in topological gauge theories (see e.g. \cite{Freed:2018cec, gaiotto2020orbifold, theoMoonshine} for an extended explanation). In particular, given any 2d theory $T$ with non-anomalous $\bbZ_m$ symmetry, the theory provides an ``enriched Neumann'' boundary condition $B[T]$ for a (2+1)d $\bbZ_m$ Dijkgraaf-Witten theory by coupling the boundary value of the bulk connection to $T$. The local operators of $B[T]$ are effectively the local operators of $T^{\bbZ_m}$.

\begin{figure}[t]
	\centering
	\begin{tikzpicture}[thick]
	
	\def\Depth{4}
	\def\DepthTwo{3}
	\def\Height{2}
	\def\Width{2}
	\def\Sep{3}        
	
	\coordinate (O) at (0,0,0);
	\coordinate (A) at (0,\Width,0);
	\coordinate (B) at (0,\Width,\Height);
	\coordinate (C) at (0,0,\Height);
	\coordinate (D) at (\Depth,0,0);
	\coordinate (E) at (\Depth,\Width,0);
	\coordinate (F) at (\Depth,\Width,\Height);
	\coordinate (G) at (\Depth,0,\Height);
	\draw[black] (O) -- (C) -- (G) -- (D) -- cycle;
	\draw[black] (O) -- (A) -- (E) -- (D) -- cycle;
	\draw[black, fill=blue!20,opacity=0.8] (O) -- (A) -- (B) -- (C) -- cycle;
	\draw[black, fill=yellow!20,opacity=0.8] (D) -- (E) -- (F) -- (G) -- cycle;
	\draw[black] (C) -- (B) -- (F) -- (G) -- cycle;
	\draw[black] (A) -- (B) -- (F) -- (E) -- cycle;
	\draw[below] (0, 0*\Width, \Height) node{\begin{tabular}{l} a) $B[T$] \\ b) $\mathrm{Dir}^{\mathrm{e}}$ \end{tabular}};
	\draw[below] (\Depth, 0*\Width, \Height) node{$I_{em}$};
	\draw[midway] (\Depth/2,\Width-\Width/2,\Height/2) node {\begin{tabular}{c} Electric $\bbZ_m$\\ Gauge Theory\end{tabular}};
	\draw[midway] (\Depth/2+\Depth,\Width-\Width/2,\Height/2) node {\begin{tabular}{c} Magnetic $\hat{\bbZ}_m$\\ Gauge Theory\end{tabular}};
	
	\coordinate (O2) at (\Depth,0,0);
	\coordinate (A2) at (\Depth,\Width,0);
	\coordinate (B2) at (\Depth,\Width,\Height);
	\coordinate (C2) at (\Depth,0,\Height);
	\coordinate (D2) at (\Depth+\DepthTwo,0,0);
	\coordinate (E2) at (\Depth+\DepthTwo,\Width,0);
	\coordinate (F2) at (\Depth+\DepthTwo,\Width,\Height);
	\coordinate (G2) at (\Depth+\DepthTwo,0,\Height);
	\draw[black] (O2) -- (D2);
	\draw[black] (A2) -- (E2);
	\draw[black] (B2) -- (F2);
	\draw[black] (C2) -- (G2);
	
	\coordinate (O3) at (\Depth+\DepthTwo+\Sep+3*\Depth/24,0,0);
	\coordinate (A3) at (\Depth+\DepthTwo+\Sep+3*\Depth/24,\Width,0);
	\coordinate (B3) at (\Depth+\DepthTwo+\Sep+3*\Depth/24,\Width,\Height);
	\coordinate (C3) at (\Depth+\DepthTwo+\Sep+3*\Depth/24,0,\Height);
	\coordinate (D3) at (\Depth+2*\DepthTwo+\Sep+3*\Depth/24,0,0);
	\coordinate (E3) at (\Depth+2*\DepthTwo+\Sep+3*\Depth/24,\Width,0);
	\coordinate (F3) at (\Depth+2*\DepthTwo+\Sep+3*\Depth/24,\Width,\Height);
	\coordinate (G3) at (\Depth+2*\DepthTwo+\Sep+3*\Depth/24,0,\Height);
	\draw[black, fill=green!20,opacity=0.8] (O3) -- (C3) -- (B3) -- (A3) -- cycle;
	\draw[black] (O3) -- (D3);
	\draw[black] (A3) -- (E3);
	\draw[black] (B3) -- (F3);
	\draw[black] (C3) -- (G3);
	\draw[below] (\Depth+\DepthTwo+\Sep*3/2-\Sep*5/32, 0*\Width, \Height) node{\begin{tabular}{l} a) $B[[T/\bbZ_m]]$ \\ b) $\mathrm{Neu}^{\mathrm{m}}$ \end{tabular}};
	\draw[midway] (\Depth/2+\Depth+\DepthTwo+\Sep+3*\Depth/24,\Width-\Width/2,\Height/2) node {\begin{tabular}{c} Magnetic $\hat{\bbZ}_m$\\ Gauge Theory\end{tabular}};
	
	\draw[->,decorate,decoration={snake,amplitude=.4mm,segment length=2mm,post length=1mm}] (\Depth+\DepthTwo+\Width/4+\Depth/24, \Width/2, \Height/2) -- (\Sep+\Depth+\DepthTwo-\Width/4+2*\Depth/24,\Width/2,\Height/2) node[midway, above] {Collide $I_{em}$};
	\end{tikzpicture}
	\caption{On the left, a slab of $\bbZ_m$ gauge theory and $\hat{\bbZ}_m$ gauge theory separated by an (invertible topological) electric-magnetic duality interface $I_{em}$ (yellow). If one moves the interface $I_{em}$ all the way to the left, then all that remains is $\hat{\bbZ}_m$ gauge theory with a new boundary condition (green).
	a) If the left boundary condition (blue) is an enriched Neumann boundary condition $B[T]$, then when it collides with the interface it appears as the enriched Neumann boundary condition $B[[T/\bbZ_m]]$.
	b) If the left boundary condition (blue) is a Dirichlet boundary condition for the $\bbZ_m$ gauge theory, then it appears as a Neumann boundary condition for the $\hat{\bbZ}_m$ gauge theory.} \label{fig:collisionOfInterface}
\end{figure}
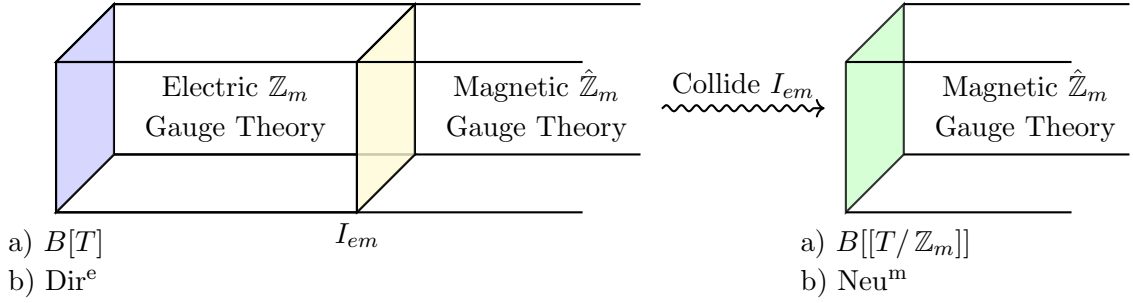

In this case, the 2d theory $T$ (on $M$ say) can be understood as a compactification of a slab of this $\bbZ_m$ gauge theory on $M\times[0,1]$, with the boundary condition $B[T]$ for one boundary, and the reference (topological) Dirichlet boundary condition for the other. Similarly, $[T/\bbZ_m]$ can be understood as the same setup but with the (topological) fully-Neumann boundary condition instead. 

However, by the same logic, $[T/\bbZ_m]$ can also be viewed as a sandwich of $\hat{\bbZ}_m$ gauge theory with boundary condition $B[[T/\bbZ_m]]$ and Dirichlet boundary condition for $\hat{\bbZ}_m$ on the other. These two descriptions are dual, what has changed is the identification of the bulk TFT as an ``electric'' $\bbZ_m$ gauge theory or a ``magnetic'' $\hat{\bbZ}_m$. By our previous section, the MTC data is the same, but we have presented a different topological boundary condition as ``Dirichlet''. The two pictures are related by an electric-magnetic duality interface $I_{em}$, which composes with $B[T]$ to produce $B[[T/\bbZ_m]]$ and turns the electric Dirichlet boundary condition into a magnetic Neumann boundary condition. We depict all of this in Figure \ref{fig:collisionOfInterface}.

Since a duality defect separates a theory from its orbifold, we recognize a duality defect as a potential ending line for the interface $I_{em}$, separating the electric and magnetic Dirichlet boundary conditions in a slab picture i.e. as a $\bbZ_2^{em}$ twist line operator, see Figure \ref{fig:KWisZ2TwistLine}. Gauging the electric-magnetic $\bbZ_2^{em}$ 0-form symmetry \textit{of the bulk TFT} (with Neumann boundary conditions on the left and Dirichlet boundary conditions on the right of Figure \ref{fig:KWisZ2TwistLine}), produces a new bulk TFT where the interface $I_{em}$ ``becomes invisible,'' leaving the twist line behind, as in Figure \ref{fig:KWfromTYCat}. In Section XI of \cite{Barkeshli:2014cna} (also Section 3.5 of \cite{Bhardwaj:2016clt}) the authors show explicitly how the boundary twist line left behind after gauging has the fusion rules of $\TY(\bbZ_m)$.

This is simply squaring two commensurate pictures. In one case, we are thinking of $T$ as having a $\bbZ_m$ symmetry with topological line defects for the symmetry modeled by the fusion category $\calD = \Vc_{\bbZ_m}$, and so we can couple it to a $\calZ(\Vc_{\bbZ_m})$ bulk and view the duality defect line as the endpoint of an electric-magnetic duality wall. Alternatively, we can view $T$ as having ``categorical symmetry'' $\calF = \TY(\bbZ_m)$, in which case it can be coupled to a bulk of $\calZ(\TY(\bbZ_m))$ theory, and the duality defect line is understood simply as a condensate of bulk lines on the topological boundary.

In general, suppose $\calD$ admits a $G$-extension
\begin{equation}
    \calF = \bigoplus_{g\in G} \calF_g
\end{equation}
with identity piece $\calF_e \cong \calD$, then the (2+1)d TFT $Q$ with MTC $\calZ(\calD)$ and reference Dirichlet boundary condition $\calD$ necessarily has topological surfaces with a $G$-composition law. If we gauge this $G$ symmetry to produce a new (2+1)d TFT $[Q/G]$, then the new MTC data is given by $\calZ(\calF)$ with reference boundary condition $\calF$. This follows from the work of \cite{gelaki2009centers,Barkeshli:2014cna}.

\begin{figure}[t]
\centering
\begin{subfigure}[b]{0.9\textwidth}
\begin{equation*}
\begin{tikzpicture}[thick, baseline={([yshift=-.5ex]current bounding box.center)}]
	\def\Depth{6}
	\def\Height{3}
	\def\Width{2}
	\def\Sep{3}        
	\coordinate (O) at (0,0,0);
	\coordinate (A) at (0,\Width,0);
	\coordinate (B) at (0,\Width,\Height);
	\coordinate (C) at (0,0,\Height);
	\coordinate (KW1) at (0,0,\Height/2);
	\coordinate (KW2) at (0,\Width,\Height/2);
	\draw[black, fill=blue!20,opacity=0.8] (KW1) -- (KW2) -- (B) -- (C) -- cycle;
	\draw[black, fill=purple!20,opacity=0.8] (O) -- (A) -- (KW2) -- (KW1) -- cycle;
	\draw[black] (O) -- (A) -- (B) -- (C) -- cycle;
	\draw[below] (C) node{$T$};
	\draw[above] (-1*\Depth*5/60,\Width,0) node{$[T/\bbZ_m]$};
	\draw[ultra thick, red] (KW1) -- (KW2);
\end{tikzpicture}
\quad\longleftrightarrow
\begin{tikzpicture}[thick, baseline={([yshift=-.5ex]current bounding box.center)}]
	\def\Depth{6}
	\def\Height{3}
	\def\Width{2}
	\def\Sep{3}        
	\coordinate (O) at (0,0,0);
	\coordinate (A) at (0,\Width,0);
	\coordinate (B) at (0,\Width,\Height);
	\coordinate (C) at (0,0,\Height);
	\coordinate (D) at (\Depth,0,0);
	\coordinate (E) at (\Depth,\Width,0);
	\coordinate (F) at (\Depth,\Width,\Height);
	\coordinate (G) at (\Depth,0,\Height);
	\coordinate (I1) at (0,0,\Height/2);
	\coordinate (I2) at (0,\Width,\Height/2);
	\coordinate (I3) at (\Depth,\Width,\Height/2);
	\coordinate (I4) at (\Depth,0,\Height/2);	
	\draw[black, fill=yellow!20,opacity=0.8] (I1) -- (I2) -- (I3) -- (I4) -- cycle;
	\draw[below] (\Depth/2,\Width*2/3,\Height/2) node{$I_{em}$};
	\draw[below] (\Depth/2, 0*\Width, \Height) node{$\,\bbZ_m$ Gauge Theory};
	\draw[above] (\Depth/2, \Width, \Height*0) node{$\,\,\hat{\bbZ}_m$ Gauge Theory};
	\draw[black] (O) -- (C) -- (G) -- (D) -- cycle;
	\draw[black] (O) -- (A) -- (E) -- (D) -- cycle;
	\draw[black] (O) -- (A) -- (B) -- (C) -- cycle;
	\draw[black] (D) -- (E) -- (F) -- (G) -- cycle;
	\draw[black] (C) -- (B) -- (F) -- (G) -- cycle;
	\draw[black] (A) -- (B) -- (F) -- (E) -- cycle;
	\draw[below] (0, 0*\Width, \Height) node{$T^{\bbZ_m\phantom{Y()}}$};
	\draw[above] (0, \Width, \Height*0) node{$[T/\bbZ_m]^{\hat{\bbZ}_m}\,$};
	\draw[below] (\Depth, 0*\Width, \Height) node{$\,\,\mathrm{Dir}^e$};
	\draw[above] (\Depth, \Width, \Height*0) node{$\qquad\mathrm{Dir}^m$};
	\coordinate (KW1) at (\Depth,0,\Height/2);
	\coordinate (KW2) at (\Depth,\Width,\Height/2);
	\coordinate (KW1) at (\Depth,0,\Height/2);
	\coordinate (KW2) at (\Depth,\Width,\Height/2);
	\draw[ultra thick, red] (\Depth,0,\Height/2) -- (\Depth,\Width,\Height/2);
\end{tikzpicture}
\end{equation*}
\caption{The 2d theory $T$ (blue) is separated from $[T/\bbZ_m]$ (purple) by a duality defect (red). This picture can be blown up to a slab of electric $\bbZ_m$ gauge theory (front) and magnetic $\hat{\bbZ}_m$ gauge theory (back) separated by an electric-magnetic duality wall (yellow). $\bbZ_m$ invariant local operators of $T$ live at the left boundary and similarly for $[T/\bbZ_m]$. Electric and magnetic Dirichlet boundary conditions for the respective gauge theories are established at the right boundary. The duality defect is a twist line for the bulk $\bbZ_2^{em}$ symmetry.}
\label{fig:KWisZ2TwistLine}
\end{subfigure}
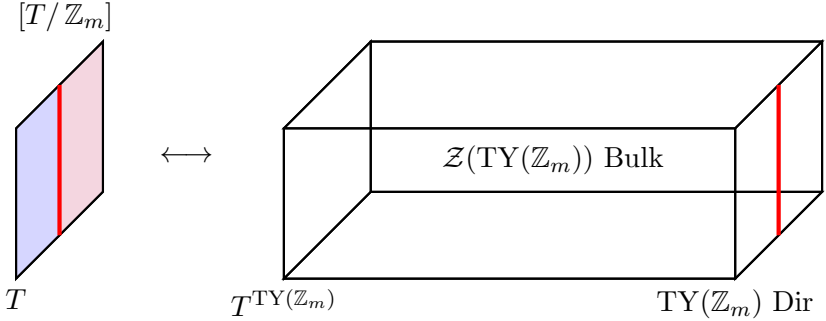
\begin{subfigure}[b]{0.9\textwidth}
\begin{equation*}
\begin{tikzpicture}[thick, baseline={([yshift=-.5ex]current bounding box.center)}]
	\def\Depth{6}
	\def\Height{3}
	\def\Width{2}
	\def\Sep{3}        
	\coordinate (O) at (0,0,0);
	\coordinate (A) at (0,\Width,0);
	\coordinate (B) at (0,\Width,\Height);
	\coordinate (C) at (0,0,\Height);
	\coordinate (KW1) at (0,0,\Height/2);
	\coordinate (KW2) at (0,\Width,\Height/2);
	\draw[black, fill=blue!20,opacity=0.8] (KW1) -- (KW2) -- (B) -- (C) -- cycle;
	\draw[black, fill=purple!20,opacity=0.8] (O) -- (A) -- (KW2) -- (KW1) -- cycle;
	\draw[black] (O) -- (A) -- (B) -- (C) -- cycle;
	\draw[below] (C) node{$T$};
	\draw[above] (-1*\Depth*5/60,\Width,0) node{$[T/\bbZ_m]$};
	\draw[ultra thick, red] (KW1) -- (KW2);
\end{tikzpicture}
\quad\longleftrightarrow
\begin{tikzpicture}[thick, baseline={([yshift=-.5ex]current bounding box.center)}]
	\def\Depth{6}
	\def\Height{3}
	\def\Width{2}
	\def\Sep{3}        
	\coordinate (O) at (0,0,0);
	\coordinate (A) at (0,\Width,0);
	\coordinate (B) at (0,\Width,\Height);
	\coordinate (C) at (0,0,\Height);
	\coordinate (D) at (\Depth,0,0);
	\coordinate (E) at (\Depth,\Width,0);
	\coordinate (F) at (\Depth,\Width,\Height);
	\coordinate (G) at (\Depth,0,\Height);
	\coordinate (I1) at (0,0,\Height/2);
	\coordinate (I2) at (0,\Width,\Height/2);
	\coordinate (I3) at (\Depth,\Width,\Height/2);
	\coordinate (I4) at (\Depth,0,\Height/2);	
	\draw[below] (\Depth/2,\Width*2/3,\Height/2) node{$\calZ(\TY(\bbZ_m))$ Bulk};
	\draw[black] (O) -- (C) -- (G) -- (D) -- cycle;
	\draw[black] (O) -- (A) -- (E) -- (D) -- cycle;
	\draw[black] (O) -- (A) -- (B) -- (C) -- cycle;
	\draw[black] (D) -- (E) -- (F) -- (G) -- cycle;
	\draw[black] (C) -- (B) -- (F) -- (G) -- cycle;
	\draw[black] (A) -- (B) -- (F) -- (E) -- cycle;
	\draw[below] (0, 0*\Width, \Height) node{$T^{\TY(\bbZ_m)}$};
	\draw[below] (\Depth, 0*\Width, \Height) node{$\TY(\bbZ_m)$ $\mathrm{Dir}$};
	\draw[ultra thick, red] (\Depth,0,\Height/2) -- (\Depth,\Width,\Height/2);
	\draw[above] (\Depth, \Width, \Height*0) node{$\phantom{\qquad\mathrm{Dir}^m}$};
	\draw[above] (0, \Width, \Height*0) node{$\phantom{[T/\bbZ_m]^{\hat{\bbZ}_m}\,}$};
\end{tikzpicture}
\end{equation*}
\caption{The same 2d picture can be blown up into a slab of $\calZ(\TY(\bbZ_m))$ bulk with $\TY(\bbZ_m)$ invariant local operators on the left boundary and Dirichlet boundary conditions on the right boundary. The duality defect now lives at the end of the trivial wall. $\calZ(\TY(\bbZ_m))$ has $2m$ lines with quantum dimension $\sqrt{m}$, hence many bulk lines condense to the same line at the Dirichlet boundary.}
\label{fig:KWfromTYCat}
\end{subfigure}
\caption{Two different (2+1)d blow-ups for a theory with a duality defect line inserted. To go from (a) to (b), one must gauge the $\bbZ_2^{em}$ 0-form symmetry generated by the duality wall $I_{em}$, with Neumann boundary conditions for $\bbZ_2^{em}$ on the left and Dirichlet on the right. This makes the duality wall $I_{em}$ ``become invisible,'' leaving behind the duality defect. To go from (b) to (a), one must gauge the bulk $\bbZ_2^{em}$ 1-form symmetry of $\calZ(\TY(\bbZ_m))$.}
\end{figure}

\subsubsection{Application to Holomorphic VOAs}\label{sec:holoVOA}
Following \cite{moller2017cyclic}, we call a VOA \textit{nice} if it is simple, rational, $C_2$-cofinite, self-contragredient, and of CFT-type. Examples of nice VOAs include lattice VOAs built from even lattices and the Monster $V^\natural$.

\begin{proposition}\label{prop:confInc}
    Given a conformal inclusion of nice VOAs $W \subset V$ there is a canonical fusion category $\mathcal{F}$ determined by the inclusion.
\end{proposition}
This follows from the theory of anyon condensations described in Theorem 4.7 and Section 6.4 of \cite{Kong:2013aya}. In the VOA literature, the proposition is typically phrased in the reverse direction as being about extensions of the VOA $W$.
\begin{proof}
    Since $W$ and $V$ are nice, their representation categories are MTCs \cite{Huang:2005gs}, let $\calC := \Rep(W)$ and $\calD := \Rep(V)$. Since $W \subset V$, $V$ can be decomposed into a finite direct sum of irreducible $W$ modules, we will write this as $A\in \calC$. By \cite{Huang:2014ixa} (and the classical results Theorem 5.2 of \cite{alex2001qanalog} and Section 6 of \cite{DMNO} in the holomorphic case), $A$ is a commutative, connected, and separable algebra object in $\calC$.\footnote{In fact, when $V$ is holomorphic $A$ is actually a Lagrangian algebra. This can be seen since $\mathrm{FPdim}(A)^2 \, \mathrm{FPdim}(\calC^0_A) = \mathrm{FPdim}(\calC)$, but $\mathrm{FPdim}(\calC^0_A) = \mathrm{FPdim}(\Rep(V)) = 1$. From the anyon-condensation point-of-view this is because the conformal inclusion is describing a gapped boundary for $\calC$, rather than a gapped interface between two phases $\calC$ and $\calD$.}
    
    The fusion category we seek is $\calF=\calC_A$, the category of right $A$-modules in $\calC$. In \cite{alex2001qanalog} this is called $\Rep A$, the category of twisted $V$-modules
    
    In the language of conformal nets, the proof runs parallel, with the algebra object $A$ now called a ``Q-system,'' and $\calF$ the ``category of solitons'' obtained by ``$\alpha$-induction'' \cite{kawahigashi2015, bischoff2015characterization, bischoff2016orbifold}.
\end{proof}
Generally, if $W \subset V$ as in the proposition, then we will write $W$ as $V^{\calF}$ and say that $\calF$ \textit{acts nicely} on $V$. Note if $\calF$ acts nicely on $V$ and $\calA$ is a fusion subcategory, then $\calA$ also acts nicely on $V$.

As an example, consider the case $W = V^G$ where $G$ is a finite group of automorphisms of a holomorphic VOA $V$, then $\calF \cong \Vc_G^\alpha$ as in the celebrated-results of \cite{Kirillov1, Kirillov2} (see \cite{theoMoonshine,bischoff2018conformal} for conformal nets).

Physically, $\calF$ is a collection of topological defect lines which acts on $V$ and $W$ is the sub-VOA of $V$ which commutes with the lines of $\calF$. More explicitly, if $X\in \calF$ is a topological defect line in $\calF$ which acts on $V$, then it has a vector space of (not necessarily topological) twist operators $\calH_{X}$ which $X$ can end on. Since any element of $\calH_{X}$ is not a true local operator, but must be attached to a topological $X$ tail, a generic operator of $V$ will collect $X$-monodromy if it encircles the endpoint operator; the operators of $W$ are exactly those which collect no monodromy.

In \cite{carnahan2016regularity}, the authors prove the following
\begin{proposition}[\cite{carnahan2016regularity}, Corollary 5.25]
    Let $V$ be a nice VOA and $G$ a finite solvable group of automorphisms of $V$ then $V^G$ is nice. 
\end{proposition}
\noindent In fact, the result they prove makes weaker ``niceness assumptions'' than what we've stated here. In the context of duality defects, we will be concerned with the case that $G$ is cyclic.

So suppose $V$ is holomorphic and $\calF = \TY(\bbZ_m)$ is a Tambara-Yamagami category which acts nicely on $V$. Note: here we are being ambiguous about the associator data. The degree-0 subcategory of $\calF$ defines a $\bbZ_m$ action on $V$ which acts nicely, so that $V^{\calF} \subset V^{\bbZ_m} \subset V$ is an inclusion of nice VOAs. By Proposition \ref{prop:confInc}, $V^{\bbZ_m} \subset V$ determines a fusion category isomorphic to $\Vc_{\bbZ_m}$ and so $V^{\calF} \subset V^{\bbZ_m}$ determines a fusion category of order $2$. Thus we have
\begin{equation}
    V^{\TY(\bbZ_m)} = (V^{\bbZ_m})^{\bbZ_2}\,,
\end{equation}
where the $\bbZ_2$ action on $V^{\bbZ_m}$ may be anomalous.

We see that the data of a $\TY(\bbZ_m)$ action on $V$ is defined by a choice of $\bbZ_m$ action on $V$ and $\bbZ_2$ action on $V^{\bbZ_m}$, but not all of the former actions realize a $\TY(\bbZ_m)$ action. As explained in Section \ref{sec:DefectedPFs}, our claim is that it does so if and only if the $\bbZ_2$ action ``fully swaps the axes'' of the metric Abelian group $A = \Irr(V^{\bbZ_m})$. By ``fully,'' we mean the $\bbZ_2$ action truly maps all elements of $A$ to elements of $\hat{A}$ and vice-versa, i.e., it doesn't just act non-trivially on a subgroup of order dividing $m$. 
\begin{theorem}\label{prop:proofTY}
    The inclusion $(V^{\bbZ_m})^{\bbZ_2} \subset V$ corresponds to a $\bbZ_m$ Tambara-Yamagami action on $V$ when the $\bbZ_2$ action fully swaps the axes of the metric Abelian group $A = \Irr(V^{\bbZ_m})$.
\end{theorem}
\begin{proof}
    First suppose it does not fully swap the axes, i.e. that $\bbZ_2 \subset SO(A,h)$, then the $\bbZ_2$ action is only acting by automorphisms on the individual axes, permuting the contents of $\bbZ_m$ and $\hat{\bbZ}_m$, but not amongst each other. So it is really just acting as an automorphism on $\bbZ_m$, hence there is a finite group action $G = \bbZ_m.\bbZ_2$ acting on $V$ such that $(V^{\bbZ_m})^{\bbZ_2} = V^G$.
    
    Now suppose we have an order 2 automorphism in $O(A,h)$ which does swap the axes. We must show that $\calF$ determined by $(V^{\bbZ_m})^{\bbZ_2} \subset V$ is a Tambara-Yamagami category. We know that $\calF$ is some $\bbZ_2$ extension of $\Vc_{\bbZ_m}$ with total dimension $2m$, so it will suffice to show that the part in degree 1 has one simple object (which will then be dimension $\sqrt{m}$). In other words, if we write $\calF = \Vc_{\bbZ_m} \oplus \calA$, then we want to show $\calA$ has exactly one simple object.
    
    $\calA$ is an invertible bimodule category for $\Vc_{\bbZ_m}$.\footnote{Note: physically, this is not saying that the TY line is invertible, but that the 2d wall which extends off of it into the (2+1)d bulk is invertible.} Now, bimodule categories of $\Vc_{\bbZ_m}$ are in bijection with module categories for $\calZ(\Vc_{\bbZ_m})$. In any case, module categories for $\Vc_G$ are labelled by $(H,\beta)$ where $H<G$ and  $\beta \in H^2(H;U(1))$ \cite{homotopyFusion}.
    
    Since $\calZ(\Vc_{\bbZ_m}) \cong \Vc_{\bbZ_m \times \hat{\bbZ}_m}$ \textit{as fusion categories}, the module categories are labelled as above. Physically, the $(H,\beta)$ data are labelling a topological boundary condition for the $\bbZ_m$ gauge theory described in our previous section. The $\bbZ_2$ action is defining an electric-magnetic duality which swaps our Dirichlet and Neumann boundary conditions, and we are convincing ourselves that at the end of the wall lives a $\bbZ_2$ twist-line with Tambara-Yamagami like properties.
    
    There is a distinguished module category given by our $\bbZ_2$ action. In particular, the $\bbZ_2$ axis-swapping action defines an isomorphism $\chi:\bbZ_m \to \hat{\bbZ}_m$ whose graph is our subgroup $H = \bbZ_m < \bbZ_m \times \hat{\bbZ}_m$. Physically speaking, this is telling us if the electric-magnetic duality is swapping the axes and sending things of ``electric charge $1$'' to things with ``magnetic charge $1$'' or also rearranging the definitions of electric and magnetic charge.
    
    Now the module category of $\calZ(\Vc_{\bbZ_m})$ defined by that diagonal $\bbZ_m$ (with trivial cocycle since $H^2(\bbZ_m,U(1))$ is trivial), gives us a bimodule category of $\Vc_{\bbZ_m}$ using the bulk-boundary functor $\calZ(\Vc_{\bbZ_m}) \to \Vc_{\bbZ_m}$. In particular, the data is simply the projection onto the first factor in $\pi: \bbZ_m \times \hat{\bbZ}_m \to \bbZ_m$. Following our particular module category $(H,1)$, we get a $\Vc_{\bbZ_m}$ module category (we don't need the full bimodule category to count simple objects). 
    
    The simple objects of $\calA$ are the simple module-objects for the $\Vc_{\bbZ_m}$ algebra object $\pi((H,1))$, but these correspond to subgroups $\bbZ_m < \bbZ_m$, of which there is only $1$. So $\calA$ has one object.
\end{proof}

\section{Conclusion}
\label{sec:OpenProblems}
In this paper we studied duality defects of the holomorphic $E_8$ lattice VOA $V$ which separate a theory from its orbifold by a cyclic subgroup $\bbZ_m$. We computed the defect partition functions explicitly as $\bbZ_2$ twists of the invariant sub-VOAs $V^{\bbZ_m}$ and compared our results to fermionization in the $m=2$ case and a particular conformal inclusion in the $m=3$ case (see Appendix \ref{sec:PottsDefect}), and obtained matching results. We also explained this computation using the (2+1)d topological field theory perspective of 2d orbifolds.

Focusing on $V$ gave us the ability to understand the group actions very explicitly. In fact, since the non-anomalous $\bbZ_m$ symmetries of $V$ are systematically enumerable, as well as $\bbZ_2$ actions of the invariant sub-VOA, we have actually \textit{classified all $\bbZ_m$ Tambara-Yamagami actions on the $E_8$ lattice VOA} (but see the open problems list). They are given as a pair consisting of a non-anomalous $\bbZ_m$ symmetry of $V$ and a (possibly anomalous) $\bbZ_2$ action which ``swaps the axes'' of the metric Abelian group of irreducible representations $\Irr(V^{\bbZ_m})$.

However, the (2+1)d discussion of duality defects did not actually rely very particularly on the details of $V$, so similar constructions should work for arbitrary holomorphic theories that are sufficiently ``nice,'' as demonstrated in Theorem \ref{prop:proofTY}. Generally, by the results of \cite{Gaiotto:2019xmp} (see also Appendix D of \cite{gaiotto2020orbifold}), we expect a modified version of this procedure to work for more complicated theories in higher dimensions (as in \cite{Koide:2021zxj, Choi:2021kmx, Kaidi:2021xfk}), as it's all simply a manifestation of the higher-dimensional TFT perspective on these topological defect operators.

\subsubsection*{Open Problems}
\begin{enumerate}
    \item \textbf{Full associator data}. In this paper, we do not obtain the full associator data for our Tambara-Yamagami categories. From a 2d perspective, data like the sign $\tau$ has an effect on the associativity relations for the duality defect line in the CFT. Based on our discussion, we suspect it is controlled by whether or not the $\bbZ_2$ action is anomalous or non-anomalous. Comparing to the (2+1)d TFTs, there is similarly a question of whether or not we stack with an invertible phase of matter before gauging the $\bbZ_2^{em}$ bulk $0$-form symmetry, as $H^3(\bbZ_2;U(1))=\bbZ_2$. It would be fun to work this data out explicitly and compare with the formalism of $G$-crossed categories \cite{gelaki2009centers, Barkeshli:2014cna}.
    \item \textbf{$n$-ality defects}. In \cite{thorngren2019fusion, thorngren2021fusion}, the authors discuss higher $n$-ality defects 
    $\mathcal{N}$ with fusion rules
    \begin{equation}
        X_g \otimes \mathcal{N} = \mathcal{N} = \mathcal{N} \otimes X_g\,,\quad \mathcal{N}^n = \sum_{g \in G} g\,.
    \end{equation}
    There is nothing particularly restrictive about our procedure for finding duality defects that could not be extended to these higher $n$-ality defects. In particular, for $E_8$, one could use Kac's theorem to construct automorphisms $G$ so that the weight-one subspace $(V_{E_8})^G$ has an $\mathfrak{so}(8)$ subalgebra, then try twisting by (a lift of) the $\bbZ_3$ automorphism of the $D_4$ Dynkin diagram. This would give a triality defect of a very different origin than those discussed in \cite{thorngren2021fusion}.
    \item \textbf{Matching to and from other examples}. To gain physical insight into other theories, the $E_8$ theory, or the structure of duality defects more broadly, it could be helpful to compare the constructions of defects in more specific cases. For example, using parafermionization or playing games with conformal inclusions as in Appendix \ref{sec:PottsDefect}. The fact that factors of $\mathfrak{u}(1)$ appear in the weight-one subspaces of various fixed VOAs $((V^{\bbZ_m})^{\bbZ_2})_1$ indicates the possibility of explicit character decompositions mirroring those in Appendix \ref{sec:PottsDefect} but with free bosons attached.
    \item \textbf{Other Theories.} We focused on the $E_8$ VOA, but our strategy should apply to arbitrary holomorphic VOAs without extra complications, although the lattice descriptions may not be available. From the discussion in Section \ref{sec:3dTFT}, we actually expect it to apply to duality interfaces in more general theories.
\end{enumerate}

\acknowledgments
We would like to thank Theo Johnson-Freyd for suggesting this project, for immeasurablely helpful comments throughout the research process, and feedback at various stages of the draft. We would also like to thank Sven M\"oller for providing us with Magma code to perform twisted character calculations and for explanations of the VOA literature. JK would also like to thank Davide Gaiotto for illuminating discussions on 2d CFT and (2+1)d TFT and feedback on the draft; and Matt Yu for discussions on anyon condensation. This work was initiated during the 2020 PSI Winter School, we thank the organizers of the PSI program and the staff at Camp Kintail for hospitality during this time. We are particularly grateful for the participation of Alicia Lima and Melissa Rodr\'iguez-Z\'arate in the early stages of the project. IMB would like to acknowledge financial support from the Berkeley Graduate Division and the Friends of Warren Hellman Fund. JK is funded through the NSERC CGS-D program. Research at Perimeter Institute is supported in part by the Government of Canada through the Department of Innovation, Science and Economic Development Canada and by the Province of Ontario through the Ministry of Colleges and Universities. JN thanks the Max Weber Stiftung and the Hanns-Seidel Stiftung for financial and intellectual support. 

\appendix
\section{A \texorpdfstring{$\bbZ_3$}{Z3} Defect from the Potts Model}\label{sec:PottsDefect}
In Section \ref{sec:fermionization} we obtained our $\bbZ_2$ duality defects from fermionization to contrast it to the Lie-theoretic method. We can also obtain other duality defects by playing similar games with parafermions and/or various conformal inclusions. 

We will explain how to obtain one of our $\bbZ_3$ duality defects from the Potts Model. This largely just amounts to doing representation theory for 2d CFTs, but we take a (2+1)d TFT and anyon condensation viewpoint to highlight the story in Section \ref{sec:3dTFT}. The final answer is given in Equation \eqref{eq:pottsDefect}.

The (bosonic) $c=\frac{4}{5}$ minimal models are constructed by pairing data from the chiral algebra $m_{6,5}$ and the anti-chiral algebra $\overline{m}_{6,5}$ in a modular invariant way. $m_{6,5}$ has 10 chiral primaries with spins \cite{BPZ}
\begin{equation}
    \{{0},
    {\tfrac{2}{5}},
    {\tfrac{1}{40}},
    {\tfrac{7}{5}},
    {\tfrac{21}{40}},
    {\tfrac{1}{15}},
    {3},
    {\tfrac{13}{8}},
    {\tfrac{2}{3}},
    {\tfrac{1}{8}}
    \}\,.
\end{equation}

The chiral algebra data may pair up via the diagonal (or ``A-type'') modular invariant, which corresponds to the tetracritical Ising model, with 10 full CFT primaries and a $\bbZ_2$ symmetry group. Alternatively, there is also a non-diagonal (or ``D-type'') modular invariant, which corresponds to the critical three-state Potts model, with 12 primaries (we use the notation of \cite{Chang_2019})
\begin{equation}
    \mathds{1}_{0,0}, 
    \epsilon_{\frac{2}{5},\frac{2}{5}},
    X_{\frac{7}{5},\frac{7}{5}},
    Y_{3,3},
    \Phi_{\frac{7}{5},\frac{2}{5}},
    \tilde{\Phi}_{\frac{2}{5},\frac{7}{5}},
    \Omega_{3,0},
    \tilde{\Omega}_{0,3},
    \sigma_{\frac{1}{15},\frac{1}{15}},
    \sigma^*_{\frac{1}{15},\frac{1}{15}},
    Z_{\frac{2}{3},\frac{2}{3}},
    Z^*_{\frac{2}{3},\frac{2}{3}}\,,
\end{equation}
and an $S_3$ symmetry group. The two are related by $\bbZ_2$ orbifold. 

By the construction of \cite{kapustin2010surface}, the Potts CFT on $\Sigma$ can be blown up to a slab of (2+1)d TFT on $\Sigma\times I$ with bulk $\calM := \Rep(m_{6,5})$. The TFT assigns the vector space $\calH_\Sigma$ of conformal blocks of $m_{6,5}$ to the chiral boundary (and $\overline{\calH}_\Sigma$ to the other).\footnote{We label the boundaries of the (2+1)d TFT by the chiral algebras to avoid overcluttering.} From this point of view, the choice of $D$-type modular invariant is specified by a topological interface $I_{D(6,5)}$,\footnote{One can quickly obtain $I_{D(6,5)}$ as one the only two (irreducible bosonic) gapped boundary conditions of $\calM \times \overline{\calM\,}$ using the Lagrangian algebra discussion outlined in Section \ref{sec:3dTFT}. The interface $I_{D(6,5)}$ satisfies $I_{D(6,5)}^2 = 2 I_{D(6,5)}$.} so that we may draw a picture like
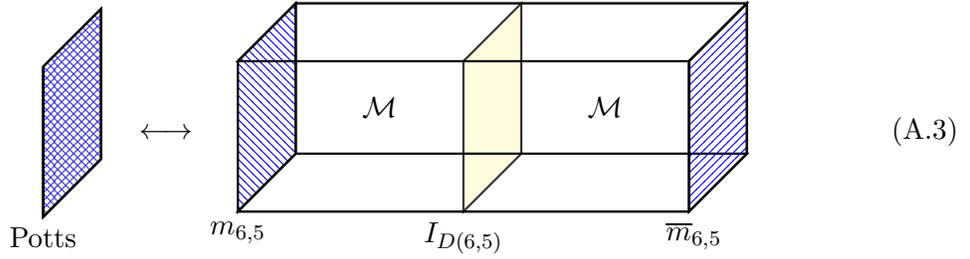
\begin{equation}
\begin{tikzpicture}[thick, baseline={([yshift=-.5ex]current bounding box.center)}]
	\def\Depth{6}
	\def\Height{2}
	\def\Width{2}
	\def\Sep{3}        
	
	\coordinate (O) at (0,0,0);
	\coordinate (A) at (0,\Width,0);
	\coordinate (B) at (0,\Width,\Height);
	\coordinate (C) at (0,0,\Height);
	
	\draw[black, fill=yellow!20,opacity=0.8] (O) -- (A) -- (B) -- (C) -- cycle;

    \draw[pattern=crosshatch, pattern color=blue!80] (O) -- (A) -- (B) -- (C) -- cycle;
	
	\draw[black] (O) -- (A) -- (B) -- (C) -- cycle;
	
	\draw[below] (C) node{Potts};
    
\end{tikzpicture}
\quad\longleftrightarrow
\begin{tikzpicture}[thick, baseline={([yshift=-.5ex]current bounding box.center)}]
	\def\Depth{6}
	\def\Height{2}
	\def\Width{2}
	\def\Sep{3}        
	
	\coordinate (O) at (0,0,0);
	\coordinate (A) at (0,\Width,0);
	\coordinate (B) at (0,\Width,\Height);
	\coordinate (C) at (0,0,\Height);
	\coordinate (D) at (\Depth,0,0);
	\coordinate (E) at (\Depth,\Width,0);
	\coordinate (F) at (\Depth,\Width,\Height);
	\coordinate (G) at (\Depth,0,\Height);
	\coordinate (I1) at (\Depth/2,0,0);
	\coordinate (I2) at (\Depth/2,\Width,0);
	\coordinate (I3) at (\Depth/2,\Width,\Height);
	\coordinate (I4) at (\Depth/2,0,\Height);
	
	\draw[black, fill=yellow!20,opacity=0.8] (I1) -- (I2) -- (I3) -- (I4) -- cycle;
	\draw[below] (I4) node{$I_{D(6,5)}$};

	\draw[midway] (\Depth/4,\Width-\Width/2,\Height/2) node {$\calM$};
	\draw[midway] (3*\Depth/4,\Width-\Width/2,\Height/2) node {$\calM$};

    \draw[pattern=north west lines, pattern color=blue] (O) -- (A) -- (B) -- (C) -- cycle;
	
	\draw[black] (O) -- (C) -- (G) -- (D) -- cycle;
	\draw[black] (O) -- (A) -- (E) -- (D) -- cycle;
	\draw[black] (O) -- (A) -- (B) -- (C) -- cycle;
	\draw[black] (D) -- (E) -- (F) -- (G) -- cycle;
	\draw[black] (C) -- (B) -- (F) -- (G) -- cycle;
	\draw[black] (A) -- (B) -- (F) -- (E) -- cycle;
	\draw[below] (0, 0*\Width, \Height) node{$m_{6,5}$};
	\draw[below] (\Depth, 0*\Width, \Height) node{$\,\,\overline{m}_{6,5}$};
	
    \draw[pattern=north east lines, pattern color=blue] (D) -- (E) -- (F) -- (G) -- cycle;
\end{tikzpicture}
\end{equation}
where ``$\longleftrightarrow$'' means ``blows up to'' in one direction and ``compactifies down to'' in the other.

The untwisted (torus) partition function for the three-state Potts model describes this $D$-type pairing, it is given by
\begin{equation}\label{eq:Potts}
    Z_{\mathrm{Potts}}[0,0]   = \abs{\chi_{1,1}+\chi_{4,1}}^2 
        + \abs{\chi_{2,1}+\chi_{3,1}}^2
        + 2 \abs{\chi_{4,3}}^2
        + 2 \abs{\chi_{3,3}}^2\,.
\end{equation}
However, it is also well known that the Potts model can be realized as the diagonal modular invariant for the $W_3$ algebra
\begin{equation}
    Z_{W_3}[0,0]   = \abs{\chi_{\mathds{1}}}^2 
        + \abs{\chi_{\epsilon}}^2
        + \abs{\chi_{Z}}^2
        + \abs{\chi_{Z^*}}^2
        + \abs{\chi_{\sigma}}^2
        + \abs{\chi_{\sigma^*}}^2\,.
\end{equation}

The fact that the Potts model can also be written as the diagonal modular invariant CFT for a $W_3$ algebra signals that there exists an anyon condensation wall between $\calM$ and $\calW := \Rep(W_3)$, condensing the $10$ anyons of $\calM$ down to the $6$ bulk anyons of $\calW$ (as we will confirm). Pictorially, there exists a topological interface $I_{\calM\vert\calW}$ such that
\begin{equation}
\begin{tikzpicture}[thick, baseline={([yshift=-.5ex]current bounding box.center)}]
	\def\Depth{6}
	\def\Height{2}
	\def\Width{2}
	\def\Sep{3}        
	
	\coordinate (O) at (0,0,0);
	\coordinate (A) at (0,\Width,0);
	\coordinate (B) at (0,\Width,\Height);
	\coordinate (C) at (0,0,\Height);
	\coordinate (D) at (\Depth,0,0);
	\coordinate (E) at (\Depth,\Width,0);
	\coordinate (F) at (\Depth,\Width,\Height);
	\coordinate (G) at (\Depth,0,\Height);
	\coordinate (I1) at (\Depth/2,0,0);
	\coordinate (I2) at (\Depth/2,\Width,0);
	\coordinate (I3) at (\Depth/2,\Width,\Height);
	\coordinate (I4) at (\Depth/2,0,\Height);
	
	\draw[black, fill=yellow!20] (I1) -- (I2) -- (I3) -- (I4) -- cycle;
	\draw[black, pattern=north east lines, pattern color=blue!80] (I1) -- (I2) -- (I3) -- (I4) -- cycle;

	\draw[below] (I4) node{$I_{\calM\vert\calW}$};

	\draw[midway] (\Depth/4,\Width-\Width/2,\Height/2) node {$\calM$};
	\draw[midway] (3*\Depth/4,\Width-\Width/2,\Height/2) node {$\calW$};

    \draw[pattern=north west lines, pattern color=blue] (O) -- (A) -- (B) -- (C) -- cycle;
	
	\draw[black] (O) -- (C) -- (G) -- (D) -- cycle;
	\draw[black] (O) -- (A) -- (E) -- (D) -- cycle;
	\draw[black] (O) -- (A) -- (B) -- (C) -- cycle;
	\draw[black] (D) -- (E) -- (F) -- (G) -- cycle;
	\draw[black] (C) -- (B) -- (F) -- (G) -- cycle;
	\draw[black] (A) -- (B) -- (F) -- (E) -- cycle;
	\draw[below] (0, 0*\Width, \Height) node{$m_{6,5}$};
	\draw[below] (\Depth, 0*\Width, \Height) node{$\,\,\overline{W}_3$};
	
    \draw[pattern=north east lines, pattern color=green] (D) -- (E) -- (F) -- (G) -- cycle;
\end{tikzpicture}
\end{equation}

The 10 lines of $\calM = \Rep(m_{6,5})$ are
\begin{equation}
    \{0_{0},
    1_{\frac{2}{5}},
    2_{\frac{1}{40}},
    3_{\frac{7}{5}},
    4_{\frac{21}{40}},
    5_{\frac{1}{15}},
    6_{3},
    7_{\frac{13}{8}},
    8_{\frac{2}{3}},
    9_{\frac{1}{8}}
    \}\,,
\end{equation}
where the subscript denotes the chiral spin of the operator they correspond to (but will henceforth be dropped). Note that spin in the (2+1)d TFT is only meaningful $\Mod 1$ so that lines $0$ and $6$ have the same spin. At this point, there's only one bosonic condensation we could even try to perform; it must be that the new vacuum line is $\varphi = 0 + 6$. This makes sense since the spin-3 operator generates the $W_3$ algebra. 

To understand the condensation at the interface $I_{\calM\vert\calW}$, we compute the fusion rules of the (non-trivial) lines in $\calM$ with the new vacuum $\varphi$
\begin{alignat}{3}
    &\varphi \times 1 = 1+3\,, \qquad 
        &&\varphi \times2 = 2 + 4\,, \qquad
        &&\varphi \times3 = 3 + 1\,,\nonumber\\
    &\varphi \times4 = 4+2\,, \qquad 
        &&\varphi \times5 = 5_1 + 5_2\,, \qquad
        &&\varphi \times6 = 6 + 0\,,\\
    &\varphi \times7 = 7+9 \qquad 
        &&\varphi \times8 = 8_1 + 8_2\,, \qquad
        &&\varphi \times9 = 9 + 7\,.\nonumber
\end{alignat}
Fusions like $\varphi \times 1 = 1 + 3$ are straightforward, as the fusion rules of lines in $\calM$ match those of primaries in $m_{6,5}$. One subtlety is the appearance of terms like $\varphi\times 5 = 5_1 + 5_2$, where the subscript reminds us that the line formerly known as $5$ splits into two simple objects on $I_{\calM\vert\calW}$. Altogether, the lines living on $I_{\calM\vert\calW}$, coming from the condensation, are
\begin{equation}
    \{\varphi, (1+3), (2+4), 5_1, 5_2, (7+9), 8_1, 8_2\}\,.
\end{equation}
We note that the lines $(2+4)$ and $(7+9)$ cannot be lifted to the bulk $\calW$ phase because the constituent anyons have different spins and so the combined object cannot be given a braiding; in the language of \cite{mattAnyon} they are ``totally confined.'' The anyons of the phase $\calW$ must be
\begin{equation}
    \calW = \{
        \varphi_0, 
        (1+3)_{\frac{2}{5}}, 
        (5_1)_{\frac{1}{15}}, 
        (5_2)_{\frac{1}{15}}, 
        (8_1)_{\frac{2}{3}}, 
        (8_2)_{\frac{2}{3}}\}\,,
\end{equation}
which matches what we know about the $W_3$ algebra.

One can now verify that $\calW$ has the same fusion rules, but opposite spins, as the category $\calC := \Rep(\mathfrak{su}(3)_1\times(\mathfrak{f}_4)_1)$. In particular, the spins are
\begin{alignat}{2}
    (\mathfrak{f}_4)_1&: \quad &&\{0, \tfrac{3}{5}\}\,,\\
    \mathfrak{su}(3)_1&: \quad &&\{0, \tfrac{1}{3}, \tfrac{1}{3}\}\,,\\
    \calC&: \quad && \{0, \tfrac{1}{3},\tfrac{1}{3}, 
    \tfrac{3}{5}, \tfrac{14}{15}, \tfrac{14}{15}\}\,,
\end{alignat}
and we have $\calC = \overline{\calW}$. 

This means we can draw a picture like
\begin{equation}
\begin{tikzpicture}[thick, baseline={([yshift=-.5ex]current bounding box.center)}]
	\def\Depth{6}
	\def\Height{2}
	\def\Width{2}
	\def\Sep{3}        
	
	\coordinate (O) at (0,0,0);
	\coordinate (A) at (0,\Width,0);
	\coordinate (B) at (0,\Width,\Height);
	\coordinate (C) at (0,0,\Height);
	\coordinate (D) at (\Depth,0,0);
	\coordinate (E) at (\Depth,\Width,0);
	\coordinate (F) at (\Depth,\Width,\Height);
	\coordinate (G) at (\Depth,0,\Height);
	\coordinate (I1) at (\Depth/2,0,0);
	\coordinate (I2) at (\Depth/2,\Width,0);
	\coordinate (I3) at (\Depth/2,\Width,\Height);
	\coordinate (I4) at (\Depth/2,0,\Height);
	
	\draw[fill=pink!40] (I4) -- (I3) -- (F) -- (G) -- cycle;
	\draw[fill=pink!40] (I2) -- (I3) -- (F) -- (E) -- cycle;
	
	\draw[black, fill=yellow!20] (I1) -- (I2) -- (I3) -- (I4) -- cycle;
	\draw[black, pattern=north east lines, pattern color=blue!80] (I1) -- (I2) -- (I3) -- (I4) -- cycle;

	\draw[below] (I4) node{$I_{\calM\vert\overline\calC}$};

	\draw[midway] (\Depth/4,\Width-\Width/2,\Height/2) node {$\calM$};
	\draw[midway] (3*\Depth/4,\Width-\Width/2,\Height/2) node {$ \mathscr{G} \times \overline{\calC}$};

    \draw[pattern=north west lines, pattern color=blue] (O) -- (A) -- (B) -- (C) -- cycle;
	
	\draw[black] (O) -- (C) -- (G) -- (D) -- cycle;
	\draw[black] (O) -- (A) -- (E) -- (D) -- cycle;
	\draw[black] (O) -- (A) -- (B) -- (C) -- cycle;
	\draw[black] (D) -- (E) -- (F) -- (G) -- cycle;
	\draw[black] (C) -- (B) -- (F) -- (G) -- cycle;
	\draw[black] (A) -- (B) -- (F) -- (E) -- cycle;
	\draw[below] (0, 0*\Width, \Height) node{$m_{6,5}$};
	\draw[below] (\Depth, 0*\Width, \Height) node{$\,\,\mathfrak{su}(3)_1 \times (\mathfrak{f}_4)_1$};
	
    \draw[pattern=north west lines, pattern color=green] (D) -- (E) -- (F) -- (G) -- cycle;
\end{tikzpicture}
\end{equation}
where we have flooded the right-hand side with $\mathscr{G}$, which denotes the $k=-1$ gravitational Chern-Simons theory, to soak up the gravitational anomaly \cite{Alvarez-Gaume:1983ihn, Kaidi:2021gbs}. Recall $c(m_{6,5}) = +4/5$, $c(\mathfrak{su}(3)_1) = +3$, and $c((\mathfrak{f}_4)_1) = +26/5$.

By folding just the $\overline{\calC}$ phase, we obtain the way to blow up the CFT $V_{E_8}$ as $m_{6,5} \times \mathfrak{su}(3)_1 \times (\mathfrak{f_4})_1$ with a topological boundary condition (or an interface to a gravitational Chern-Simons phase) on the other end
\begin{equation}
\begin{tikzpicture}[thick, baseline={([yshift=-.5ex]current bounding box.center)}]
	\def\Depth{2}
	\def\Height{2}
	\def\Width{2}
	\def\Sep{3}        
	
	\coordinate (O) at (0,0,0);
	\coordinate (A) at (0,\Width,0);
	\coordinate (B) at (0,\Width,\Height);
	\coordinate (C) at (0,0,\Height);
	\coordinate (D) at (\Depth,0,0);
	\coordinate (E) at (\Depth,\Width,0);
	\coordinate (F) at (\Depth,\Width,\Height);
	\coordinate (G) at (\Depth,0,\Height);
	\coordinate (I1) at (0,0,0);
	\coordinate (I2) at (0,\Width,0);
	\coordinate (I3) at (0,\Width,\Height);
	\coordinate (I4) at (0,0,\Height);
	
	\draw[fill=pink!40] (I4) -- (I3) -- (F) -- (G) -- cycle;
	\draw[fill=pink!40] (D) -- (E) -- (F) -- (G) -- cycle;
	\draw[fill=pink!40] (I2) -- (I3) -- (F) -- (E) -- cycle;

	\draw[midway] (\Depth/2,\Width-\Width/2,\Height/2) node {$\mathscr{G}$};

    \draw[black, fill=yellow!20,opacity=0.8] (O) -- (A) -- (B) -- (C) -- cycle;
    \draw[pattern=north west lines, pattern color=blue!80] (O) -- (A) -- (B) -- (C) -- cycle;
	\draw[black] (O) -- (A) -- (B) -- (C) -- cycle;
	\draw[below] (C) node{$V_{E_{8}}$};
	
	\draw[black] (O) -- (C) -- (G) -- (D) -- cycle;
	\draw[black] (O) -- (A) -- (E) -- (D) -- cycle;
	\draw[black] (O) -- (A) -- (B) -- (C) -- cycle;
	\draw[black] (D) -- (E) -- (F) -- (G) -- cycle;
	\draw[black] (C) -- (B) -- (F) -- (G) -- cycle;
	\draw[black] (A) -- (B) -- (F) -- (E) -- cycle;
	\draw[below] (0, 0*\Width, \Height) node{};
	\draw[below] (\Depth, 0*\Width, \Height) node{};
	
\end{tikzpicture}
\quad\longleftrightarrow
\begin{tikzpicture}[thick, baseline={([yshift=-.5ex]current bounding box.center)}]
	\def\Depth{6}
	\def\Height{2}
	\def\Width{2}
	\def\Sep{3}        
	
	\coordinate (O) at (0,0,0);
	\coordinate (A) at (0,\Width,0);
	\coordinate (B) at (0,\Width,\Height);
	\coordinate (C) at (0,0,\Height);
	\coordinate (D) at (\Depth,0,0);
	\coordinate (E) at (\Depth,\Width,0);
	\coordinate (F) at (\Depth,\Width,\Height);
	\coordinate (G) at (\Depth,0,\Height);
	\coordinate (I1) at (\Depth/2,0,0);
	\coordinate (I2) at (\Depth/2,\Width,0);
	\coordinate (I3) at (\Depth/2,\Width,\Height);
	\coordinate (I4) at (\Depth/2,0,\Height);
	
	\draw[fill=pink!40] (I4) -- (I3) -- (F) -- (G) -- cycle;
	\draw[fill=pink!40] (D) -- (E) -- (F) -- (G) -- cycle;
	\draw[fill=pink!40] (I2) -- (I3) -- (F) -- (E) -- cycle;
	
	\draw[black, fill=yellow!20] (I1) -- (I2) -- (I3) -- (I4) -- cycle;
	\draw[black, pattern=north east lines, pattern color=blue!80] (I1) -- (I2) -- (I3) -- (I4) -- cycle;

	\draw[below] (I4) node{$B_{\calM\times\calC}$};

	\draw[midway] (\Depth/4,\Width-\Width/2,\Height/2) node {$\calM \times \calC$};
	\draw[midway] (3*\Depth/4,\Width-\Width/2,\Height/2) node {$\mathscr{G}$};

    \draw[pattern=north west lines, pattern color=blue!50] (O) -- (A) -- (B) -- (C) -- cycle;
    \draw[pattern=north west lines, pattern color=green!50] (O) -- (A) -- (B) -- (C) -- cycle;
	
	\draw[black] (O) -- (C) -- (G) -- (D) -- cycle;
	\draw[black] (O) -- (A) -- (E) -- (D) -- cycle;
	\draw[black] (O) -- (A) -- (B) -- (C) -- cycle;
	\draw[black] (D) -- (E) -- (F) -- (G) -- cycle;
	\draw[black] (C) -- (B) -- (F) -- (G) -- cycle;
	\draw[black] (A) -- (B) -- (F) -- (E) -- cycle;
	\draw[below] (0, 0*\Width, \Height) node{\begin{tabular}{c} $m_{6,5} \, \times$ \\ $\mathfrak{su}(3)_1 \times (\mathfrak{f}_4)_1$  \end{tabular}};
	\draw[below] (\Depth, 0*\Width, \Height) node{};
	
\end{tikzpicture}
\end{equation}

In terms of characters, all these pictures have just been to say that
\begin{align}
    \chi^{E_8}(\tau) 
        &= (\chi_{1,1}(\tau)  + \chi_{4,1}(\tau) )\,\chi_0^{A_2}(\tau)\, \chi_0^{F_4}(\tau)
        + (\chi_{2,1}(\tau) + \chi_{3,1}(\tau))\,\chi_0^{A_2}(\tau)\,\chi_{\frac{3}{5}}^{F_4}(\tau)\nonumber\\
        & + 2\, \chi_{4,3}(\tau)\, \chi_{\frac{1}{3}}^{A_2}(\tau)\, \chi_0^{F_4}(\tau)
        + 2\, \chi_{3,3}(\tau)\, \chi_{\frac{1}{3}}^{A_2}(\tau)\, \chi_{\frac{3}{5}}^{F_4}(\tau)
\end{align}
where we recognize the shadow of the original Potts partition function from Equation \eqref{eq:Potts}. For implementation purposes, it may be helpful to know that the $(\mathfrak{f}_4)_1$ characters can similarly be written in terms of $\mathfrak{so}(9)_1$ characters (given in \eqref{eq:BnCharacters}) and the $c=\frac{7}{10}$ tricritical ising characters (see \cite{lopes2019non} for a cool application)
\begin{align}
    \chi_0^{F_4}(\tau)
        &=\chi^{m_{5,4}}_{1,1}(\tau)\,\chi^{B_4}_{0}(\tau)
        +\chi^{m_{5,4}}_{2,1}(\tau)\,\chi^{B_4}_{\frac{1}{16}}(\tau)
        +\chi^{m_{5,4}}_{3,1}(\tau)\,\chi^{B_4}_{\frac{1}{2}}(\tau)\,,
        \\
    \chi_{\frac{3}{5}}^{F_4}(\tau)
        &=\chi^{m_{5,4}}_{3,2}(\tau)\,\chi^{B_4}_{0}(\tau)
        +\chi^{m_{5,4}}_{2,2}(\tau)\,\chi^{B_4}_{\frac{1}{16}}(\tau)
        +\chi^{m_{5,4}}_{3,3}(\tau)\,\chi^{B_4}_{\frac{1}{2}}(\tau)\,.
\end{align}
    
In \cite{Petkova:2000} the authors enumerate the simple topological defect lines of the Potts model, with defect partition functions. We can use our character decompositions to similarly produce a $\bbZ_3$ defect partition function for $E_8$. There are actually two $\bbZ_3$ duality defect lines in the Potts model, related by charge conjugation (see \cite{Chang_2019}), they have the same defect partition function
\begin{align}
    Z_{\mathrm{Potts}}[0,X_{\TY}] 
        &= \sqrt{3} 
    (\chi_{1,1}(\tau)+\chi_{1,5}(\tau))(\bar\chi_{1,1}(\bar\tau)-\bar\chi_{1,5}(\bar\tau))\nonumber\\
        &+\sqrt{3} 
    (\chi_{3,1}(\tau)+\chi_{3,5}(\tau))(\bar\chi_{3,1}(\bar\tau)-\bar\chi_{3,5}(\bar\tau))
\end{align}
So that translated to $V_{E_8}$ we have
\begin{align}\label{eq:pottsDefect}
    Z_{V_{E_8}}[0,X_{\TY}] 
        &= \sqrt{3} 
    (\chi_{1,1}(\tau)+\chi_{1,5}(\tau))\,\chi_0^{A_2}(\tau)\,\chi_0^{F_4}(\tau)\nonumber\\
        &+\sqrt{3} 
    (\chi_{3,1}(\tau)+\chi_{3,5}(\tau))\,\chi_0^{A_2}(\tau)\,\chi_{\frac{3}{5}}^{F_4}(\tau)
\end{align}
whose $q$-expansion matches our result in Equation \eqref{eq:Z3Defect4}.

With hindsight, the path to this result is somewhat obvious. Kac's theorem tells us about the weight one subspace of $(V^{\bbZ_3})^{\bbZ_2}$, one of them happens to be $A_2 \times F_4$ (where the $F_4$ appeared as the invariant subalgebra under the $\bbZ_2$ Dynkin diagram automorphism for $E_6$). If we assume these are both reflecting the presence of level-1 WZW models, then we only have to explain a missing central charge of $\frac{4}{5}$, and we know that the Potts model has a $\bbZ_3$ symmetry and is self-dual under $\bbZ_3$ orbifold (because there's no other theory for it to map to with $\bbZ_3$ symmetry).

As a related exercise, one can reobtain the results from fermionization with their associator data by noting that $\TY(\bbZ_2,1,+1/\sqrt{2})\cong \mathrm{Ising}$ and $\TY(\bbZ_2,1,-1/\sqrt{2})\cong \mathfrak{su}(2)_2$ as fusion categories and using techniques of anyon condensation.

Such manipulations would become rather cumbersome if one wants to enumerate all duality defects. The treatment given in the main text provides a more systematic way to obtain these same results. A combined approach could be of physical interest to better understand the final expressions for the duality defects in terms of the characters of $(V^{\bbZ_m})^{\bbZ_2}$. 

\section{Symmetric Non-Degenerate Bicharacters and ``Rearrangement Data''}\label{sec:Rearrangements}
Tambara-Yamagami categories are specified by an Abelian group $G$, a symmetric non-degenerate bicharacter $\chi:G\times G \to U(1)$, and a Frobenius-Schur indicator $\tau = \pm \sqrt{\abs{G}}$. In the main body of this paper we largely focus on the fusion rules, but the symmetric non-degenerate bicharacter is also a mostly accessible piece of data.

Recall that a symmetric non-degenerate bicharacter gives isomorphisms $\chi_G:G\to\hat{G}$ and $\chi_{\hat{G}}:\hat{G} \to G$ such that $\chi_{\hat{G}}\circ \chi_{G} = \id_{G}$ and $\chi_G \circ \chi_{\hat{G}} = \id_{\hat{G}}$. In our case, $G=\bbZ_m$, and $\Irr(V^{\bbZ_m}) \cong \bbZ_m \times \hat{\bbZ}_m$ is a metric Abelian group (with metric function given by the natural pairing between $\bbZ_m$ and $\hat{\bbZ}_m$).

In this paper, we look for $\bbZ_2$ symmetries of $V^G$ which act to ``swap the axes'' of $A$. The data of a symmetric non-degenerate bicharacter describes whether or not the axes are also ``rearranged'' when they are swapped. To what extent can we obtain this ``rearrangement data'' of our $\bbZ_2$ action?

There are many ``non-canonicalities'' involved in this axis-swapping. For example, $\bbZ_m \cong \hat{\bbZ}_m$, but not-canonically, with $\abs{\bbZ_m^\times}$ isomorphisms between them. Even then, when dealing with cyclic group symmetries, there isn't really a meaning of ``one'' unit of electric charge. We claim that we can distinguish, at minimum, if two different $\bbZ_2$ actions on $V^G$ lead to different symmetric non-degenerate bicharacters.

We are dealing with some concrete data, we have:
\begin{enumerate}
    \itemsep0em
    \item An ``original'' $E_8$ lattice $L$.
    \item A sublattice $L_0 \subseteq L$ which is obtained as a fixed sublattice of $L$ under our original (non-anomalous) $\bbZ_m$ action.
    \item A dual lattice $L_0^*$ such that $L \subseteq L_0^*$ and $L_0^*/L_0 \cong A$.
    \item Another $E_8$ lattice $L^\prime \subseteq L_0^*$ such that $L \cap L^\prime = L_0$.
\end{enumerate}
The first piece of data in some sense explains if we can identify $G$ or $\hat{G}$ inside of $A$, as we can identify if some element is in $L$. But this will not be of primary concern.

Let $\{\sigma_i\}$ be a collection of order 2 automorphisms of $L_0$ (whose actions extend to all of $\bbR^8$ by linearity), and $\{\hat{\sigma}_i\}$ the (well-defined) induced actions on $A = L_0^*/L_0$. We will only concern ourselves with those $\sigma_i$ such that $\hat{\sigma}_i$ swaps the axes of $A$. We can determine if they correspond to different symmetric non-degenerate bicharacters as follows:
\begin{enumerate}
    \itemsep0em
    \item We start with the lattice $L_0$ and obtain $L_0^*/L_0$. In particular, we obtain a list of representatives for the $m^2$ equivalence classes. We can identify elements corresponding to equivalence classes ``on the axes'' by using the discriminant (which is well-defined on the equivalence classes), i.e. checking which representatives are even vectors.
    \item Next we pick some representative $\xi$ with discriminant $0$ and order $m$, and declare that the coset $\xi + L_0$ generates the subgroup corresponding to $L$ in $A$. We define things in the sector $\xi + L_0$ to have reference electric charge $1$.
    \item Now we pick one of our $\bbZ_2$ actions, $\sigma_1$ say, and declare that the operators in $\omega+L_0$ have reference magnetic charge $1$, where $\omega = \sigma_1(\xi)$. In other words, $\hat{\sigma}_1$ maps the sector with electric charge 1 to the sector with magnetic charge 1. This part is non-canonical in that $\sigma_1$ picks a particular isomorphism $\hat{\sigma}_1:G\to \hat{G}$ inside $A$.
    \item Finally, we act with all other $\sigma_i$ on $\xi$ to find the  ``image magnetic charge'' relative to $(\xi,\sigma_1)$. i.e. $\hat{\sigma}_i(\xi+L_0)=m_i \omega + L_0$ for some $m_i \in \bbZ_m^\times$.
\end{enumerate}
We see that relative to the declarations that $\xi$ has electric charge $1$ and $\sigma_1$ maps it to magnetic charge $1$, $m_i$ tells us whether or not $\sigma_i$ also ``rearranges'' the axes when it swaps them, and so constitutes the data we want.

To what extent do our choices matter? On one hand, note that if we had chosen a different reference electric charge $\xi^\prime$, then $\sigma_1$ would define a different generator $\omega^\prime + L_0 = \hat\sigma_1(\xi^\prime+L_0)$ for the magnetic subgroup. But $\xi^\prime + L_0 = k \xi + L_0$ for some invertible $k \Mod m$. Thus, multiplying both sides by $k$, we see $\hat\sigma_i(\xi+L_0) = m_i\omega +L_0$ if and only if $\hat\sigma_i(\xi^\prime+L_0) = m_i\omega^\prime +L_0$ since $k$ is invertible. So it's not our choice of reference electric charge sector $\xi+L_0$ that matters, just $\sigma_1$.

On the other hand, $m_i$ does depend on our original choice of $\sigma_1$ in defining the sector with one unit of magnetic charge. This is obvious since $m_i=1$ relative to itself. We will denote this ``rearrangement number'' relative to $\sigma_1$ as $\mathrm{Rearr}_{\sigma_1}(\hat{\sigma}_i) = m_i$. 

We can do slightly better though, as ratios of $m_i$ are meaningful independent of $\sigma_1$. To see this is the same as before. Fix some $\xi$ and suppose we have two reference isomorphisms $\sigma$ and $\tilde{\sigma}$ with their respective $\omega$ and $\tilde{\omega}$, and suppose $\hat{\mu}$ is some $\bbZ_2$ action of interest. Then relative to $\sigma$ and $\tilde{\sigma}$ we have
\begin{equation}
    \mathrm{Rearr}_\sigma(\hat{\mu}) =: m_\mu\,,\quad 
    \mathrm{Rearr}_{\tilde{\sigma}}(\hat{\mu}) =: \tilde{m}_\mu\,.
\end{equation}
Then, by definition, this means $\hat\mu(\xi+L_0) = \tilde{m}_\mu \tilde\omega + L_0 = m_\mu \omega + L_0$. But $\tilde \omega + L_0 = k \omega + L_0$ for some invertible $k$, so $m_\mu = k \tilde{m}_\mu$ and the statement about ratios of rearrangement numbers for two $\hat{\mu}_1$ and $\hat{\mu}_2$ follows.

\subsection{Our Results from the Text}

We note that this only depends on how the $\bbZ_2$ automorphism acts on the underlying lattice, so duality defects that are lifted from the same underlying lattice action on $L_0$ are in the same rearrangement class.

Our data is as follows:
\begin{enumerate}
    \item[$\bbZ_2$ Case] The fixed lattice is $L_0 = D_8$ and $L_0^*/L_0 \cong \bbZ_2\times \bbZ_2$. There is no rearrangement data in this case.
    \item[$\bbZ_3$ Case] The fixed lattice is $L_0 = A_2 \times E_6$ and $L_0^*/L_0 \cong \bbZ_3 \times \bbZ_3$. There are 7 duality defects which come from the Dynkin diagram automorphism on $A_2$ or on $E_6$, but not both. We find that these are in different rearrangement classes, and thus have rearrangement number 2 relative to each other.
    \item[$\bbZ_4$ Case] The fixed lattice is given in \eqref{eq:Z4fixedLattice}, and there are 6 duality defects coming from lifts of 2 of the 3 non-trivial $\bbZ_2$ lattice automorphisms, namely $\nu_2^{(4)}$ and $\nu_3^{(4)}$ in Equation \eqref{eq:NonTrivialZ4s}. Since $\abs{\bbZ_4^\times} = 2$ there are only two possibly rearrangement classes. We find that they are in different rearrangement classes, having rearrangement number 3 relative to each other.
    \item[$\bbZ_5^A$ Case] The fixed lattice is given by the $A_4\times A_4$ root lattice inside $E_8$, and there are 4 duality defects coming from lifts of 2 of the 5 conjugacy classes of outer automorphisms. These correspond to the $\{\sigma_1,\sigma_2\}$ class ``flipping'' one $A_4$, and the $\{\tau,\sigma_1\sigma_2\tau\}$ class ``exchanging'' $A_4$'s. We find that these are in different rearrangement classes, with the ``exchange'' conjugacy class having rearrangement number $2$ relative to the ``flip'' conjugacy class, of course this means the flip-class has rearrangement number $2^{-1} = 3$ relative to the exchange-class.
    \item[$\bbZ_5^B$ Case] The fixed lattice is described at the start of Section \ref{sec:Z5B}, and there are 7 duality defects coming from lifts of 2 lattice automorphisms $\nu_2^{(5)}$ and $\nu_3^{(5)}$ given in Equation \eqref{eq:NonTrivialZ5s}. We find that they are in different rearrangement classes, with the $\nu_3^{(5)}$ defects having rearrangement number $2$ relative to $\nu_2^{(5)}$ defects.
\end{enumerate}

\bibliographystyle{JHEP}
\bibliography{DualityDefects} 

\end{document}